\documentclass[accepted]{uai2026} %
\usepackage[american]{babel}
\usepackage[round]{natbib} %
\usepackage{mathtools} %
\usepackage{booktabs} %
\usepackage{tikz} %

\usepackage{lscape}

\usepackage{comment}
\usepackage{siunitx}
\usepackage{microtype}
\usepackage{graphicx}
\usepackage{subcaption}
\usepackage{booktabs} %

\usepackage{algorithm}
\usepackage{algorithmic}
\usepackage{forloop}
\usepackage{multirow}

\usepackage{enumitem}

\usepackage{wrapfig}
\usepackage{tikz}
\usepackage{cancel}
\usepackage{subcaption}

\usepackage[utf8]{inputenc} %
\usepackage[T1]{fontenc}    %
\usepackage{hyperref}       %
\usepackage{url}            %
\usepackage{booktabs}       %
\usepackage{amsfonts}       %
\usepackage{amsmath}
\usepackage{amsthm}
\usepackage{physics}
\usepackage{multicol}
\usepackage{bbm}
\usepackage{xcolor}
\usepackage{nicefrac}       %
\usepackage{microtype}      %
\usepackage{lipsum}		%
\usepackage{graphicx}
\usepackage{doi}

\usepackage{apxproof}

\usepackage[outline]{contour}
\contourlength{1pt}

\usepackage{amsmath}
\usepackage{amssymb}
\usepackage{mathtools}
\usepackage{amsthm}

\newcommand{\cs}{\footnotesize}
\theoremstyle{plain}
\newtheorem{theorem}{Theorem}[section]

\newtheorem{lemma}[theorem]{Lemma}

\theoremstyle{definition}

\theoremstyle{remark}

\usepackage{hyperref}
\usepackage[capitalize]{cleveref}
\crefname{algorithm}{Alg.}{Algs.}
\crefname{section}{Sec.}{Secs.}
\crefname{theorem}{Thm.}{Thms.}
\crefname{lemma}{Lem.}{Lems.}

\newenvironment{delayedproof}[1]
 {\begin{proof}[\raisedtarget{#1}Proof of \cref{#1}]}
 {\end{proof}}
\newcommand{\raisedtarget}[1]{%
  \raisebox{\fontcharht\font`P}[0pt][0pt]{\hypertarget{#1}{}}%
}
\newcommand{\proofref}[1]{\hyperlink{#1}{proof}}

\newcommand{\makebold}[1]{\text{\boldmath$\mathnormal{#1}$}} %

\DeclareMathOperator{\dens}{\mathnormal{\pi}}
\DeclareMathOperator{\udens}{\mathnormal{\varphi}}
\DeclareMathOperator{\N}{\mathcal{N}}

\DeclareMathOperator{\R}{\mathbb{R}}
\DeclareMathOperator{\polytope}{\mathcal{P}}
\DeclareMathOperator{\A}{\makebold{A}}

\DeclareMathOperator{\bb}{\makebold{b}}
\DeclareMathOperator{\x}{\makebold{x}}
\DeclareMathOperator{\y}{\makebold{y}}
\DeclareMathOperator{\z}{\makebold{z}}

\DeclareMathOperator{\xu}{\makebold{u}}
\DeclareMathOperator{\xv}{\makebold{v}}
\DeclareMathOperator{\eye}{\makebold{I}}
\DeclareMathOperator{\G}{\makebold{\mathcal{G}}}
\DeclareMathOperator{\cov}{\text{\boldmath$\mathrm{\Sigma}$}}

\DeclareMathOperator{\LL}{\makebold{L}}

\DeclareMathOperator{\rand}{\makebold{\eta}}
\DeclareMathOperator{\s}{\gamma}

\DeclareMathOperator{\HR}{\mathrm{HR}}
\DeclareMathOperator{\EHR}{\mathrm{EHR}}
\DeclareMathOperator{\LHR}{\mathrm{LHR}}
\DeclareMathOperator{\smHR}{\mathrm{smHR}}
\DeclareMathOperator{\smLHR}{\mathrm{smLHR}}
\DeclareMathOperator{\hessian}{\makebold{\mathcal{H}}}

\DeclareMathOperator{\supp}{\mathrm{supp}}

\DeclareMathOperator{\smax}{\s_{\mathrm{max}}}

\newcommand{\HRname}{Hit-\&-Run}
\newcommand{\HRshortname}{$\mathrm{HR}$}

\newcommand{\spd}{s.p.d.\xspace}

\newcommand{\mMALA}{mMALA}
\newcommand{\smMALA}{smMALA}

\newcommand{\Ciso}{$^{13}$C}
\newcommand{\blackjax}{\texttt{blackjax}}
\newcommand{\jax}{\texttt{jax}}
\newcommand{\xcflux}{\texttt{13CFLUX}}
\newcommand{\minESS}{$\min\!\text{ESS}$}

\definecolor{lgry}{HTML}{aaaaaa}
\definecolor{tabblu}{HTML}{1f77b4}
\definecolor{taborg}{HTML}{ff7f0e}
\definecolor{tabgre}{HTML}{2ca02c}
\definecolor{tabred}{HTML}{d62728}
\definecolor{tabppl}{HTML}{9467bd}
\definecolor{tabbrw}{HTML}{8c564b}
\definecolor{tabpnk}{HTML}{e377c2}
\definecolor{tabgry}{HTML}{7f7f7f}
\definecolor{tabolv}{HTML}{bcbd22}
\definecolor{tabcya}{HTML}{17becf}

\def\width{6pt}
\def\height{3pt}

\usepackage{xspace}

\newcommand{\cb}[1]{
    \raisebox{.3ex}{%
        \begin{tikzpicture}%
            \node[
                minimum width=\width,
                minimum height=\height,
                fill=#1,
                inner sep=1pt] {};
        \end{tikzpicture}%
    }%
    \xspace%
}

\newcommand{\cbb}[1]{
    \raisebox{.3ex}{%
        \begin{tikzpicture}%
            \node[
                minimum width=\width,
                minimum height=\height,
                fill=#1,
                opacity=0.5,
                inner sep=1pt] {};
        \end{tikzpicture}%
    }%
    \xspace%
}

\newcommand{\cbbb}[1]{
    \raisebox{.3ex}{%
        \begin{tikzpicture}[x=\width]%
            \node[
                minimum width=.5*\width,
                minimum height=\height,
                fill=#1,
                opacity=0.5,
                outer sep=0,
                inner sep=1pt] at (0, 0) {};
            \node[
                minimum width=.5*\width,
                minimum height=\height,
                fill=#1,
                outer sep=0,
                inner sep=1pt] at (.5, 0) {};
        \end{tikzpicture}%
    }%
    \xspace%
}

\newcommand{\HRm}{\mbox{HR\!\cb{tabppl}}}
\newcommand{\LHRm}{\mbox{LHR\!\cb{tabbrw}}}
\newcommand{\smHRdm}{\mbox{smHR$_{\delta}$\!\cb{tabpnk}}}
\newcommand{\smLHRdm}{\mbox{smLHR$_{\delta}$\!\cb{tabgry}}}
\newcommand{\RWMHm}{\mbox{RWMH\!\cb{tabblu}}}
\newcommand{\MALAm}{\mbox{MALA\!\cb{tabcya}}}
\newcommand{\smMALAdm}{\mbox{smMALA$_{\delta}$\!\cb{tabgre}}}
\newcommand{\Dikinm}{\mbox{Dikin\!\cb{taborg}}}
\newcommand{\MAPLAm}{\mbox{MAPLA\!\cb{tabred}}}

\newcommand{\smHRem}{\mbox{smHR$_{\varepsilon}$\!\cbb{tabpnk}}}
\newcommand{\smLHRem}{\mbox{smLHR$_{\varepsilon}$\!\cbb{tabgry}}}
\newcommand{\smMALAem}{\mbox{smMALA$_{\varepsilon}$\!\cbb{tabgre}}}

\newcommand{\smHRm}{\mbox{smHR\!\cbbb{tabpnk}}}
\newcommand{\smLHRm}{\mbox{smLHR\!\cbbb{tabgry}}}
\newcommand{\smMALAm}{\mbox{smMALA\!\cbbb{tabgre}}}

\newcommand{\p}[1]{\phantom{#1}} %

\title{Higher-Order Hit-\&-Run Samplers for Linearly Constrained Densities}

\author[1,2]{\href{mailto:<r.paul@fz-juelich.de>?Subject=Higher-Order Hit-&-Run}{Richard~D.~Paul}{}}
\author[1]{Anton Stratmann}
\author[1]{Johann~F.~Jadebeck}
\author[1]{Martin~Bey\ss}
\author[1]{Hanno~Scharr}
\author[2,3]{David~R\"ugamer}
\author[1]{Katharina~N\"oh}
\affil[1]{%
    Forschungszentrum J\"ulich, Germany
}
\affil[2]{%
    Department of Statistics, LMU Munich, Germany
}
\affil[3]{%
    Munich Center for Machine Learning (MCML), Germany.
}
  
\begin{document}
\maketitle

\begin{abstract}
    Markov chain Monte Carlo (MCMC) sampling of densities 
    restricted to linearly constrained domains is an important 
    task arising in Bayesian treatment of inverse problems in 
    the natural sciences.
    While efficient algorithms for uniform polytope sampling 
    exist, much less work has dealt with more complex 
    constrained densities. %
    In particular, gradient information as used in %
    unconstrained MCMC is not necessarily helpful in the 
    constrained case, where the gradient may push the 
    proposal's density out of the polytope.
    In this work, we propose novel constrained sampling
    algorithms, which combine strengths of higher-order 
    information, like the target log-density's gradients 
    and curvature, with the \HRname{}
    proposal, a simple
    mechanism which guarantees the generation of feasible
    proposals, fulfilling the linear constraints.
    Our extensive experiments demonstrate improved sampling 
    efficiency on complex constrained densities over various 
    constrained and unconstrained samplers.
\end{abstract}

\section{Introduction}

\emph{\HRname{}} \citep[\HRshortname{}, ][]{smith_efficient_1984,belisle_hit-and-run_1993_} 
is a well-established \emph{Markov chain Monte Carlo} (MCMC) method
for sampling from probability densities defined on constrained domains.
\HRshortname{} is applied in diverse domains like operations research~\citep{Tervonen2013},
cosmology~\citep{Lubini2014}, systems biology~\citep{Herrmann2019,theorell_be_2017_} and ecology~\citep{Gellner2023}, 
where modelling principles give rise to linearly constrained domains.
Moreover, \HRshortname{} has been used as a building block 
in advanced sampling algorithms~\citep{Theorell_2019,  nested_sampling_primer_2022, hit-and-run-slice-sampling_2024, yallup2026nestedslicesamplingvectorized}.
Its main strength for constrained sampling is the algorithm's %
guarantee to only generate feasible samples.
\citet{Lovasz2006} prove \HRshortname{} mixes rapidly for uniform sampling of bounded and convex linearly
constrained domains, referred to as convex polytopes.

\begin{figure}
\centering
    \resizebox{.95\columnwidth}{!}%
    {
    \begin{tikzpicture}
        \node at (1,0) {\includegraphics[width=\columnwidth]{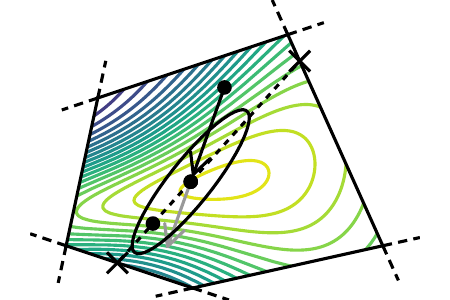}};
        \begin{scope}[cm={5.8,0,0,5.8,(-1.9,-1.75)}]
            \node at (-.12, -.2) {};
            \node at ( 1.1,  .9) {};

\node[right] at (0.5, 0.5) {\large \contour{white}{$\x$}};
\node[right] at (0.393389535266696, 0.20220348824443474) {\large \contour{white}{$\x+\hat{\varepsilon}\nabla\log\phi(\x)$}};
\node[left] at (0.16194118171518293, -0.05398039390506097) {\large \contour{white}{$\x+\s_a\!\xv$}};
\node[right] at (0.7375730946146677, 0.5831709008562938) {\large \contour{white}{$\x+\s_b\!\xv$}};
\node[left] at (0.2740908397720597, 0.0701550097039112) {\large \contour{white}{$\y$}};
        \end{scope}
    \end{tikzpicture}
    }
\caption{
    Our proposed smLHR sampler preconditions the direction using local curvature and
    clips the natural gradient step to prevent it from leaving the
    feasible region.
}
\label{fig:smlhr}
\end{figure}
However, real-world inverse problems, e.g.\ those arising from Bayesian inference~\citep{theorell_be_2017_,Slater2023}, 
lead to more complex non-uniform densities. 
These densities present additional sampling challenges such as strong nonlinear correlations between variables, 
or concentration of measure in small regions of the constrained domain.
For sampling densities defined on unconstrained domains, the incorporation of higher-order information of the density
has shown great success in dealing with these challenges. 
Prominent examples are gradient-guided samplers such as the
\emph{Metropolis-adjusted Langevin algorithm} \citep[MALA,][]{rossky_brownian_1978} 
and \emph{Hamiltonian Monte Carlo} \citep[HMC,][]{duane_hybrid_1987},
or curvature-based preconditioning as in \emph{Riemannian Manifold MALA} (\mMALA) and \emph{Riemannian HMC} \citep[RHMC, ][]{girolami_riemann_2011_}.
In addition, recent advances in automatic differentiation have greatly benefited the calculation of the higher-order information~\citep{baydin_automatic_2018}.

To leverage higher-order information 
for sampling densities defined on convexly constrained domains, alternatives to \HRshortname{} have been proposed. 
\citep{Kook2022} introduce a constrained version of RHMC for sampling log-concave densities,
but test it only for uniform problems.
Recently, \citep{srinivasan_high-accuracy_2024} proposed the first-order \emph{Metropolis-adjusted Preconditioned Langevin algorithm} (MAPLA).
Built upon \citet{girolami_riemann_2011_}, MAPLA preconditions
a Gaussian proposal density with a metric induced by the log-barrier function~\citep{Fiacco1990-nh} 
to incorporate information about the constraints instead of the curvature information of the domain-constrained density.
Although MAPLA aims to guide the sampling such that proposed samples respect the constraints,
its proposal density remains Gaussian, thus admitting some positive probability of 
proposing infeasible samples.
This raises the question whether higher-order algorithms that only propose feasible
samples improve the sampling efficiency for non-uniform constrained sampling problems.

\paragraph{Contribution}

In this work, we propose feasible higher-order samplers, which have positive proposal density just exactly
on the convex-constrained domain.
To this end, we combine the \HRshortname{} proposal mechanism with existing first- and second-order
samplers like MALA and the simplified \mMALA \citep{girolami_riemann_2011_}. 
As a result, we develop three new samplers \emph{Langevin Hit-\&-Run} (LHR),
\emph{simplified manifold Hit-\&-Run} (smHR), 
and \emph{simplified manifold Langevin Hit-\&-Run} (smLHR).
Our approach is mainly driven by the observation that a Gaussian random variable can 
be sampled in a \HRshortname{}-like fashion by decomposing it into direction and magnitude.
We provide theoretical analysis to prove convergence of our samplers 
to the desired target densities.
We numerically evaluate our algorithms on a systematically constructed benchmark consisting
of 2240 problems, which combine parametrizable polytopes and probability densities of 
varying geometry or ill condition, 
as well as on real-world examples from the field of Bayesian \Ciso{} metabolic flux analysis \citep[\Ciso-MFA,][]{theorell_be_2017_}.
Our results show that our combination of both the incorporation of higher-order information and the constraining 
of the proposal density by the \HRshortname{} proposal mechanism improve sampling efficiency.

\section{Preliminaries}
\label{sec:prelim}

We consider the problem of MCMC sampling from
smooth probability densities $\dens(\x)$ with support restricted to
non-empty, convex-constrained domains %
\begin{align}
    \polytope = \{ \x \in \R^d : \A\!\x \leq \bb \},
\end{align}
defined by $m$ linear constraints $\A\!\x\leq\bb$, i.e.\ $\A\in\R^{m\times d}, \bb\in\R^{m}$.
If $\polytope$ is bounded, we refer to it as a \emph{polytope}.

\subsection{Metropolis-Hastings Algorithm}
\label{sec:mh}

The \emph{Metropolis-Hastings} \citep[MH,][]{metropolis_equation_1953, hastings_monte_1970}  
algorithm is a widely applied method 
to draw samples from a target density $\dens: \R^d \to \R$, even
if access is only provided to some unnormalized density $\udens(\x) = Z\cdot\dens(\x)$,
where $Z$ is the normalization constant.
It works by iteratively drawing samples, called \emph{proposals}, 
from some proposal distribution $\y \sim q(\,\cdot\,|\x)$ conditional on $\x$, 
before \emph{``correcting''} them to match the 
desired target distribution.
To this end, the MH algorithm adds a filter that accepts 
or rejects transitions from $\x$ to $\y$ with acceptance probability %
\begin{align}
    \label{eq:mh-filter}
    \alpha(\y|\x) = \min\left\{
        1, 
        \frac{\udens(\y)}{\udens(\x)} \cdot \frac{\cancel{Z}}{\cancel{Z}} \cdot \frac{q(\x|\y)}{q(\y|\x)}
    \right\}.
\end{align}
A sufficient condition for the convergence of this Markov chain is that $q(\y|\x) > 0$ 
for any two states $\x, \y \in \mathrm{supp}(\pi)$ \citep{roberts_general_2004}.

The MH algorithm is also applicable to constrained target densities
$\dens : \polytope \to \R$, if we define an extension thereof with zero
density for any state outside $\polytope$, i.e.\ $\hat{\dens}(\x) := \mathbbm{1}_{\!\polytope}(\x)\dens(\x)$,
where $\mathbbm{1}_{\!\polytope}$ is the indicator function on $\polytope$.
Assuming that we start within $\supp(\dens) \subseteq \polytope$,
the MH filter prevents the sampler to move outside of $\polytope$
even for an infeasible proposal $\y \notin \polytope$, as $\alpha(\y|\x) = 0$ in that case.
Thus, in principle any unconstrained proposal distribution
can be used to produce samples from constrained densities $\pi$. %
However, using such a proposal distribution quickly 
grows ineffective when many infeasible $\y \notin \polytope$ are proposed
and, thus, immediately rejected.
Hence, the development and analysis of constrained sampling algorithms has been of
great interest 
\citep{smith_efficient_1984,belisle_hit-and-run_1993_,lovasz_hit-and-run_1999,kannan_random_2012,narayanan_randomized_2016,chen_vaidya_2017,chen_fast_2018,Mangoubi2022,Kook2022,gatmiry_sampling_2024,srinivasan_high-accuracy_2024}. %

\subsection{Hit-\&-Run Sampler}
\label{sec:hr}

A simple method, which guarantees the proposal to remain within the constrained domain,
is the HR proposal algorithm \citep{smith_efficient_1984, belisle_hit-and-run_1993_}:
\begin{algorithm}[H]
  \caption{Hit-\&-Run Proposal}
  \begin{algorithmic}[1]
    \STATE Draw a point $\xu$ uniformly at random from the 
        $d-1$ dimensional
        hypersphere, i.e.~$\|\!\xu\!\|_2 = 1$,
    \STATE scale the update with the step size $\varepsilon$, 
        i.e.~$\xv = \varepsilon\!\xu$,
    \STATE compute the step size 
        $\smax$ for which $\x + \smax\!\xv$ intersects with the 
        constraints,
    \STATE draw a step $\s \sim p_{[0, \smax]}$ from 
        the step distribution $p$ truncated to $[0, \s_{\smax}]$, and %
    \STATE compute the update $\y = \x + \s\!\xv$.
  \end{algorithmic}
  \label{alg:hr}
\end{algorithm}
For brevity, we denote drawing a sample using \cref{alg:hr}
as $\y \sim \HR(\x, \varepsilon^2).$
The truncated step distribution has density 
\begin{align}
    \label{eq:trunc-p}
    p_{[0, \smax]}(\s) = \begin{cases}
        \flatfrac{p(\s)}{F_p(\s)}, & \text{if } \s \in [0, \smax], \\
        0, & \text{else}
    \end{cases}
\end{align}
where $F_p$ is the cumulative density of the univariate, continuous step distribution $p$
and $\smax := \max_{\x+\s\!\xv \in \polytope} \s$ is the largest permissible step size
such that $\x+\s\!\xv \in \polytope$, which can be computed as
\begin{align}
    \label{eq:intersect}
    \smax = \mathrm{clp}(\x, \xv) := \min_{\makebold{a}_i^\top\!\xv > 0, i=1,\ldots,m}\frac{b_i - \makebold{a}_i^\top\!\x}{\makebold{a}_i^\top\!\xv}.
\end{align}
If $\makebold{a}_i^\top\!\xv \leq 0$ for all $i=1,\ldots,m$, then no 
constraints lie ahead in direction of $\xv$.
In this case, we set $\smax = \mathrm{clp}(\x, \xv) := \infty$ in which case $F_p(\infty) := 1$.

Given $\smax$, $p_{[0, \smax]}$ can be sampled using the 
inverse transform algorithm, if access
to the cumulative and inverse cumulative density functions of $p$ exists.
The HR algorithm has, thus, proposal density
\begin{align}
    \label{eq:hr-pdf}
    q_{_{\HR}}(\y | \x, \varepsilon) = 
    \frac{p_{[0, \smax]}\!\left(\|\y - \x\|_2 / \varepsilon\right)}
        {(\|\y - \x\|_2 / \varepsilon)^{d-1}} \cdot \frac{\Gamma(d/2)}{2\pi^{d/2}}.
\end{align}
Under the mild assumption of $\mathrm{supp}(p) = \R^+$, it follows that $q(\y|\x) > 0$
for any two $\x, \y \in \R^d$ and, thus, the MH chain with HR proposal distribution
converges.

In the limit of $\smax \to \infty$ and
for the particular choice of $p$ being the $\chi_d$ distribution, 
the HR algorithm draws
samples from $\N(\x, \varepsilon^2\!\eye)$.
This can be seen from decomposing a Gaussian sample $\rand \sim \N(0, \eye)$ 
into its directional and magnitudal components, $\xu = \rand/\|\!\rand\!\|$ and $\gamma = \|\!\rand\!\|$,
respectively.
Now we observe that the direction $\xu \sim \mathcal{U}_{S^{d-1}}$ follows a uniform distribution on the $d-1$ 
hypersphere $S^{d-1}$ and the magnitude $\gamma$ follows a $\chi_d$ distribution.

\subsection{Higher-Order Sampling Algorithms}
\label{sec:higher-order-sampling}

In many sampling problems, access to higher-order information of the
target density $\pi$ is available.
A~well-established method to incorporate first-order information %
in the sampling process is MALA \citep{rossky_brownian_1978}, 
which proposes new samples
\begin{align}
    \y \sim \N\!\left(
        \x + \frac{\varepsilon^2}{2}\grad\log\udens(\x), \
        \varepsilon^2\!\eye
    \right),
\end{align}
before applying the MH filter from \cref{eq:mh-filter}.
By applying some \emph{drift} $\varepsilon^2\grad\log\udens(\x)/2$ 
towards higher-density regions, MALA often achieves higher sampling 
efficiency than zeroth-order samplers.

\mMALA~\citep{girolami_riemann_2011_} 
extends the original MALA by incorporating second-order information.
To this end, \mMALA locally preconditions the MALA update using a user-specified
metric tensor $\G(\x)$, which improves sampling efficiency for densities admitting locally varying or ill-conditioned curvature.
However, mMALA requires the gradient of the metric tensor in order to be the correct
Euler-Maruyama discretization of the Riemannian Langevin diffusion on the manifold given by 
the metric tensor $\G(\x)$.
As this can be computationally expensive, \citet{girolami_riemann_2011_} also propose
a \emph{simplified mMALA} (\smMALA), which assumes $\pdv*{\G(\x)}{\x} = 0$.
The resulting proposal of the simplified algorithm is then
\begin{align}
    \y \sim \N\!\left(
        \x + \frac{\varepsilon^2}{2} \G(\x)^{-1}\grad\log\udens(\x), \
        \varepsilon^2 \G(\x)^{-1}
    \right).
\end{align}
Although this results in a \emph{``wrong''} discretization of the Riemannian Langevin diffusion,
this is a minor issue in the context of MH sampling, as the MH filter (cf.~\cref{eq:mh-filter})
will account for the error and guarantee convergence to the correct target distribution.
Indeed, as \citet{girolami_riemann_2011_} demonstrate empirically, assuming $\pdv*{\G(\x)}{\x} = 0$
does improve the sampling efficiency per unit of wall clock time, being less efficient yet 
much faster per sample.

\begin{figure*}

\vspace{1\baselineskip}

\centering
\begin{subfigure}{.325\linewidth}
    \caption{}
    \label{fig:hr}
    \vspace{-3\baselineskip}
    \resizebox{\columnwidth}{!}
    {
    \begin{tikzpicture}
        \node at (1,0) {\includegraphics[width=\columnwidth]{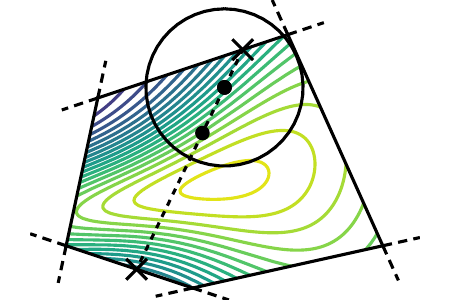}};
        \begin{scope}[cm={3.9,0,0,3.9,(-.95,-1.2)}]
            \node at (-.12, -.2) {};
            \node at ( 1.1,  .9) {};

            \node at (0, .75) {(a)};

\node[right] at (0.5, 0.5) {\contour{white}{$\x$}};
\node[right] at (0.2225910980650916, -0.07419703268836386) {\contour{white}{$\x+\s_a\!\xv$}};
\node[right] at (0.5575862827561264, 0.6191954275853755) {\contour{white}{$\x+\s_b\!\xv$}};
\node[right] at (0.43037383547361874, 0.35588372691118364) {\contour{white}{$\y$}};

        \end{scope}
    \end{tikzpicture}
    }
\end{subfigure}~\ \
\begin{subfigure}{.325\linewidth}
    \caption{}
    \label{fig:lhr}
    \vspace{-3\baselineskip}
    \resizebox{\columnwidth}{!}
    {
    \begin{tikzpicture}
        \node at (1,0) {\includegraphics[width=\columnwidth]{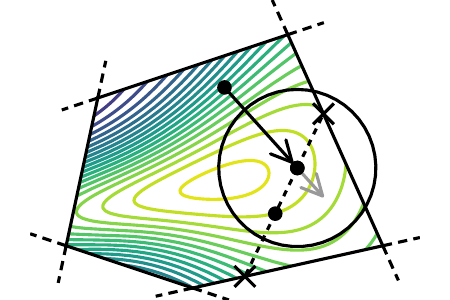}};
        \begin{scope}[cm={3.9,0,0,3.9,(-.95,-1.2)}]
            \node at (-.12, -.2) {};
            \node at ( 1.1,  .9) {};

            \node at (0, .75) {(b)};

\node[left] at (0.5, 0.5) {\contour{white}{$\x$}};
\node[left] at (0.7299665985552689, 0.24554813301228195) {\contour{white}{$\x+\hat{\varepsilon}\nabla\log\phi(\x)$}};
\node[left] at (0.5645900884369858, -0.09675775812511427) {\contour{white}{$\x+\s_a\!\xv$}};
\node[right] at (0.8125666607012855, 0.4165185317749211) {\contour{white}{$\x+\s_b\!\xv$}};
\node[left] at (0.6603404340288876, 0.10143185992346557) {\contour{white}{$\y$}};

        \end{scope}
    \end{tikzpicture}
    }
\end{subfigure}~\ \
\begin{subfigure}{.325\linewidth}
    \caption{}
    \label{fig:smhr}
    \vspace{-3\baselineskip}
    \resizebox{\columnwidth}{!}
    {
    \begin{tikzpicture}
        \node at (1,0) {\includegraphics[width=\columnwidth]{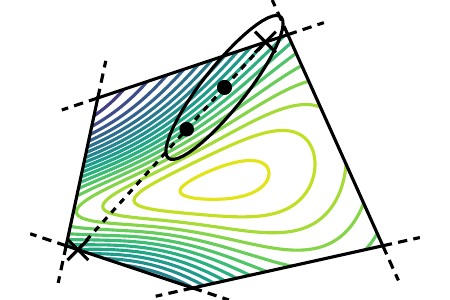}};
        \begin{scope}[cm={3.9,0,0,3.9,(-.95,-1.2)}]
            \node at (-.12, -.2) {};
            \node at ( 1.1,  .9) {};

            \node at (0, .75) {(c)};

\node[right] at (0.5, 0.5) {\contour{white}{$\x$}};
\node[right] at (0.0371032928969739, -0.012367764298991299) {\contour{white}{$\x+\s_a\!\xv$}};
\node[right] at (0.6292758897286683, 0.6430919632428894) {\contour{white}{$\x+\s_b\!\xv$}};
\node[right] at (0.38070130450536366, 0.36795152145947646) {\contour{white}{$\y$}};

        \end{scope}
    \end{tikzpicture}
    }
\end{subfigure}
\caption{
    Visualizations of (a) HR, (b) LHR and (c) smHR on a constrained 2-dimensional toy density. 
    A visulization of smLHR is given in \cref{fig:smlhr}.
}
\end{figure*}

\section{Higher-Order Hit-\&-Run}
\label{sec:higher-order-hit-and-run}

In this section, we present our approach to combine the 
zeroth-order HR algorithm with higher-order derivative information.
Before delving into technical details, we give a more intuitive 
motivation of our approach:
Our ideas are based upon the earlier mentioned observation,
that the zeroth-order HR proposal distribution virtually becomes a
diagonal Gaussian proposal distribution as the constraints \emph{``move further away''}, %
though sampled in a manner that decomposes the sample into 
direction and magnitude components. %
Now as MALA and smMALA are also just Gaussian proposal distributions, 
we devise HR-like proposal distributions, which replicate the relationship between 
HR and diagonal Gaussian proposal distributions, but for the MALA and smMALA.

\subsection{Langevin Hit-\&-Run}
\label{sec:lhr}

First, we consider the incorporation of gradient information into
the HR proposal distribution.
A key issue with incorporating gradients into constrained sampling
stems from the fact that the gradient, if too large, may quickly
push proposals out of $\polytope$, turning them infeasible.
Thus, we add a simple clipping mechanism 
similar to that already employed in HR,
i.e.\ we compute $\kappa := \mathrm{clp}(\x, \grad\log\udens(\x))$ and
then clip the drift term at $\kappa / 2$, i.e.\
we choose $\hat{\varepsilon} = \min\{\varepsilon^2/2, \kappa/2\}$
as the gradient step size.
Doing so, the drifted mean $\x+\hat{\varepsilon}\grad\log\udens(\x)$ of our proposal distribution can go at most
halfway up until it hits the closest constraint. %
As in \cref{eq:intersect}, we set $\kappa := \infty$ and, thus, $\hat{\varepsilon} = \varepsilon^2/2$
if the domain is unbounded in direction of the gradient $\grad\log\udens(\x)$.

Using this new clipped step size $\hat{\varepsilon}$, we
sample our proposal as
\begin{align}
    \label{eq:lhr}
    \y \sim \LHR(\x, \varepsilon^2) := \HR\!\left(
        \x + \hat{\varepsilon}\grad\log\udens(\x), \ 
        \varepsilon^2
    \right),
\end{align}
which we refer to as the \emph{Langevin Hit-\&-Run} ($\LHR$) proposal distribution.
We decide to clip at $\kappa / 2$ as getting exactly onto the constraints, 
i.e.\ choosing to clip at $\kappa$, can have unfavorable effects when computing the intersection $\smax$.
However, choosing exactly the halfway point $\kappa / 2$ is purely heuristically motivated.
A numerical ablation for the fraction of $\kappa$ is provided in \cref{fig:clipping} in the appendix.

As intended, in the limit of $\kappa, \smax \to \infty$, i.e.\ when 
the constraints are \emph{``far away''} from $\x$, and for
$p = \chi_d$, the LHR proposal distribution simplifies to MALA.

\subsection{Simplified Manifold Langevin Hit-\&-Run}
\label{sec:smlhr}

Next up, we incorporate curvature information.
Ignoring the drift term in smMALA for a moment, %
smMALA mainly preconditions the Gaussian proposal locally,
i.e.\ let $\LL\!\LL^{\top} = \G(\x)^{-1}$ be the lower triangular Cholesky
factorization of the inverse metric tensor, then ---
ignoring the drift term --- smMALA samples 
\begin{align}
    \y = \varepsilon\!\LL\!\rand = \varepsilon\|\!\rand\!\|_2\!\LL\!\frac{\rand}{\|\!\rand\!\|_2}
\end{align}
for $\rand \sim \N(0, \eye)$. 
Decomposing $\rand$ into its direction $\rand / \|\!\rand\!\|$ and magnitude $\|\!\rand\!\|$, %
the direction follows a uniform distribution on the $d-1$ hypersphere, 
i.e.~$\rand / \|\!\rand\!\|_2 \sim \mathcal{U}_{S^{d-1}}$, and the magnitude
follows a $\chi_d$-distribution, i.e.~$\|\!\rand\!\|_2 \sim \chi_d$.
Thus, we can sample $\y$ using a slight modification of the $\HR$ algorithm
by preconditioning the directional component 
$\rand / \|\!\rand\!\|$ by $\varepsilon\!\LL$, which --- intuitively speaking ---
results in an $\HR$ algorithm that draws the directional component from an ellipsoid
rather than a hypersphere.
We call this the \emph{elliptical Hit-\&-Run} (EHR) proposal algorithm,
which we state for some general positive definite covariance matrix $\cov$:
\begin{algorithm}[H]
  \caption{Elliptical Hit-\&-Run Proposal}
  \begin{algorithmic}[1]
    \STATE Draw a point $\xu$ uniformly at random from the 
        $d-1$ dimensional
        hypersphere, i.e. $\|\!\xu\!\|_2 = 1$,
    {\color{magenta} \STATE compute the update direction $\xv = \LL\!\xu$,}
    \STATE compute the step size 
        $\smax$ for which $\x + \smax\!\xv$ intersects with the 
        constraints,
    \STATE draw a step $\s \sim p_{[0, \smax]}$ from 
        the step distribution truncated to $[0, \smax]$, and
    \STATE compute the update $\y = \x + \s\!\xv$.
  \end{algorithmic}
  \label{alg:ehr}
\end{algorithm}
We denote \cref{{alg:ehr}} in short as $\y \sim \EHR(\x, \cov)$, 
for some positive definite matrix $\cov$ with lower triangular Cholesky
decomposition $\cov = \LL\!\LL^{\top}$ and assume $\A, \bb,p$ to be given.

Elliptical HR proposals are well-known in the polytope sampling literature 
\citep{lovasz_hit-and-run_2004,haraldsdottir_chrr_2017,theorell_polyround_2022,jadebeck_practical_2023, yallup2026nestedslicesamplingvectorized},
where they are prominently used to apply rounding transformations, which
make the polytope more isotropic, but can also be simply considered as
a preconditioner to the HR proposal.
Unlike in our intended use case, the rounding 
preconditioner is constant and its contribution 
to the proposal density cancels out due to symmetry in the MH filter (cf.\ \cref{eq:mh-filter}).
However, MH chains using the EHR proposal distribution with non-constant $\cov(\x)$
require the proposal density $q_{_{\EHR}}(\y|\x,\cov(\x))$
in order to compute the MH filter from \cref{eq:mh-filter}.
To this end, we derive its proposal density, given in the following Lem.~\ref{lemma:ehr-pdf}. 
The \proofref{lemma:ehr-pdf} is provided in \cref{sec:proofs} of the appendix.
\begin{lemma}
    \label{lemma:ehr-pdf}
    Assume $\x \in \polytope$ and $\cov$ \spd{}, then
    for $\y \sim \EHR(\x, \cov)$, $\y \in \polytope$ and has density
    \begin{align}
        \label{eq:ehr-pdf}
        q_{_{\EHR}}(\y|\x, \cov) \propto 
            \frac{p_{[0, \smax]}\!\left(\|\LL^{-1}(\y-\x)\|_2\right)}{\|\LL^{-1}(\y-\x)\|_2^{d-1} |\det\LL|}.
    \end{align}
\end{lemma}

With the $\EHR$ and $\LHR$ algorithms now at hand, 
we state the \emph{simplified manifold Langevin
Hit-\&-Run} ($\smLHR$) proposal as 
\begin{align}
    \y \sim 
        & \smLHR(\x, \varepsilon^2) := \\
        & \EHR\!\left(\x + \,\hat{\varepsilon}\G(\x)^{-1} \grad\log\udens(\x), \
            \varepsilon^2\!\G(\x)^{-1}\right), \nonumber
\end{align}
where $\hat{\varepsilon}$ is the clipped step size as in $\LHR$, 
but for the natural gradient $\G(\x)^{-1}\grad\log\udens(\x)$.

We further consider the
\emph{simplified manifold Hit-\&-Run} ($\smHR$),
which omits the gradient step of $\LHR$ and only retains the 
elliptical preconditioning component of $\smLHR$
\begin{align}
    \begin{array}{rl}
    \y \sim \smHR(\x, \varepsilon^2) := \EHR\!\left(\x, \
            \varepsilon^2\!\G(\x)^{-1}\right).
    \end{array}
\end{align}
This serves to empirically investigate the influence of 
the introduced changes to the HR algorithm.

\subsubsection{Choice of Metric Tensors}
\label{sec:metric}

The choice of an appropriate metric is not unique \citep{girolami_riemann_2011_}.
While the Hessian $\hessian(\x) := \grad^2\log\udens(\x)$ of the 
log-density is a natural choice if the latter happens
to be concave, it may not be positive definite in the general case,
disqualifying it as a metric tensor.
In Bayesian inference problems, the sum of the Fisher information matrix and
the Hessian of a strictly log-concave 
prior yield a meaningful metric \citep{girolami_riemann_2011_}. 
For Riemannian sampling of general densities, the \emph{SoftAbs} \citep{betancourt2013softabs}
and the \emph{Monge} metric \citep{hartmann2022lagrangian} have been proposed.

In this work, we additionally test two novel metrics based on
the Hessian of the log-density.
The first new metric that we introduce is the squared Hessian metric
\begin{align}
    \label{eq:metric-gc}
    \G_{\text{sq}}(\x) := \hessian(\x)^\top\!\hessian(\x) + \delta\!\eye,
\end{align}
which has the favorable property of being positive definite for 
$\delta > 0$,
as the squared Hessian term being a Gram matrix is positive semi-definite.
However, if the Hessian happens to be locally positive definite, 
the squared Hessian will apply an undesired second scaling. 
For this reason, we propose a heuristically motivated approach to %
adjust the scaling of the squared Hessian metric by taking the element-wise
square root of its non-zero entries, i.e.\
let $\makebold{M} := \hessian(\x)^\top\!\hessian(\x) \in \R^{d \times d}$
with entries $m_{ij}$ and $\mathrm{sgn}$ the sign function, then
the new scaled squared Hessian metric becomes
\begin{align}
    \label{eq:metric}
    \G_{\text{sc}}(\x) := \left(\mathrm{sgn}(m_{ij}) \cdot \sqrt{|m_{ij}|}\right)_{ij} + \delta\!\eye.
\end{align} %
As the new metric may not be positive definite, we 
use the diagonal of \cref{eq:metric} as a surrogate metric tensor, 
if the Cholesky decomposition of the former fails.

\paragraph{Controlling the Step Size}

In practice, the $\delta$ parameter is not known a priori and typically
needs to be tuned using e.g.\ grid search approaches, as a too large $\delta$
makes $\G(\x)$ more and more isotropic, effectively reverting the 
contribution of the Hessian.
If $\delta$ is too small, this leads to numerical issues when
$\G(\x)$ is not positive definite.
Hence, $\delta$ needs to be tuned for each individual problem, similarly to the step size.
However, noting that $\delta$ also affects the step size and 
in order to maintain a single tunable hyperparameter, we propose
a novel kind of step size parametrization, where we keep $\varepsilon=1$ fixed
and instead only tune $\delta$.
In order to keep this approach of controlling the step size comparable
to the \emph{``classical''} one, we reparametrize $\delta = \lambda^{-2}$ and
consider $\lambda$ the tunable hyperparameter.
To distinguish the two approaches, we refer to the algorithms with fixed $\delta$
as \smMALA$_{\varepsilon}$, smHR$_{\varepsilon}$, and smLHR$_{\varepsilon}$,
and those that fix $\varepsilon=1$ as \smMALA$_{\delta}$, smHR$_{\delta}$ and smLHR$_{\delta}$.
Note that we can apply the $\delta$-parametrization to general 
metrics $\G(\x)$ using $\hat{\G}(\x) := \G(\x) + \delta\!\eye$ as metric instead.
We evaluate both approaches experimentally 
in \cref{sec:experiments}.

\begin{figure*}
\begin{subfigure}{.3\linewidth}
    \centering
    \caption{}
    \label{fig:polytopes}
    \vspace{-2\baselineskip}
    \begin{tikzpicture}[x=75px,y=75px]
        \node at (0, 0) {\includegraphics[height=75px]{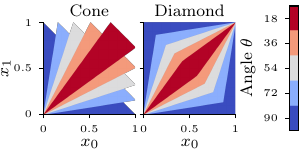}};
        \node at (-.97, .4) {\small (a)};
    \end{tikzpicture}
\end{subfigure}~\quad~
\begin{subfigure}{.6\linewidth}
    \centering
    \caption{}
    \label{fig:densities} %
    \vspace{-2\baselineskip}
    \begin{tikzpicture}[x=75px,y=75px]
        \node at (0, 0) {\includegraphics[height=75px]{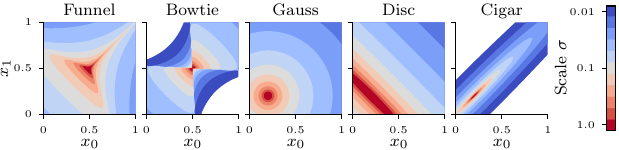}};
        \node at (-2.03, .4) {\small (b)};
        
        \node at (0*.688-.275,-.138) {\footnotesize .};
        \node at (1*.688-.275,-.138) {\footnotesize .};
        \node at (2*.688-.275,-.138) {\footnotesize .};
        
        \node at (0*.688-.2,-.138) {\scriptsize $\mu$};
        \node at (1*.688-.2,-.138) {\scriptsize $\mu$};
        \node at (2*.688-.2,-.138) {\scriptsize $\mu$};
    \end{tikzpicture}
\end{subfigure}
\caption{
    Two-dimensional, exemplary visualizations of the (a) polytopes and (b) densities under consideration,
    as well as the effect of the respective scale $\sigma$ and angle $\theta$
    parameters.
    The mode of the Gaussian densities is controlled by $\mu$ along the $(1, \ldots, 1)$ axis.
}
\label{fig:toy-problems}
\end{figure*}

\subsection{Convergence Analysis}
\label{sec:theory}

We state the convergence of our algorithms under the following mild assumption in \cref{thm:convergence}:
\begin{enumerate}[leftmargin=1cm]
    \item[$\mathit{(A1)}$] \textit{The state space $\polytope \subseteq \R^d$ is convex,} %
    \item[$\mathit{(A2)}$] \textit{$p$ is positive and continuous on $\R^{+}$, and}  %
    \item[$\mathit{(A3)}$] \textit{$\dens$ is finite and twice continuously differentiable on $\polytope$.} %
\end{enumerate}
\begin{theorem}
    \label{thm:convergence}
    Let $P_{\LHR}, P_{\smHR}, P_{\smLHR}$ be the MH kernels 
    (cf.~\cref{eq:mh-kernel}) with proposal
    distributions $\LHR$, $\smHR$, $\smLHR$, respectively.
    Under assumptions $\mathit{A1\text{--}3}$, 
    for $\dens$-almost all $\x$
    \begin{align}
        \| P^n(\x, \, \cdot \, ) - \pi(\,\cdot\,) \|_{_{\mathrm{TV}}} \to 0, \quad \text{ as } n \to \infty,
    \end{align}
    and for $P$ being any of $P_{\LHR}, P_{\smHR}, P_{\smLHR}$.
\end{theorem}

A proof is provided in \Cref{sec:proofs} in the appendix, here we give a sketch of our argument.
Following \citet[][Theorem 3.7(i) \& Prop.\ 6.3]{nummelin_general_1984} 
and \citet[][Theorem 1]{tierney_markov_1994},
any $\dens$-irreducible, aperiodic Markov chain with invariant measure
$\dens$ converges in total variation for $\dens$-almost all $\x \in \polytope$.
However, since the MH chain has $\dens$ as its invariant measure by construction
\citep[][Prop.\ 1 \& 2]{roberts_general_2004}, it remains to show $\dens$-irreducibility
and aperiodicity.
Aperiodicity follows from the finiteness of $\dens$ and $\dens$-irreducibility
of the MH chain \citep[][Running Example]{roberts_general_2004}.
A sufficient condition for $\dens$-irreducibility is 
positivity and continuity of the proposal density on $\supp(\dens)=\polytope$ 
\citep[][Eq.~(7.5)]{robert_monte_2004}.
For \cref{eq:ehr-pdf}, 
both requirements follow from assumptions $\mathit{(A1\text{--}2)}$ and an additional
assumption of $\cov$ being \spd{}
Note that $\cov$ being \spd{}\ lies in the responsibility of the practitioner.
Since our clipped gradient step does not violate 
the assumption of $\x \in \polytope$ by construction, 
\cref{lemma:positivity} remains valid.

\section{Related Work}
\label{sec:related-work}

Most work on linearly constrained sampling focuses on uniform densities. 
\citet{smith_efficient_1984} introduced polytope-adapted HR. %
\citet{Berbee1987} combined HR with Gibbs sampling, 
also known as \emph{Coordinate \HRshortname{}}, which was shown to be fast and reliable 
for uniform sampling, both empirically \citep{emiris_efficient_2014} and 
theoretically \citep{Laddha2023}.
Moreover, preconditioning \citep{haraldsdottir_chrr_2017} and thinning \citep{jadebeck_practical_2023} to minimize cost per sample were shown to further improve performance. 
Alternatives to \HRshortname{} incorporate linear constraints using local preconditioning 
based on the Hessian of a barrier function. 
For example, \emph{Dikin} walk uses the log-barrier \citep{kannan_random_2012} and \emph{Vaidya} walk
uses volumetric log-barrier \citep{chen_vaidya_2017}.
Reflection-based approaches also exist \citep{GRYAZINA2014497}.

Apart from uniform sampling, research has mainly explored log-concave distributions. 
Existing uniform samplers have been transformed into non-uniform ones using the MH filter
and were used to sample Gaussian distributions for polytope volume estimation 
\citep{cousins_bypassing_2015, Emiris2018}.
Implementations for MH-based constrained sampling were provided by \citet{jadebeck_hops_2021} \&
\citet{paul_hopsy_2024}.
The first methods to incorporate first-order information into non-uniform constrained sampling were HMC methods, 
which use the gradient of their target distribution.
To address constraints, these methods modify the underlying Hamiltonian dynamics: 
\citet{Kook2022} and \citet{NEURIPS2023_6745cb98} integrated penalty terms to 
slow down particles in proximity to the boundary, while \citet{Chalkis2023} introduced a reflection mechanism on the boundaries.
These approaches were benchmarked numerically only for truncated Gaussian
or uniform target densities.
Based on \citet{girolami_riemann_2011_} and \citet{kook24b}, 
\citet{srinivasan_high-accuracy_2024} propose the 
MAPLA, a first-order sampler for general 
distributions on convex domains. 
MAPLA can be understood as a a \smMALA{} variant, which uses the log-barrier's Hessian as metric tensor, 
preventing second-order information of the log-density from being easily integrated.

\section{Experiments}
\label{sec:experiments}

We benchmark our algorithms experimentally on 
synthetic constrained densities with varying curvature
or ill-conditioned covariance, and %
on real-world
examples from Bayesian \Ciso{} metabolic flux analysis \citep{theorell_be_2017_}.
All algorithms were implemented in \blackjax~\citep{cabezas_blackjax_2024}.
We used \jax{}~\citep{jax2018github} to compute gradients and Hessians of the synthetic densities.
For the \Ciso-MFA problems, we used the simulation software \xcflux{} \citep{stratmann202513cflux}, 
which also provides gradient and Fisher information computations.
Our \blackjax{} fork is publicly available at 
\href{https://github.com/ripaul/blackjax}{\url{github.com/ripaul/blackjax}} %
and the code for our experiments and figures at 
\href{https://github.com/ripaul/manifold-hit-and-run}{\url{github.com/ripaul/manifold-hit-and-run}}.
All experiments are replicated using four different random seeds.

\subsection{Synthetic Densities}
\label{sec:synthetic}

\paragraph{Problem Setup}

For our constrained densities, we \emph{``placed''} unconstrained densities
in two different kinds of polytopes, which we refer to as \emph{cone} and \emph{diamond}.
Both polytopes are parametrized by an angle parameter $\theta \in (0^{\circ}, 90^{\circ}]$ 
controlling its narrowness.
In all cases, the constructed polytope is contained within the $[0, 1]^d$ box.
We visualize this in \cref{fig:polytopes}, more details 
including the generalization of our polytopes to higher dimensions are provided in 
\cref{sec:toy} in the appendix.

As densities we consider a slight modification of Neal's \emph{funnel} \citep{neal_slice_2003},
a \emph{bowtie} distribution and a total of six Gaussians with isotropic and anisotropic
covariance structures and two different locations of the mean, such that the mode is 
once contained within the polytope $(\mu = 0.5)$ and once located on the polytope's border $(\mu=0)$.
Bowtie and funnel are shifted such that~their mode is contained within the polytope.
A visualization is given in \cref{fig:densities}, more details are provided in 
\cref{sec:toy} in the appendix.
Moreover, we introduce a scaling parameter $\sigma$,
which allows to control how close the densities'
high-probability region is to the constraints.
That is, if the mode is not located at the polytope's border, 
then for small $\sigma$ the target distribution may be far 
enough from the constraints to consider the problem effectively unconstrained,
whereas for larger $\sigma$ the target distribution approaches a uniform
distribution on $\polytope$.

Finally, for our benchmark, we build our set of problems as 
the Cartesian product of our eight different densities at
log-uniformly chosen scales $\sigma=10^{-2}, 10^{-1.5}, \ldots, 10^{1}$
restricted to our two kinds of polytopes, parametrized with angles
$\theta=9^{\circ}, 19^{\circ}, 45^{\circ}, 90^{\circ}$ and for dimensions $d=2, 4, 8, 16, 32$,
resulting in 2240 different test problems.

\paragraph{Evaluation}

\begin{figure}
    \centering
    \begin{subfigure}{.32\linewidth}
    \caption{}
    \label{fig:ess_angle}
    \vspace{-3\baselineskip}
    \end{subfigure}
    \begin{subfigure}{.32\linewidth}
    \caption{}
    \label{fig:ess_scale}
    \vspace{-3\baselineskip}
    \end{subfigure}
    \vspace*{-.8mm}
    
    \begin{tikzpicture}[x=70px,y=80]
        \node[anchor=north] at (-.5, .63) {\footnotesize $\mu = 0$};
        \node[anchor=north] at ( .5, .63) {\footnotesize $\mu = 0.5$};
        \node[anchor=north] at (0, .77) {\footnotesize Cigar / Disc / Gauss};
        \node[anchor=north] at (1.5, .77) {\footnotesize Bowtie / Funnel};
        
        \node[anchor=east] at (0, 0) {\includegraphics[height=75px]{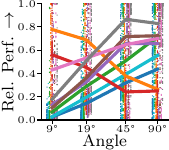}};
        \node[anchor=east] at (1, 0) {\includegraphics[height=75px]{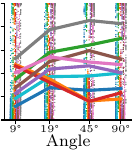}};
        \node[anchor=east] at (2, 0) {\includegraphics[height=75px]{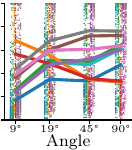}};
        
        \node[anchor=east] at (0,-1-.1) {\includegraphics[height=75px]{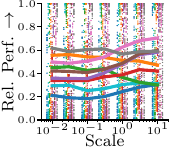}};
        \node[anchor=east] at (1,-1-.1) {\includegraphics[height=75px]{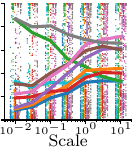}};
        \node[anchor=east] at (2,-1-.1) {\includegraphics[height=75px]{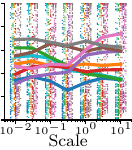}};

        \node at (-1.2, .55) {\small (a)};
        \node at (-1.2,-.45-.1) {\small (b)};
    \end{tikzpicture}
    \caption{
        Relative performance on the synthetic densities grouped
        by (a) the polytope angle and (b) the densities' scale parameter.
        Individual dots show relative performance of individual sampling runs,
        the solid lines show average relative performance across all problems.
        Algorithms shown are \protect\RWMHm{}, \protect\MALAm{}, \protect\smMALAdm{}, 
        \protect\Dikinm{}, \protect\MAPLAm{}, \protect\HRm{}, \protect\LHRm{}, 
        \protect\smHRdm{}, and \protect\smLHRdm{}.
    }
    \label{fig:rel-eff}
\end{figure}

Our main criteria for the evaluation of our algorithms are 
the L1 error of the joint marginal distribution across the first two dimensions
and
the minimal marginal effective sample size (\minESS{}),
similar to \citet{biron-lattes_automala_2024}.
For the former, we ran large-scale MCMC simulations using samplers
from the toolbox \texttt{hopsy} \citep{paul_hopsy_2024} to generate \emph{``ground-truth''} marginal distributions.
L1 error is then simply computed across the 2D histograms. 
We give more details in \cref{app:eval} in the appendix.

Since the achievable \minESS{} depends strongly on the problem at hand, we 
consider the relative performance of an algorithm for 
every problem as the ratio between its achieved average \minESS{} and the highest achieved
average \minESS{} across replicates on the same problem.%
We perform
a log-uniform grid search of the step size parameter for each algorithm under consideration.
and always report results for the step size achieving the smallest L1 error.
We test a Gaussian random walk MH sampler (RWMH), MALA, smMALA,
Dikin, MAPLA, \HRshortname{} and our introduced algorithms 
LHR$_{\varepsilon}$, smHR$_{\varepsilon}$ and smLHR$_{\varepsilon}$, 
as well as
LHR$_{\delta}$, smHR$_{\delta}$ and smLHR$_{\delta}$.
For all HR variants, we choose $p$ to be a half-normal distribution which moment-matches
a $\chi_d$ distribution, as 
using a $\chi_d$ distribution
is numerically expensive (cf.\ \cref{fig:chi}).
For the manifold samplers, we use the Hessian as metric tensor for all Gaussian targets
and test the SoftAbs \citep{betancourt2013softabs}, Monge \citep{hartmann2022lagrangian}, 
squared and scaled squared Hessian metrics 
from \cref{eq:metric-gc} \& \cref{eq:metric}, respectively, for the bowtie and funnel densities.
For each experiment we run four chains, 
drawing $20\,000\cdot d$ samples each and thinning them down to $20\,000$ to reduce the memory footprint.

\begin{figure}
    \centering
    \hspace*{-3mm}~
    \begin{subfigure}{.78\columnwidth}
        \caption{}
        \label{fig:metrics}
        \vspace{-2\baselineskip}
        \begin{tikzpicture}[x=80px, y=80px]
            \node at (0, 0) {\includegraphics[height=80px]{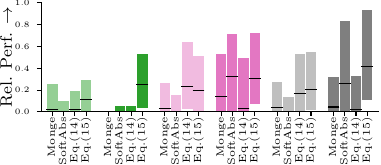}};
            \node at (-.8,.45) {\small (a)};
        \end{tikzpicture}
    \end{subfigure}~
    \begin{subfigure}{.195\columnwidth}
        \caption{}
        \label{fig:stepsizes}
        \vspace{-2\baselineskip}
        \begin{tikzpicture}[x=80px, y=80px]
            \node at (0, 0) {\includegraphics[height=80px]{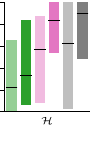}};
            \node at (-.14,.45) {\small (b)};
        \end{tikzpicture}
    \end{subfigure}
    \caption{
        Relative performance of \protect\smMALAem{}, \protect\smMALAdm{}, 
        \protect\smHRem{}, \protect\smHRdm{}, \protect\smLHRem{} \& \protect\smLHRdm{} using 
        (a) different choices of metric for the bowtie and funnel, and 
        (b) on the Gaussian targets.
        Colored bars show the 25\%
    }
    \label{fig:param}
\end{figure}

\begin{table*}[!ht]
    \centering
    \caption{
        Mean L1 error of joint marginals of first and second dimension.
        Best entries per column are marked in bold.
    }
    \resizebox{!}{3.7cm}{
    \begin{tikzpicture}
        \node at (0, 0) {\begin{tabular}{r|S[table-format=1.3]S[table-format=1.3]S[table-format=1.3]S[table-format=1.3]|S[table-format=1.3]S[table-format=1.3]S[table-format=1.3]S[table-format=1.3]|S[table-format=1.3]S[table-format=1.3]S[table-format=1.3]S[table-format=1.3]}
            \hline
            \multirow{2}{*}{Sampler}
                & \multicolumn{4}{c|}{Gauss / Cigar / Disc ($\mu = 0$)}
                & \multicolumn{4}{c|}{Gauss / Cigar / Disc ($\mu = 0.5$)}
                & \multicolumn{4}{c}{Bowtie / Funnel}       \\
                & {\footnotesize 9°} & {\footnotesize 19°} & {\footnotesize 45°} & {\footnotesize 90°} 
                & {\footnotesize 9°} & {\footnotesize 19°} & {\footnotesize 45°} & {\footnotesize 90°} 
                & {\footnotesize 9°} & {\footnotesize 19°} & {\footnotesize 45°} & {\footnotesize 90°} \\
            \hline
\RWMHm{}    & \cs 0.373          & \cs 0.094          & \cs 0.046          & \cs 0.047          &  \cs 0.045          & \cs 0.043          & \cs 0.049          & \cs 0.051          &  \cs 0.05           & \cs 0.045          & \cs 0.05           & \cs 0.052          \\
\MALAm{}    & \cs 0.374          & \cs 0.093          & \cs 0.045          & \cs 0.045          &  \cs 0.043          & \cs 0.042          & \cs 0.047          & \cs 0.049          &  \cs 0.046          & \cs 0.043          & \cs 0.048          & \cs 0.051          \\
\smMALAdm{} & \cs 0.366          & \cs 0.088          & \cs 0.043          & \cs 0.044          &  \cs 0.041          & \cs 0.04           & \cs 0.045          & \cs 0.046          &  \cs 0.043          & \cs 0.04           & \cs 0.045          & \cs 0.047          \\
\Dikinm{}   & \cs \textbf{0.028} & \cs \textbf{0.034} & \cs 0.045          & \cs 0.049          &  \cs 0.042          & \cs 0.045          & \cs 0.051          & \cs 0.053          &  \cs \textbf{0.037} & \cs 0.043          & \cs 0.05           & \cs 0.054          \\
\MAPLAm{}   & \cs 0.03           & \cs 0.036          & \cs 0.047          & \cs 0.051          &  \cs 0.05           & \cs 0.05           & \cs 0.052          & \cs 0.051          &  \cs 0.04           & \cs 0.045          & \cs 0.051          & \cs 0.055          \\
\HRm{}      & \cs 0.109          & \cs 0.039          & \cs \textbf{0.041} & \cs 0.044          &  \cs 0.039          & \cs 0.042          & \cs 0.048          & \cs 0.049          &  \cs 0.041          & \cs 0.043          & \cs 0.049          & \cs 0.051          \\
\LHRm{}     & \cs 0.109          & \cs 0.038          & \cs \textbf{0.041} & \cs \textbf{0.043} &  \cs 0.038          & \cs 0.04           & \cs 0.047          & \cs 0.048          &  \cs 0.04           & \cs 0.041          & \cs 0.048          & \cs 0.05           \\
\smHRdm{}   & \cs 0.078          & \cs 0.036          & \cs \textbf{0.041} & \cs 0.044          &  \cs \textbf{0.036} & \cs 0.04           & \cs 0.046          & \cs 0.047          &  \cs 0.038          & \cs 0.04           & \cs 0.046          & \cs 0.047          \\
\smLHRdm{}  & \cs 0.103          & \cs 0.037          & \cs \textbf{0.041} & \cs \textbf{0.043} &  \cs \textbf{0.036} & \cs \textbf{0.038} & \cs \textbf{0.044} & \cs \textbf{0.045} &  \cs \textbf{0.037} & \cs \textbf{0.038} & \cs \textbf{0.044} & \cs \textbf{0.046} \\
            \hline
        \end{tabular}};
    \end{tikzpicture}
    }
    \label{tab:err}
\end{table*}
\paragraph{Results}

In \cref{fig:ess_angle} and \cref{fig:ess_scale}, we show relative performance of the samplers
as functions of the polytopes' angle parameter
and the densities' scale parameter, respectively.
From \cref{fig:ess_angle}, we observe that \Dikinm{} and \MAPLAm{}
perform particularly well for the more narrow polytopes ($\theta=9^{\circ}, 19^{\circ}$) 
and for the Gaussian target's with $\mu = 0$.
For the other problems, we observe notable improvement in performance
achieved by our methods \LHRm{}, \smHRdm{}, and \smLHRdm{}.
This also becomes apparent in \cref{tab:err}, where \Dikinm{} and \MAPLAm{} 
achieve lowest sampling error for the Gaussian target's at $\mu=0$ and $\theta=9^{\circ}, 19^{\circ}$,
but are outperformed by our HR variants on all other problems.
Our observations are in line with the known issues of \HRshortname{} in ill-conditioned 
polytopes, namely getting stuck in sharp \emph{``corners''} \citep{lovasz_hit-and-run_2004}.
Further, as expected, the unconstrained samplers \RWMHm{}, \MALAm{}, \smMALAdm{} 
are consistently outperformed by their HR variants \HRm{}, \LHRm{} and \smLHRdm{}.
The decrease in performance of the unconstrained samplers is particularly noticeable 
in \cref{fig:ess_scale}, where the constraints become more influential as the scale
parameter increases.
\cref{fig:ess_scale} also reveals how gradient information becomes
less effective in the more uniform setting, as \HRm{} and \smHRdm{} become
competitive and even surpass their gradient-based counterparts \LHRm{} and \smLHRdm{}
as the scale parameter increases.

\begin{table}
    \centering
    \caption{
        Results on \Ciso-MFA problems
        Best entries per column are marked in bold.
    }
    \resizebox{!}{3.44cm}{
    \begin{tikzpicture}
        \node at (0, 0) {\begin{tabular}{r|S[table-format=1.4]S[table-format=3.4]S[table-format=1.4]|S[table-format=1.4]S[table-format=1.4]S[table-format=1.4]}
            \hline
            \multirow{2}{*}{Sampler}
                & \multicolumn{3}{c|}{Stationary \emph{Spiralus}}
                & \multicolumn{3}{c }{Instationary \emph{Spiralus}}
            \\
            & \multicolumn{1}{c}{\footnotesize \minESS} & \multicolumn{1}{c}{\footnotesize \minESS/s} & \multicolumn{1}{c|}{\footnotesize $\hat{R}$}
            & \multicolumn{1}{c}{\footnotesize \minESS} & \multicolumn{1}{c}{\footnotesize \minESS/s} & \multicolumn{1}{c }{\footnotesize $\hat{R}$}
            \\
            \hline
\RWMHm{}   &\cs 0.0337          & \cs 19.891            &  \cs 1.0052              & \cs 0.0             & \cs 0.0             & \cs 2.4501          \\
\MALAm{}   &\cs 0.0154          & \cs 4.1845            &  \cs 1.0061              & \cs 0.0             & \cs 0.0             & \cs {$\infty$}      \\
\smMALAdm{}&\cs \textbf{0.8667} & \cs \textbf{149.5027} &  \cs \textbf{1.0\p{000}} & \cs 0.0014          & \cs 0.0009          & \cs 1.0278          \\
\Dikinm{}  &\cs 0.0938          & \cs 55.5012           &  \cs 1.0011              & \cs 0.0             & \cs 0.0             & \cs 2.5191          \\
\MAPLAm{}  &\cs 0.0231          & \cs 6.0852            &  \cs 1.0029              & \cs 0.0             & \cs 0.0             & \cs {$\infty$}      \\
\HRm{}     &\cs 0.0321          & \cs 13.6633           &  \cs 1.0031              & \cs 0.0             & \cs 0.0             & \cs 2.4559          \\
\LHRm{}    &\cs 0.0177          & \cs 3.982             &  \cs 1.0075              & \cs 0.0             & \cs 0.0             & \cs {$\infty$}      \\
\smHRdm{}  &\cs 0.2663          & \cs 62.3391           &  \cs 1.0001              & \cs 0.0005          & \cs 0.0006          & \cs 1.0335          \\
\smLHRdm{} &\cs 0.7936          & \cs 126.1378          &  \cs 1.0001              & \cs \textbf{0.0091} & \cs \textbf{0.0061} & \cs \textbf{1.0015} \\
            \hline
        \end{tabular}};
    \end{tikzpicture}
    }
    \label{tab:mfa}
\end{table}

Next,
we analyze the effect of our proposed metrics and step size parametrization.
For the bowtie and funnel targets, 
the combination of 
our scaled squared Hessian metric
and the $\delta$-parametrization from \cref{sec:metric}
achieves highest average relative performance and lowest average
L1 error (cf.\ \cref{fig:metrics,fig:metrics-err}, respectively)
across \smMALAm{}, \smHRm{}, and \smLHRm{}.
For the Gaussian targets, we also find that our proposed $\delta$-parametrization
works better than the classical $\varepsilon$-parametrization (cf.\ \cref{fig:stepsizes,fig:stepsizes-err}).
We suspect that despite the Hessian metric being arguably an optimal proposal covariance in the
unconstrained Gaussian case, it fails to account for the polytope constraints.
However, it remains unclear to us as to why making the proposal covariance more isotropic
using the $\delta$-parametrization apparently improves sampling efficiency
in the constrained case.

\subsection{\Ciso{} Metabolic Flux Analysis}
\label{sec:13c-mfa}

We further test our algorithms on two test problems,
a stationary and non-stationary \emph{Spiralus}~\citep{wiechert2015primer}, 
from Bayesian \Ciso-MFA~\citep{theorell_be_2017_}, where the likelihood 
computation requires the 
solution of a set of ordinary differential equations,
which reduces to solving a set 
of algebraic equations for the stationary case.
Moreover, in the stationary case, so-called \emph{pool size} parameters
vanish, reducing the sampling problem's dimensionality.
Where required, we use the
Fisher information as
metric tensor.
As before, we run four chains, each drawing $10\,000\cdot d$ and $40\,000\cdot d$ samples
for stationary and non-stationary \emph{Spiralus}, respectively,
and apply thinning by a factor $d$.

Our results are presented in \cref{tab:mfa}.
For the simpler, 2-dimensional isotopically stationary \emph{Spiralus}, 
we observe that \smMALAdm{} yields the best performance,
closely followed by \smLHRdm{}.
Interestingly, \LHRm{} decreases the performance over \HRm, whereas \smHRdm{} improves it,
suggesting that curvature information is crucial for this problem.
For the 9-dimensional isotopically non-stationary \emph{Spiralus},
which entails non-identifiable as well as non-linearly correlated parameters (cf.~\cref{fig:ispiralus}),
curvature information seems crucial in order to successfully sample this problem,
as only \smMALAdm{}, \smHRdm{}, and \smLHRdm{} are able to converge within
the given budget.
Among these, \smLHRdm{} performs considerably better 
than its competitors.

\section{Conclusion}
\label{sec:conclusion}

We propose higher-order \HRshortname{} samplers, which work well for linearly
constrained, non-uniform densities, often outperforming the Dikin walk and recently
introduced MAPLA in our numerical experiments.
This suggests, that the combination of both the higher-order 
information as well as the strictly constrained \HRshortname{} proposal mechanism
are key to improve the sampling efficiency for such problems.
Notably, our methods seem more robust to \emph{``how far away''} the constraints are
from the target distribution's high-density region,
which a priori may not be clear for any given constrained sampling problem.
Similarly to the unconstrained case, we observe that incorporating higher-order information 
typically improves sampling efficiency. 
Despite \HRshortname{} being slower than its unconstrained counterparts (cf.\ \cref{fig:runtimes})
due to added complexity, our isotopically non-stationary \Ciso-MFA example underlines that higher sampling efficiency translates into faster sampling when
likelihood evaluations are expensive.

\paragraph{Limitations \& Future Work}

Given the good performance of MAPLA and the Dikin walk in the \emph{``narrow''} polytopes
and \HRshortname{}s weakness in such cases,
engineering second-order methods that combine rounding transformations or barrier-based
metrics with curvature estimates might strike a balance for even
more robust constrained samplers.

While our introduced scaled squared Hessian metric
and $\delta$-parametrization of the step size for 
the second-order methods show empirically improved 
performance, future work needs to investigate whether
this is more generally applicable.

%%\newpage

\section*{Impact Statement}

This paper presents work whose goal is to advance the field
of Machine Learning. There are many potential societal
consequences of our work, none which we feel must be
specifically highlighted here.

\section*{Acknowledgements}

RDP's \& AS' research was partly funded by the Helmholtz School for Data Science
in Life, Earth, and Energy (HDS-LEE). 
AS, JFJ, MB \& KN are thankful to Wolfgang Wiechert for excellent 
working conditions.
The authors gratefully acknowledge computing time on the 
supercomputer JURECA \citep{jureca} at Forschungszentrum J\"ulich 
under grant no. \texttt{hpcmfa} and \texttt{loki}.
We also thank anonymous reviewers at AISTATS for critical, 
yet helpful comments on an earlier version of this work.

\bibliography{references,more_references,even_more_references}

\newpage
\appendix
\onecolumn

\title{Higher-Order Hit-\&-Run Samplers for Linearly Constrained Densities\\(Supplementary Material)}
\maketitle

\section{Missing Proofs}
\label{sec:proofs}

In this section we present missing proofs for \cref{lemma:ehr-pdf} and \cref{thm:convergence}.
We begin with a proof of the $\EHR$ density (cf.~\cref{eq:ehr-pdf}):

\begin{delayedproof}{lemma:ehr-pdf}
    Let $\Delta:=\y-\x=\s\!\xv$, then step size and direction on the 
    ellipsoid are recovered as 
    $\|\LL^{-1}\!\Delta\|_2 = \gamma \|\LL^{-1}\!\LL\!\xu\|_2 = \gamma$,
    since $\|\xu\|_2=1$ and $\xv = \Delta/\!\s$.
    This transformation is a change of basis from Cartesian to spherical coordinates
    with differential volume $\dd\Delta = |\det\LL| \cdot \s^{d-1}\cdot\dd\!\s\dd\omega$,
    thus, the Jacobian determinant of the bijection $\Delta = \gamma\!\LL\!\xu$ is
    \begin{align}
        \left|\pdv{\Delta}{(\s,\xu)}\right| = |\det \LL|\cdot\s^{d-1}
    \end{align}
    and by the inverse function theorem, the Jacobian determinant of the inverse map 
    is the inverse of the Jacobian determinant of the forward map, i.e.
    \begin{align}
        \left|\pdv{(\s, \xu)}{\Delta}\right| 
            = \left|\pdv{\Delta}{(\s, \xu)}\right|^{-1} 
            = \frac{1}{|\det \LL|\cdot\s^{d-1}}.
    \end{align}
    Now by transform of random variables, we have that
    \begin{align}
        q_{_{\HR}}(\y|\x,\cov) 
            = p(\s, \xu | \x, \cov) \left|\pdv{(\s, \xu)}{\Delta}\right|
            = \frac{p(\s | \xu, \x, \cov)p(\xu) }{|\det \LL|\cdot\s^{d-1}}
            = \frac{p(\s | \xu, \x, \cov) }{|\det \LL|\cdot\s^{d-1}} 
                \frac{\Gamma(d/2)}{2\pi^{d/2}},
    \end{align}
    where the last equality follows from $p(\xu) = 1 / A_{d-1}$ being 
    the uniform distribution on the surface of the $d-1$ hypersphere, 
    which has area $A_{d-1} = 2\pi^{d/2}/\,\Gamma(d/2)$.
    Finally, $p(\s | \xu, \x, \cov)$ is just the truncated step 
    density $p_{[0, \smax]}(\s)$ as defined in \cref{alg:ehr}.
    Recalling our earlier definitions, we have that 
    $\gamma = \|\LL^{-1}(\y-\x)\|_2$,
    thus recovering \cref{eq:ehr-pdf}.
\end{delayedproof}

Next, we fill out the technical details of our sketched convergence argument
from \cref{sec:theory}.
Let
\begin{align}
    \label{eq:mh-kernel}
    P(\x, \y) := \alpha(\y | \x)\,q(\y | \x) + \left(1 - r(\x)\right)\delta_{\x}(\y),
    \quad \text{ with } \quad
    r(\x) := \int_{\polytope}\!\alpha(\z | \x)\, q(\z | \x) \,\dd\!\z
\end{align}
be the \emph{Metropolis kernel},
where $\delta_{\x}(\y)$
is the Dirac measure on $\x$, $\alpha(\y|\x)$ is the MH filter (cf.~\cref{eq:mh-filter})
and $q(\,\cdot\,|\x)$ is some proposal distribution conditionally on $\x$.
Moreover,
\begin{align}
    P^n(\x, \y) = \int_{\polytope} P^{n-1}(\x, \z)P(\z, \y)\,\dd\!\z
\end{align}
is the recursively defined
$n$-step kernel with $P^0(\x, \y) := \delta_{\x}(\y)$.

As discussed in \cref{sec:theory}, for a positive proposal distribution, 
i.e.~$q(\y|\x) > 0$ for any two $\x, \y \in \polytope$,
convergence of the $n$-step kernel to the desired target distribution
is given by standard arguments under assumption $\mathit{(A3)}$, 
as positivity then implies $\dens$-irreducibility and aperiodicity,
and thus convergence~\citep{roberts_general_2004}.
Therefore, it only remains to show that our algorithms' proposal density
is indeed positive.
To show this, we first prove that $q_{_{\EHR}}(\,\cdot\,|\x)$ is positive
for any $\x \in \polytope$ under assumptions $\mathit{A1\text{--}2}$ 
and additionally if $\cov$ is positive definite.

\begin{lemma}
\label{lemma:positivity}
Let $p$ be positive and continuous on $\R$, $\cov = \LL\!\LL^\top$ \spd{} Then,
$$q_{_{\EHR}}(\y|\x, \cov) = 
            \frac{p_{[0, \smax]}\!\left(\|\LL^{-1}(\y-\x)\|_2\right)}{\|\LL^{-1}(\y-\x)\|_2^{d-1} |\det\LL|}
            \cdot
            \frac{\Gamma(d/2)}{2\pi^{d/2}}$$
is positive and continuous for all  $\x, \y \in \polytope$ with $\x \neq \y$.
\end{lemma}
\begin{proof}
We prove positivity by contradiction.
Assume $\y \in \polytope$ with $q_{_{\EHR}}(\y|\x, \cov) = 0$.
By regularity of $\cov = \LL\!\LL^{\top}$, we always find $\xu = \LL^{-1}(\y-\x) / \s$ such that\
$\y = \s\!\LL\!\xu + \x$ for some $\s$.
By convexity, we have that $\s \in [0, \smax]$, however by positivity of $p$ on $\R$,
it follows that $p_{[\s_a, \s_b]}(\s) > 0$ and thus $q_{_{\EHR}}(\y|\x, \cov) > 0$,
which is a contradiction.
Finally, 
$q_{_{\EHR}}(\y|\x, \cov)$ is a chain of continuous functions and thereby continuous itself.
\end{proof}

Given \cref{lemma:positivity}, we now state the proof of \cref{thm:convergence}.

\begin{delayedproof}{thm:convergence}
    By \cref{lemma:positivity} and the argument provided in \cref{sec:theory},
    convergence of any of our three algorithms follows if we show that $\cov$
    is \spd{} and the clipped gradient step remains within $\polytope$.
    \begin{itemize}[leftmargin=2cm]
        \item[LHR:] Since $\LHR(\x, \varepsilon^2) = \EHR(\,\cdot\,|\x+\hat{\varepsilon}\grad\log\udens(\x), \varepsilon^2 I)$,
            we have that $\varepsilon^2 I$ is \spd{}\ for $\varepsilon > 0$ and it only remains to show that 
            $\x+\hat{\varepsilon}\grad\log\udens(\x) \in \polytope$, which however follows from convexity since
            $\x, \x+\kappa\grad\log\udens(\x) \in \polytope$ and $\hat{\varepsilon} \in (0, \kappa/2]$.
            Therefore, $P_{\LHR}$ converges in total variation distance to the desired target distribution.
        \item[smHR:] For $\smHR(\x, \varepsilon^2) = \EHR(\,\cdot\,|\x, \varepsilon^2\G(\x)^{-1})$
            it is only necessary to show that $\G(\x)$ is \spd{}, as the inverse of an \spd{}\ matrix is 
            itself \spd{}, which however follows from construction in \cref{eq:metric} and for $\delta$ large enough.
            Therefore, $P_{\smHR}$ converges in total variation distance to the desired target distribution.
        \item[smLHR:] Convergence of $P_{\smLHR}$ follows directly from that of
            $P_{\LHR}$ and $P_{\smHR}$. \\[-3.35\baselineskip]
    \end{itemize}
\end{delayedproof}

\section{Experiment Details}
\label{sec:toy}

\subsection{Ground Truth Sampling}
\label{app:eval}

We collect ``ground truth'' samples from our synthetic densities 
by running large MCMC simulations using the HR sampler from the toolbox \texttt{hopsy}.
For each problem, we run four parallel chains.
For the bowtie and funnel, we draw $1\,000\,000\cdot d \log_2 d$ samples using the HR sampler
with a Gaussian step distribution with step size being the scale parameter $\sigma$ of the respective
target.
For the Gaussian targets, we draw $100\,000\cdot d \log_2 d$ samples using the truncated Gaussian
sampler \citep{Li2015}.

\subsection{Synthetic Densities}

\paragraph{Densities}

As densities we consider Neal's funnel \citep{neal_slice_2003},
with $x_1 \sim \N(0, 9), \ x_i \sim \N(0, e^{x_1}), \quad i=2,\ldots,d,$
a \emph{bowtie} distribution
with $x_1 \sim \N(0, 1), \ x_i \sim \N\left(0, x_1^2/4 + 0.1\right), \quad i=2,\ldots,d,$
and six Gaussians, parametrized with an isotropic covariance matrix, disc-like covariance matrix
with marginal variances $\sigma_1^2 = 1/100$ and $\sigma_i = 1, i=1,\ldots,d$,
and a cigar-like covariance matrix with marginal variances $\sigma_1^2 = 1$ and $\sigma_i = 1/100, i=1,\ldots,d$.
Further, the mode of the Gaussian is either located at the origin 
or on the midpoint between origin and the intersection of the $(1, \ldots, 1)$ vector with the polytope borders,
which guarantees that the mode will be contained in the polytope $\polytope$.
Moreover, we apply a multivariate location-scale transformation
\begin{align}
    \udens(\x) := \udens_0\left(\makebold{Q}^{\top}\!\left(\frac{\x-m}{\sigma}\right)\!\right),
\end{align}
where $\makebold{Q}^{\top}$ rotates the $(1,0, \ldots, 0)$ vector to
$(1,\ldots,1)$. 
For the funnel and bowtie distributions, $m$ is determined as the halfway point 
between the origin and the intersection of the $(1,\ldots,1)$ vector with the polytope,
which is supposed to keep the ``characteristic'' regions of the density inside the polytope.
For the Gaussian distributions, we set $m=0$.

\paragraph{Polytopes}

The \emph{cone} is constructed as a simplex with tilted sides and a scaled constraint in
\cref{eq:cap}, which bounds the polytope within the $[0, 1]^d$ box:
\begin{align}
    \label{eq:cone}
    \cos(\theta) \cdot x_i - \sin(\theta) \cdot x_j &\leq 0
        & \text{ for any two different } i,j=1,\ldots,d, \\
    \sum_{i=1}^d x_i & \leq \frac{\cos(\theta)}{\sin(\theta)} + 1. \label{eq:cap}
\end{align}

The \emph{diamond} is constructed as a $[0, 1]^d$ whose sides are also tilted
and scaled such that\ the resulting parallelogram is contained in the $[0, 1]^d$ box:
\begin{align}
    \label{eq:diamond}
    \cos(\theta) \cdot x_i - \sin(\theta) \cdot x_j &\leq 0
        & \text{ for any two different } i,j=1,\ldots,d, \\
    -\cos(\theta) \cdot x_i + \sin(\theta) \cdot x_j &\leq \sin(\theta) - \cos(\theta)
        & \text{ for any two different } i,j=1,\ldots,d.
\end{align}

\section{Additional Figures \& Tables}

We provide additional figures.
In \cref{fig:metrics-err} \& \ref{fig:stepsizes-err}, we provide L1 error results based on the used metric.
In \cref{fig:spiralus} \& \ref{fig:ispiralus}, we provide pair plots of the two-dimensional marginal parameter posterior
distributions, sampled using our smLHR$_\delta$ sampler.
\cref{fig:ess1}-\ref{fig:acc2} further dissect the performance of the various 
tested samplers. 
For better accessibility, we provide these results also in tabular format.

\begin{figure}[!ht]
    \centering
    \begin{minipage}{.59\textwidth}
    \centering
    
    \begin{subfigure}{.9\textwidth}
    \hspace*{-3mm}~
    \begin{subfigure}{.78\textwidth}
        \caption{}
        \label{fig:metrics-err}
        \vspace{-2\baselineskip}
        \begin{tikzpicture}[x=80px, y=80px]
            \node at (0, 0) {\includegraphics[height=80px]{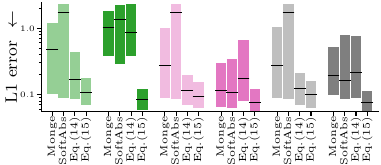}};
            \node at (-1.2,.45) {\small (a)};
        \end{tikzpicture}
    \end{subfigure}~
    \begin{subfigure}{.195\columnwidth}
        \caption{}
        \label{fig:stepsizes-err}
        \vspace{-2\baselineskip}
        \begin{tikzpicture}[x=80px, y=80px]
            \node at (0, 0) {\includegraphics[height=80px]{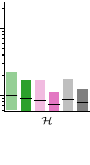}};
            \node at (-.4,.45) {\small (b)};
        \end{tikzpicture}
    \end{subfigure}
    \end{subfigure}
    \caption{
        L1 error of \protect\smMALAem{}, \protect\smMALAdm{}, 
        \protect\smHRem{}, \protect\smHRdm{}, \protect\smLHRem{} \& \protect\smLHRdm{} using 
        (a) different choices of metric for the bowtie and funnel, and 
        (b) on the Gaussian targets.
        Colored bars show the 25\%
    }
    \label{fig:param-err}
    \end{minipage}~\quad ~
    \begin{minipage}{.39\textwidth}
    \centering
    \includegraphics[height=80px]{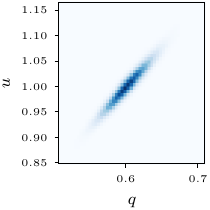}
    \caption{
        Pair plot of parameter posterior distribution of the isotopically stationary \emph{Spiralus} model,
        sampled using our smLHR$_{\delta}$ method.
        Parameter names follow \citet{wiechert2015primer}.
    }
    \label{fig:spiralus}
    \end{minipage}
\end{figure}

\begin{figure}[!ht]
    \centering
    \includegraphics[width=.8\textwidth]{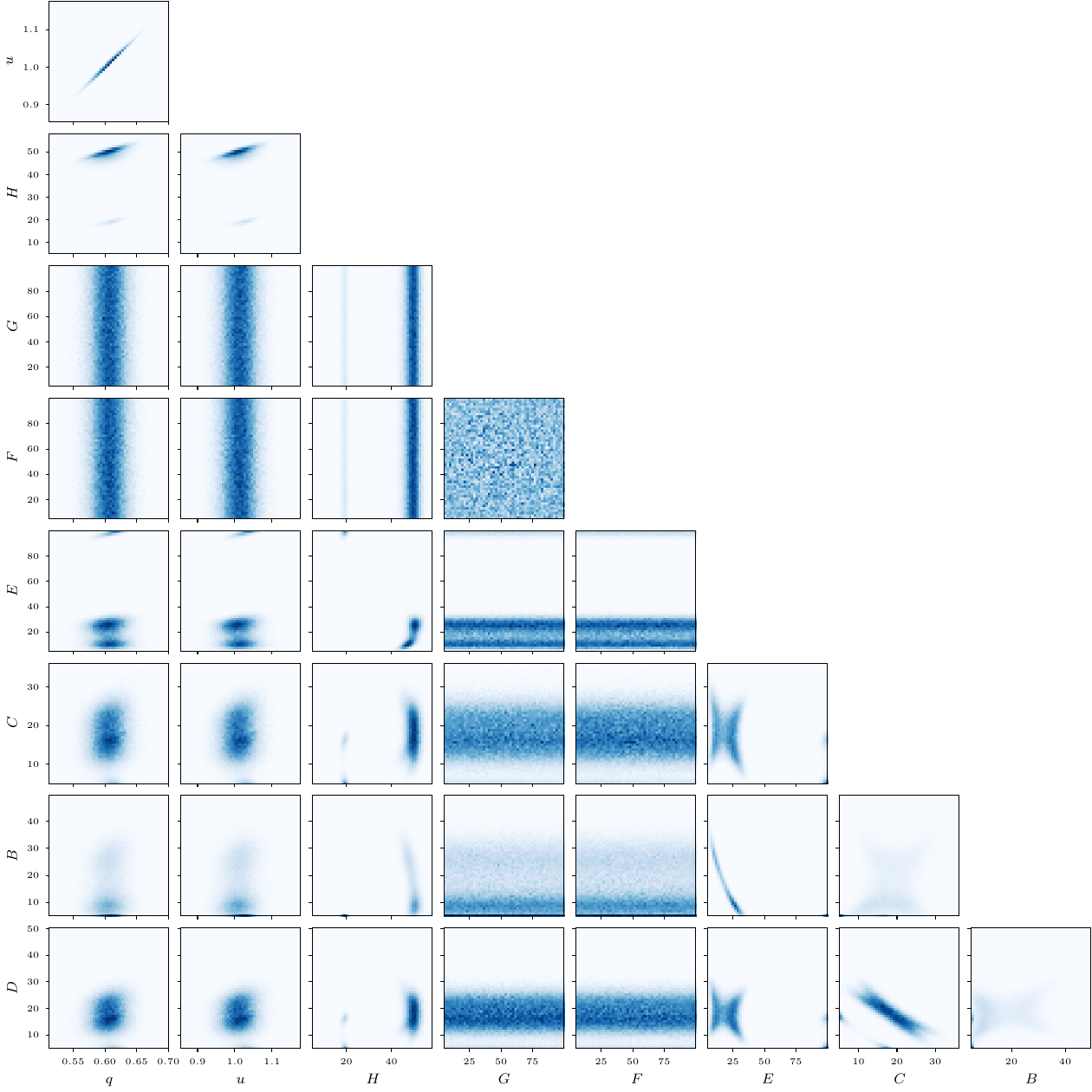}
    \caption{
        Pair plot of parameter posterior distribution of the isotopically non-stationary \emph{Spiralus} model,
        sampled using our smLHR$_{\delta}$ method.
        Parameter names follow \citet{wiechert2015primer}.
    }
    \label{fig:ispiralus}
\end{figure}

\begin{figure}[h!]
    \centering
    \resizebox{\columnwidth}{!}{
    \begin{tikzpicture}[x=70px,y=80]
        \node at (-0.5, .55) {\footnotesize Gauss $(\mu = 0)$};
        \node at ( 0.5, .55) {\footnotesize Disc $(\mu = 0)$};
        \node at ( 1.5, .55) {\footnotesize Cigar $(\mu = 0)$};
        \node at ( 2.5, .55) {\footnotesize Gauss $(\mu = 0.5)$};
        \node at ( 3.5, .55) {\footnotesize Disc $(\mu = 0.5)$};
        \node at ( 4.5, .55) {\footnotesize Cigar $(\mu = 0.5)$};
        \node at ( 5.5, .55) {\footnotesize Bowtie};
        \node at ( 6.5, .55) {\footnotesize Funnel};

        \node[anchor=east] at (0, 0) {\includegraphics[height=70px]{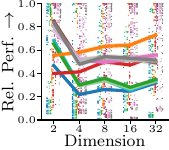}};
        \node[anchor=east] at (1, 0) {\includegraphics[height=70px]{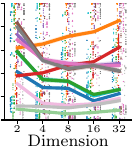}};
        \node[anchor=east] at (2, 0) {\includegraphics[height=70px]{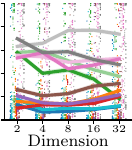}};
        \node[anchor=east] at (3, 0) {\includegraphics[height=70px]{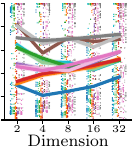}};
        \node[anchor=east] at (4, 0) {\includegraphics[height=70px]{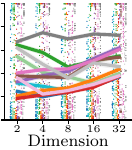}};
        \node[anchor=east] at (5, 0) {\includegraphics[height=70px]{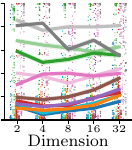}};
        \node[anchor=east] at (6, 0) {\includegraphics[height=70px]{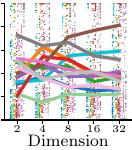}};
        \node[anchor=east] at (7, 0) {\includegraphics[height=70px]{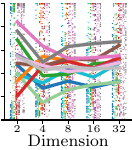}};
        
        \node[anchor=east] at (0,-1) {\includegraphics[height=70px]{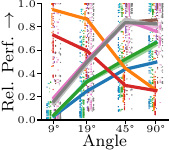}};
        \node[anchor=east] at (1,-1) {\includegraphics[height=70px]{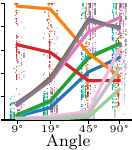}};
        \node[anchor=east] at (2,-1) {\includegraphics[height=70px]{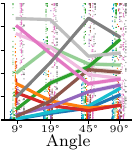}};
        \node[anchor=east] at (3,-1) {\includegraphics[height=70px]{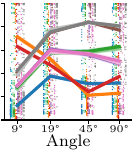}};
        \node[anchor=east] at (4,-1) {\includegraphics[height=70px]{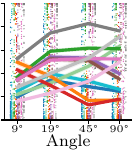}};
        \node[anchor=east] at (5,-1) {\includegraphics[height=70px]{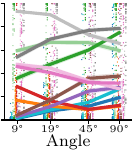}};
        \node[anchor=east] at (6,-1) {\includegraphics[height=70px]{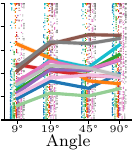}};
        \node[anchor=east] at (7,-1) {\includegraphics[height=70px]{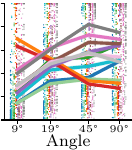}};
        
        \node[anchor=east] at (0,-2) {\includegraphics[height=70px]{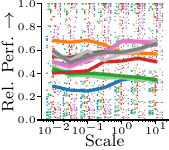}};
        \node[anchor=east] at (1,-2) {\includegraphics[height=70px]{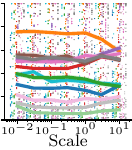}};
        \node[anchor=east] at (2,-2) {\includegraphics[height=70px]{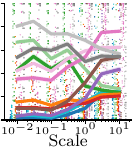}};
        \node[anchor=east] at (3,-2) {\includegraphics[height=70px]{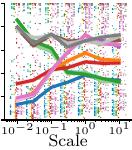}};
        \node[anchor=east] at (4,-2) {\includegraphics[height=70px]{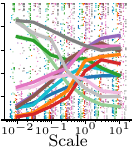}};
        \node[anchor=east] at (5,-2) {\includegraphics[height=70px]{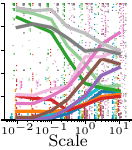}};
        \node[anchor=east] at (6,-2) {\includegraphics[height=70px]{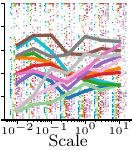}};
        \node[anchor=east] at (7,-2) {\includegraphics[height=70px]{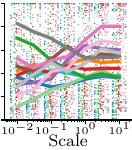}};
    \end{tikzpicture}
    }
    \caption{
        Detailed presentation of relative performance of the tested samplers on our benchmark problems as described in \cref{sec:experiments}.
        Individual dots show relative performance of individual sampling runs,
        the solid lines show average relative performance across all problems.
        Algorithms shown are \protect\RWMHm{}, \protect\MALAm{}, \protect\smMALAem{} \protect\smMALAdm{}, 
        \protect\Dikinm{}, \protect\MAPLAm{}, \protect\HRm{}, \protect\LHRm{}, 
        \protect\smHRem{}, \protect\smHRdm{},
        \protect\smLHRem{}, and \protect\smLHRdm{}.
    }
    \label{fig:ess1}
\end{figure}

\begin{figure}[h!]
    \centering
    \resizebox{\columnwidth}{!}{
    \begin{tikzpicture}[x=70px,y=80]
        \node at (-0.5, .55) {\footnotesize Gauss $(\mu = 0)$};
        \node at ( 0.5, .55) {\footnotesize Disc $(\mu = 0)$};
        \node at ( 1.5, .55) {\footnotesize Cigar $(\mu = 0)$};
        \node at ( 2.5, .55) {\footnotesize Gauss $(\mu = 0.5)$};
        \node at ( 3.5, .55) {\footnotesize Disc $(\mu = 0.5)$};
        \node at ( 4.5, .55) {\footnotesize Cigar $(\mu = 0.5)$};
        \node at ( 5.5, .55) {\footnotesize Bowtie};
        \node at ( 6.5, .55) {\footnotesize Funnel};

        \node[anchor=east] at (0, 0) {\includegraphics[height=70px]{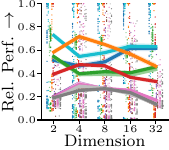}};
        \node[anchor=east] at (1, 0) {\includegraphics[height=70px]{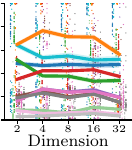}};
        \node[anchor=east] at (2, 0) {\includegraphics[height=70px]{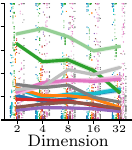}};
        \node[anchor=east] at (3, 0) {\includegraphics[height=70px]{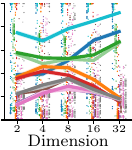}};
        \node[anchor=east] at (4, 0) {\includegraphics[height=70px]{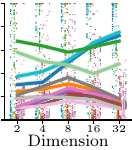}};
        \node[anchor=east] at (5, 0) {\includegraphics[height=70px]{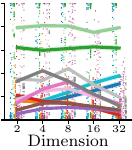}};
        \node[anchor=east] at (6, 0) {\includegraphics[height=70px]{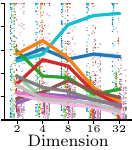}};
        \node[anchor=east] at (7, 0) {\includegraphics[height=70px]{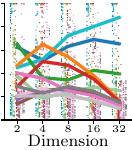}};
        
        \node[anchor=east] at (0,-1) {\includegraphics[height=70px]{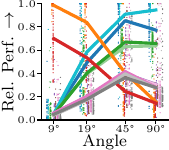}};
        \node[anchor=east] at (1,-1) {\includegraphics[height=70px]{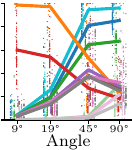}};
        \node[anchor=east] at (2,-1) {\includegraphics[height=70px]{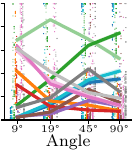}};
        \node[anchor=east] at (3,-1) {\includegraphics[height=70px]{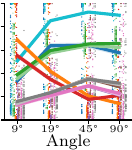}};
        \node[anchor=east] at (4,-1) {\includegraphics[height=70px]{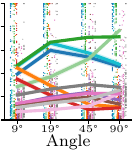}};
        \node[anchor=east] at (5,-1) {\includegraphics[height=70px]{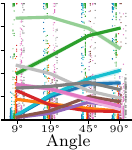}};
        \node[anchor=east] at (6,-1) {\includegraphics[height=70px]{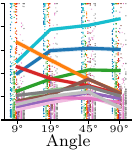}};
        \node[anchor=east] at (7,-1) {\includegraphics[height=70px]{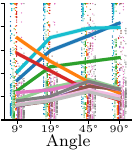}};
        
        \node[anchor=east] at (0,-2) {\includegraphics[height=70px]{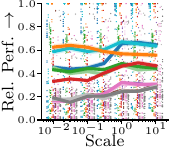}};
        \node[anchor=east] at (1,-2) {\includegraphics[height=70px]{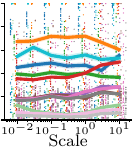}};
        \node[anchor=east] at (2,-2) {\includegraphics[height=70px]{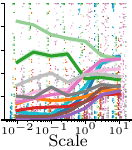}};
        \node[anchor=east] at (3,-2) {\includegraphics[height=70px]{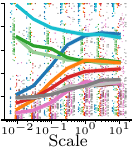}};
        \node[anchor=east] at (4,-2) {\includegraphics[height=70px]{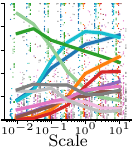}};
        \node[anchor=east] at (5,-2) {\includegraphics[height=70px]{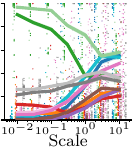}};
        \node[anchor=east] at (6,-2) {\includegraphics[height=70px]{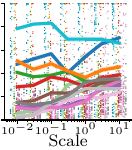}};
        \node[anchor=east] at (7,-2) {\includegraphics[height=70px]{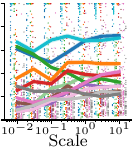}};
    \end{tikzpicture}
    }
    \caption{
        Detailed presentation of relative performance based on \minESS/s
        of the tested samplers on our benchmark problems as described in \cref{sec:experiments}.
        Individual dots show relative performance of individual sampling runs,
        the solid lines show average relative performance across all problems.
        Algorithms shown are \protect\RWMHm{}, \protect\MALAm{}, \protect\smMALAem{} \protect\smMALAdm{}, 
        \protect\Dikinm{}, \protect\MAPLAm{}, \protect\HRm{}, \protect\LHRm{}, 
        \protect\smHRem{}, \protect\smHRdm{},
        \protect\smLHRem{}, and \protect\smLHRdm{}.
    }
    \label{fig:esst1}
\end{figure}

\begin{figure}[h!]
    \centering
    \resizebox{\columnwidth}{!}{
    \begin{tikzpicture}[x=70px,y=80]
        \node at (-0.5, .55) {\footnotesize Gauss $(\mu = 0)$};
        \node at ( 0.5, .55) {\footnotesize Disc $(\mu = 0)$};
        \node at ( 1.5, .55) {\footnotesize Cigar $(\mu = 0)$};
        \node at ( 2.5, .55) {\footnotesize Gauss $(\mu = 0.5)$};
        \node at ( 3.5, .55) {\footnotesize Disc $(\mu = 0.5)$};
        \node at ( 4.5, .55) {\footnotesize Cigar $(\mu = 0.5)$};
        \node at ( 5.5, .55) {\footnotesize Bowtie};
        \node at ( 6.5, .55) {\footnotesize Funnel};

        \node[anchor=east] at (0, 0) {\includegraphics[height=70px]{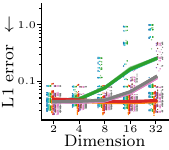}};
        \node[anchor=east] at (1, 0) {\includegraphics[height=70px]{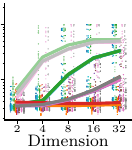}};
        \node[anchor=east] at (2, 0) {\includegraphics[height=70px]{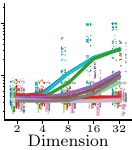}};
        \node[anchor=east] at (3, 0) {\includegraphics[height=70px]{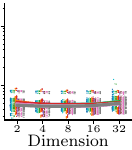}};
        \node[anchor=east] at (4, 0) {\includegraphics[height=70px]{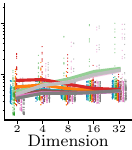}};
        \node[anchor=east] at (5, 0) {\includegraphics[height=70px]{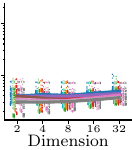}};
        \node[anchor=east] at (6, 0) {\includegraphics[height=70px]{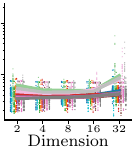}};
        \node[anchor=east] at (7, 0) {\includegraphics[height=70px]{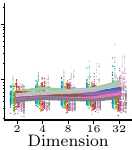}};
        
        \node[anchor=east] at (0,-1) {\includegraphics[height=70px]{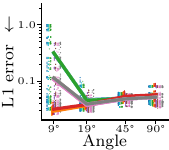}};
        \node[anchor=east] at (1,-1) {\includegraphics[height=70px]{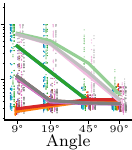}};
        \node[anchor=east] at (2,-1) {\includegraphics[height=70px]{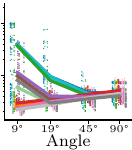}};
        \node[anchor=east] at (3,-1) {\includegraphics[height=70px]{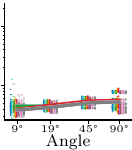}};
        \node[anchor=east] at (4,-1) {\includegraphics[height=70px]{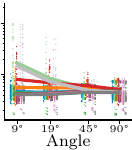}};
        \node[anchor=east] at (5,-1) {\includegraphics[height=70px]{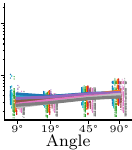}};
        \node[anchor=east] at (6,-1) {\includegraphics[height=70px]{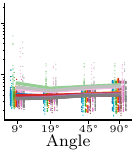}};
        \node[anchor=east] at (7,-1) {\includegraphics[height=70px]{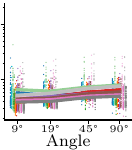}};
        
        \node[anchor=east] at (0,-2) {\includegraphics[height=70px]{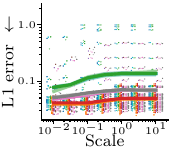}};
        \node[anchor=east] at (1,-2) {\includegraphics[height=70px]{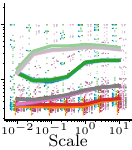}};
        \node[anchor=east] at (2,-2) {\includegraphics[height=70px]{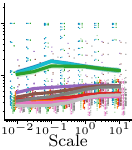}};
        \node[anchor=east] at (3,-2) {\includegraphics[height=70px]{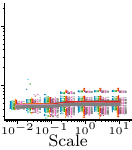}};
        \node[anchor=east] at (4,-2) {\includegraphics[height=70px]{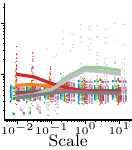}};
        \node[anchor=east] at (5,-2) {\includegraphics[height=70px]{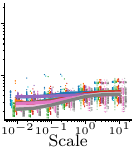}};
        \node[anchor=east] at (6,-2) {\includegraphics[height=70px]{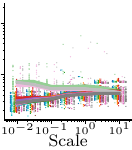}};
        \node[anchor=east] at (7,-2) {\includegraphics[height=70px]{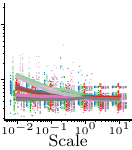}};
    \end{tikzpicture}
    }
    \caption{
        Detailed presentation of average L1 error
        of the tested samplers on our benchmark problems as described in \cref{sec:experiments}.
        Individual dots show relative performance of individual sampling runs,
        the solid lines show average relative performance across all problems.
        Algorithms shown are \protect\RWMHm{}, \protect\MALAm{}, \protect\smMALAem{} \protect\smMALAdm{}, 
        \protect\Dikinm{}, \protect\MAPLAm{}, \protect\HRm{}, \protect\LHRm{}, 
        \protect\smHRem{}, \protect\smHRdm{},
        \protect\smLHRem{}, and \protect\smLHRdm{}.
    }
    \label{fig:meanerr1}
\end{figure}

\begin{figure}[h!]
    \centering
    \resizebox{\columnwidth}{!}{
    \begin{tikzpicture}[x=70px,y=80]
        \node at (-0.5, .55) {\footnotesize Gauss $(\mu = 0)$};
        \node at ( 0.5, .55) {\footnotesize Disc $(\mu = 0)$};
        \node at ( 1.5, .55) {\footnotesize Cigar $(\mu = 0)$};
        \node at ( 2.5, .55) {\footnotesize Gauss $(\mu = 0.5)$};
        \node at ( 3.5, .55) {\footnotesize Disc $(\mu = 0.5)$};
        \node at ( 4.5, .55) {\footnotesize Cigar $(\mu = 0.5)$};
        \node at ( 5.5, .55) {\footnotesize Bowtie};
        \node at ( 6.5, .55) {\footnotesize Funnel};

        \node[anchor=east] at (0, 0) {\includegraphics[height=70px]{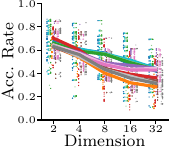}};
        \node[anchor=east] at (1, 0) {\includegraphics[height=70px]{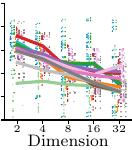}};
        \node[anchor=east] at (2, 0) {\includegraphics[height=70px]{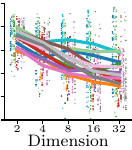}};
        \node[anchor=east] at (3, 0) {\includegraphics[height=70px]{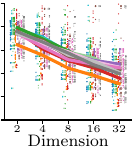}};
        \node[anchor=east] at (4, 0) {\includegraphics[height=70px]{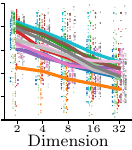}};
        \node[anchor=east] at (5, 0) {\includegraphics[height=70px]{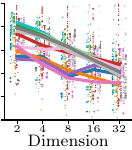}};
        \node[anchor=east] at (6, 0) {\includegraphics[height=70px]{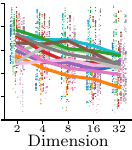}};
        \node[anchor=east] at (7, 0) {\includegraphics[height=70px]{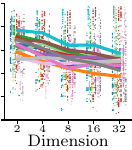}};
        
        \node[anchor=east] at (0,-1) {\includegraphics[height=70px]{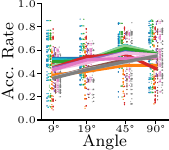}};
        \node[anchor=east] at (1,-1) {\includegraphics[height=70px]{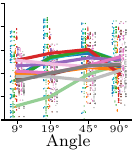}};
        \node[anchor=east] at (2,-1) {\includegraphics[height=70px]{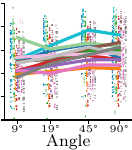}};
        \node[anchor=east] at (3,-1) {\includegraphics[height=70px]{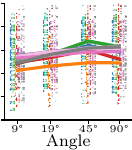}};
        \node[anchor=east] at (4,-1) {\includegraphics[height=70px]{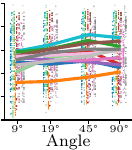}};
        \node[anchor=east] at (5,-1) {\includegraphics[height=70px]{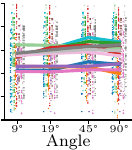}};
        \node[anchor=east] at (6,-1) {\includegraphics[height=70px]{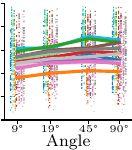}};
        \node[anchor=east] at (7,-1) {\includegraphics[height=70px]{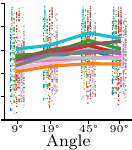}};
        
        \node[anchor=east] at (0,-2) {\includegraphics[height=70px]{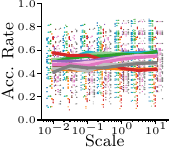}};
        \node[anchor=east] at (1,-2) {\includegraphics[height=70px]{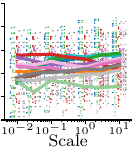}};
        \node[anchor=east] at (2,-2) {\includegraphics[height=70px]{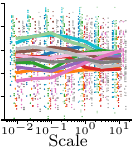}};
        \node[anchor=east] at (3,-2) {\includegraphics[height=70px]{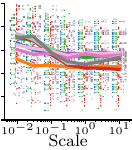}};
        \node[anchor=east] at (4,-2) {\includegraphics[height=70px]{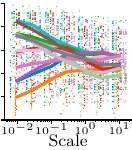}};
        \node[anchor=east] at (5,-2) {\includegraphics[height=70px]{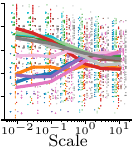}};
        \node[anchor=east] at (6,-2) {\includegraphics[height=70px]{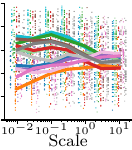}};
        \node[anchor=east] at (7,-2) {\includegraphics[height=70px]{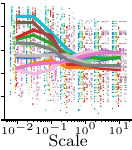}};
    \end{tikzpicture}
    }
    \caption{
        Detailed presentation of average acceptance rate
        of the tested samplers on our benchmark problems as described in \cref{sec:experiments}.
        Individual dots show relative performance of individual sampling runs,
        the solid lines show average relative performance across all problems.
        Algorithms shown are \protect\RWMHm{}, \protect\MALAm{}, \protect\smMALAem{} \protect\smMALAdm{}, 
        \protect\Dikinm{}, \protect\MAPLAm{}, \protect\HRm{}, \protect\LHRm{}, 
        \protect\smHRem{}, \protect\smHRdm{},
        \protect\smLHRem{}, and \protect\smLHRdm{}.
    }
    \label{fig:acc1}
\end{figure}

\begin{figure}[h!]
    \centering
    \resizebox{\columnwidth}{!}{
    \begin{tikzpicture}[x=70px,y=80]
        \node at (-0.5, .55) {\footnotesize Gauss $(\mu = 0)$};
        \node at ( 0.5, .55) {\footnotesize Disc $(\mu = 0)$};
        \node at ( 1.5, .55) {\footnotesize Cigar $(\mu = 0)$};
        \node at ( 2.5, .55) {\footnotesize Gauss $(\mu = 0.5)$};
        \node at ( 3.5, .55) {\footnotesize Disc $(\mu = 0.5)$};
        \node at ( 4.5, .55) {\footnotesize Cigar $(\mu = 0.5)$};
        \node at ( 5.5, .55) {\footnotesize Bowtie};
        \node at ( 6.5, .55) {\footnotesize Funnel};

        \node[anchor=east] at (0, 0) {\includegraphics[height=70px]{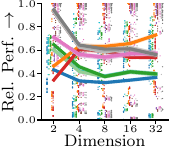}};
        \node[anchor=east] at (1, 0) {\includegraphics[height=70px]{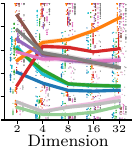}};
        \node[anchor=east] at (2, 0) {\includegraphics[height=70px]{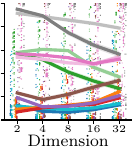}};
        \node[anchor=east] at (3, 0) {\includegraphics[height=70px]{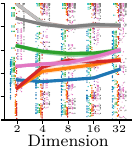}};
        \node[anchor=east] at (4, 0) {\includegraphics[height=70px]{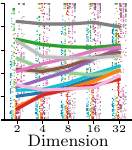}};
        \node[anchor=east] at (5, 0) {\includegraphics[height=70px]{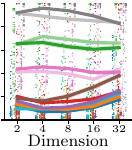}};
        \node[anchor=east] at (6, 0) {\includegraphics[height=70px]{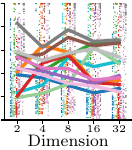}};
        \node[anchor=east] at (7, 0) {\includegraphics[height=70px]{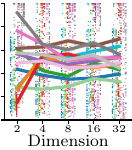}};
        
        \node[anchor=east] at (0,-1) {\includegraphics[height=70px]{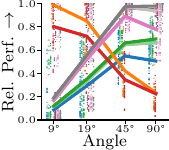}};
        \node[anchor=east] at (1,-1) {\includegraphics[height=70px]{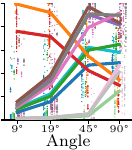}};
        \node[anchor=east] at (2,-1) {\includegraphics[height=70px]{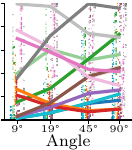}};
        \node[anchor=east] at (3,-1) {\includegraphics[height=70px]{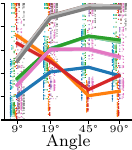}};
        \node[anchor=east] at (4,-1) {\includegraphics[height=70px]{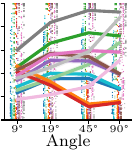}};
        \node[anchor=east] at (5,-1) {\includegraphics[height=70px]{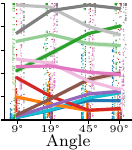}};
        \node[anchor=east] at (6,-1) {\includegraphics[height=70px]{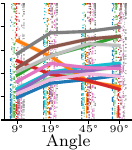}};
        \node[anchor=east] at (7,-1) {\includegraphics[height=70px]{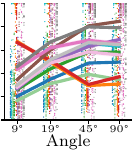}};
        
        \node[anchor=east] at (0,-2) {\includegraphics[height=70px]{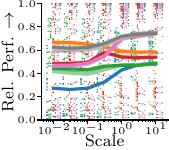}};
        \node[anchor=east] at (1,-2) {\includegraphics[height=70px]{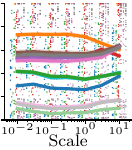}};
        \node[anchor=east] at (2,-2) {\includegraphics[height=70px]{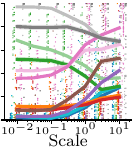}};
        \node[anchor=east] at (3,-2) {\includegraphics[height=70px]{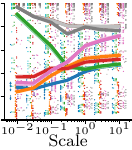}};
        \node[anchor=east] at (4,-2) {\includegraphics[height=70px]{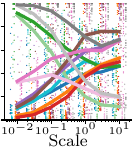}};
        \node[anchor=east] at (5,-2) {\includegraphics[height=70px]{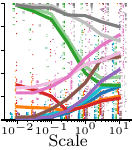}};
        \node[anchor=east] at (6,-2) {\includegraphics[height=70px]{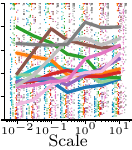}};
        \node[anchor=east] at (7,-2) {\includegraphics[height=70px]{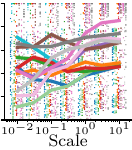}};
    \end{tikzpicture}
    }
    \caption{
        Detailed presentation of relative performance 
        of the tested samplers on our benchmark problems as described in \cref{sec:experiments},
        where optimal hyperparameters were chosen based on best achieved \minESS{}.
        Individual dots show relative performance of individual sampling runs,
        the solid lines show average relative performance across all problems.
        Algorithms shown are \protect\RWMHm{}, \protect\MALAm{}, \protect\smMALAem{} \protect\smMALAdm{}, 
        \protect\Dikinm{}, \protect\MAPLAm{}, \protect\HRm{}, \protect\LHRm{}, 
        \protect\smHRem{}, \protect\smHRdm{},
        \protect\smLHRem{}, and \protect\smLHRdm{}.
    }
    \label{fig:ess2}
\end{figure}

\begin{figure}[h!]
    \centering
    \resizebox{\columnwidth}{!}{
    \begin{tikzpicture}[x=70px,y=80]
        \node at (-0.5, .55) {\footnotesize Gauss $(\mu = 0)$};
        \node at ( 0.5, .55) {\footnotesize Disc $(\mu = 0)$};
        \node at ( 1.5, .55) {\footnotesize Cigar $(\mu = 0)$};
        \node at ( 2.5, .55) {\footnotesize Gauss $(\mu = 0.5)$};
        \node at ( 3.5, .55) {\footnotesize Disc $(\mu = 0.5)$};
        \node at ( 4.5, .55) {\footnotesize Cigar $(\mu = 0.5)$};
        \node at ( 5.5, .55) {\footnotesize Bowtie};
        \node at ( 6.5, .55) {\footnotesize Funnel};

        \node[anchor=east] at (0, 0) {\includegraphics[height=70px]{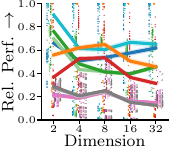}};
        \node[anchor=east] at (1, 0) {\includegraphics[height=70px]{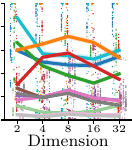}};
        \node[anchor=east] at (2, 0) {\includegraphics[height=70px]{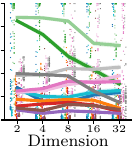}};
        \node[anchor=east] at (3, 0) {\includegraphics[height=70px]{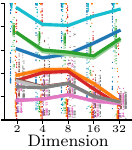}};
        \node[anchor=east] at (4, 0) {\includegraphics[height=70px]{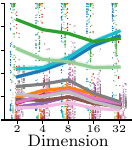}};
        \node[anchor=east] at (5, 0) {\includegraphics[height=70px]{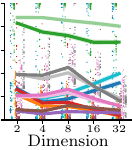}};
        \node[anchor=east] at (6, 0) {\includegraphics[height=70px]{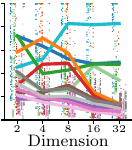}};
        \node[anchor=east] at (7, 0) {\includegraphics[height=70px]{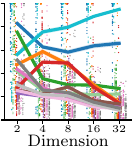}};
        
        \node[anchor=east] at (0,-1) {\includegraphics[height=70px]{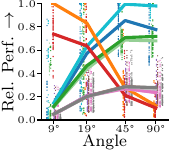}};
        \node[anchor=east] at (1,-1) {\includegraphics[height=70px]{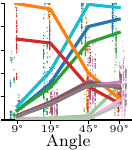}};
        \node[anchor=east] at (2,-1) {\includegraphics[height=70px]{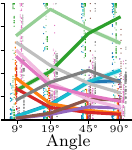}};
        \node[anchor=east] at (3,-1) {\includegraphics[height=70px]{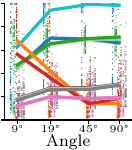}};
        \node[anchor=east] at (4,-1) {\includegraphics[height=70px]{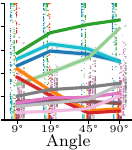}};
        \node[anchor=east] at (5,-1) {\includegraphics[height=70px]{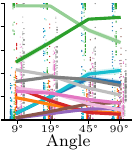}};
        \node[anchor=east] at (6,-1) {\includegraphics[height=70px]{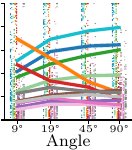}};
        \node[anchor=east] at (7,-1) {\includegraphics[height=70px]{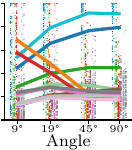}};
        
        \node[anchor=east] at (0,-2) {\includegraphics[height=70px]{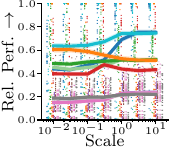}};
        \node[anchor=east] at (1,-2) {\includegraphics[height=70px]{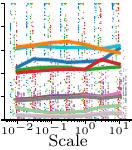}};
        \node[anchor=east] at (2,-2) {\includegraphics[height=70px]{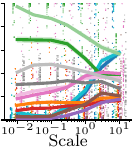}};
        \node[anchor=east] at (3,-2) {\includegraphics[height=70px]{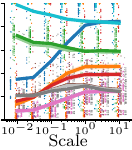}};
        \node[anchor=east] at (4,-2) {\includegraphics[height=70px]{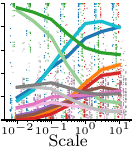}};
        \node[anchor=east] at (5,-2) {\includegraphics[height=70px]{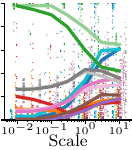}};
        \node[anchor=east] at (6,-2) {\includegraphics[height=70px]{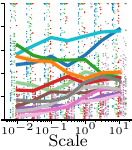}};
        \node[anchor=east] at (7,-2) {\includegraphics[height=70px]{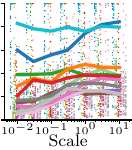}};
    \end{tikzpicture}
    }
    \caption{
        Detailed presentation of relative performance based on \minESS/s
        of the tested samplers on our benchmark problems as described in \cref{sec:experiments},
        where optimal hyperparameters were chosen based on best achieved \minESS{}.
        Individual dots show relative performance of individual sampling runs,
        the solid lines show average relative performance across all problems.
        Algorithms shown are \protect\RWMHm{}, \protect\MALAm{}, \protect\smMALAem{} \protect\smMALAdm{}, 
        \protect\Dikinm{}, \protect\MAPLAm{}, \protect\HRm{}, \protect\LHRm{}, 
        \protect\smHRem{}, \protect\smHRdm{},
        \protect\smLHRem{}, and \protect\smLHRdm{}.
    }
    \label{fig:esst2}
\end{figure}

\begin{figure}[h!]
    \centering
    \resizebox{\columnwidth}{!}{
    \begin{tikzpicture}[x=70px,y=80]
        \node at (-0.5, .55) {\footnotesize Gauss $(\mu = 0)$};
        \node at ( 0.5, .55) {\footnotesize Disc $(\mu = 0)$};
        \node at ( 1.5, .55) {\footnotesize Cigar $(\mu = 0)$};
        \node at ( 2.5, .55) {\footnotesize Gauss $(\mu = 0.5)$};
        \node at ( 3.5, .55) {\footnotesize Disc $(\mu = 0.5)$};
        \node at ( 4.5, .55) {\footnotesize Cigar $(\mu = 0.5)$};
        \node at ( 5.5, .55) {\footnotesize Bowtie};
        \node at ( 6.5, .55) {\footnotesize Funnel};

        \node[anchor=east] at (0, 0) {\includegraphics[height=70px]{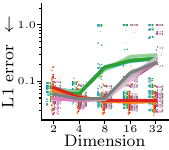}};
        \node[anchor=east] at (1, 0) {\includegraphics[height=70px]{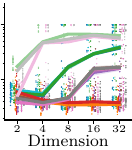}};
        \node[anchor=east] at (2, 0) {\includegraphics[height=70px]{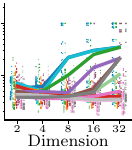}};
        \node[anchor=east] at (3, 0) {\includegraphics[height=70px]{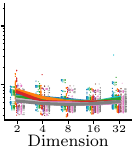}};
        \node[anchor=east] at (4, 0) {\includegraphics[height=70px]{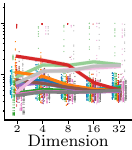}};
        \node[anchor=east] at (5, 0) {\includegraphics[height=70px]{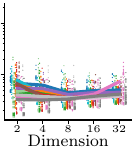}};
        \node[anchor=east] at (6, 0) {\includegraphics[height=70px]{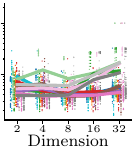}};
        \node[anchor=east] at (7, 0) {\includegraphics[height=70px]{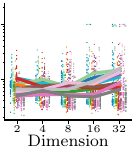}};
        
        \node[anchor=east] at (0,-1) {\includegraphics[height=70px]{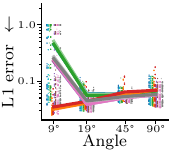}};
        \node[anchor=east] at (1,-1) {\includegraphics[height=70px]{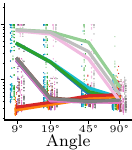}};
        \node[anchor=east] at (2,-1) {\includegraphics[height=70px]{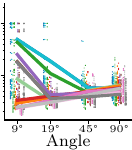}};
        \node[anchor=east] at (3,-1) {\includegraphics[height=70px]{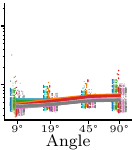}};
        \node[anchor=east] at (4,-1) {\includegraphics[height=70px]{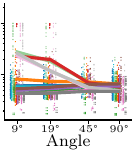}};
        \node[anchor=east] at (5,-1) {\includegraphics[height=70px]{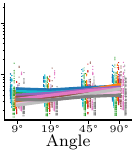}};
        \node[anchor=east] at (6,-1) {\includegraphics[height=70px]{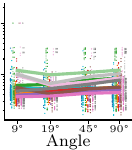}};
        \node[anchor=east] at (7,-1) {\includegraphics[height=70px]{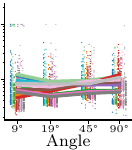}};
        
        \node[anchor=east] at (0,-2) {\includegraphics[height=70px]{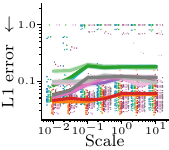}};
        \node[anchor=east] at (1,-2) {\includegraphics[height=70px]{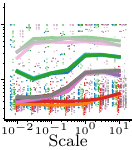}};
        \node[anchor=east] at (2,-2) {\includegraphics[height=70px]{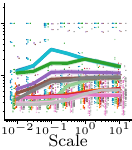}};
        \node[anchor=east] at (3,-2) {\includegraphics[height=70px]{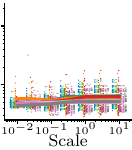}};
        \node[anchor=east] at (4,-2) {\includegraphics[height=70px]{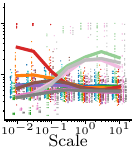}};
        \node[anchor=east] at (5,-2) {\includegraphics[height=70px]{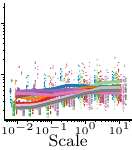}};
        \node[anchor=east] at (6,-2) {\includegraphics[height=70px]{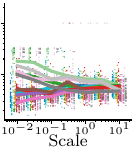}};
        \node[anchor=east] at (7,-2) {\includegraphics[height=70px]{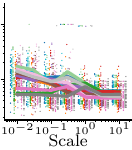}};
    \end{tikzpicture}
    }
    \caption{
        Detailed presentation of average L1 error
        of the tested samplers on our benchmark problems as described in \cref{sec:experiments},
        where optimal hyperparameters were chosen based on best achieved \minESS{}.
        Individual dots show relative performance of individual sampling runs,
        the solid lines show average relative performance across all problems.
        Algorithms shown are \protect\RWMHm{}, \protect\MALAm{}, \protect\smMALAem{} \protect\smMALAdm{}, 
        \protect\Dikinm{}, \protect\MAPLAm{}, \protect\HRm{}, \protect\LHRm{}, 
        \protect\smHRem{}, \protect\smHRdm{},
        \protect\smLHRem{}, and \protect\smLHRdm{}.
    }
    \label{fig:meanerr2}
\end{figure}

\begin{figure}[h!]
    \centering
    \resizebox{\columnwidth}{!}{
    \begin{tikzpicture}[x=70px,y=80]
        \node at (-0.5, .55) {\footnotesize Gauss $(\mu = 0)$};
        \node at ( 0.5, .55) {\footnotesize Disc $(\mu = 0)$};
        \node at ( 1.5, .55) {\footnotesize Cigar $(\mu = 0)$};
        \node at ( 2.5, .55) {\footnotesize Gauss $(\mu = 0.5)$};
        \node at ( 3.5, .55) {\footnotesize Disc $(\mu = 0.5)$};
        \node at ( 4.5, .55) {\footnotesize Cigar $(\mu = 0.5)$};
        \node at ( 5.5, .55) {\footnotesize Bowtie};
        \node at ( 6.5, .55) {\footnotesize Funnel};

        \node[anchor=east] at (0, 0) {\includegraphics[height=70px]{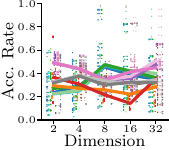}};
        \node[anchor=east] at (1, 0) {\includegraphics[height=70px]{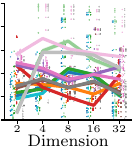}};
        \node[anchor=east] at (2, 0) {\includegraphics[height=70px]{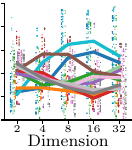}};
        \node[anchor=east] at (3, 0) {\includegraphics[height=70px]{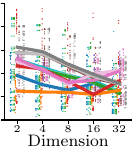}};
        \node[anchor=east] at (4, 0) {\includegraphics[height=70px]{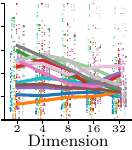}};
        \node[anchor=east] at (5, 0) {\includegraphics[height=70px]{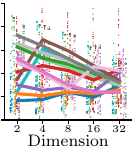}};
        \node[anchor=east] at (6, 0) {\includegraphics[height=70px]{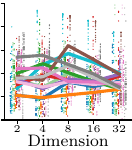}};
        \node[anchor=east] at (7, 0) {\includegraphics[height=70px]{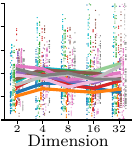}};
        
        \node[anchor=east] at (0,-1) {\includegraphics[height=70px]{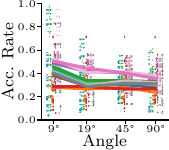}};
        \node[anchor=east] at (1,-1) {\includegraphics[height=70px]{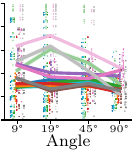}};
        \node[anchor=east] at (2,-1) {\includegraphics[height=70px]{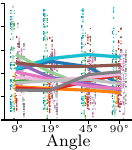}};
        \node[anchor=east] at (3,-1) {\includegraphics[height=70px]{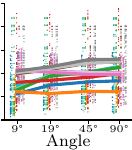}};
        \node[anchor=east] at (4,-1) {\includegraphics[height=70px]{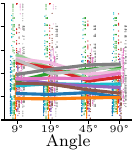}};
        \node[anchor=east] at (5,-1) {\includegraphics[height=70px]{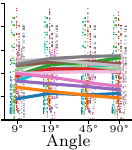}};
        \node[anchor=east] at (6,-1) {\includegraphics[height=70px]{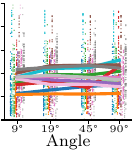}};
        \node[anchor=east] at (7,-1) {\includegraphics[height=70px]{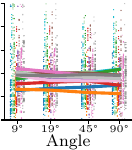}};
        
        \node[anchor=east] at (0,-2) {\includegraphics[height=70px]{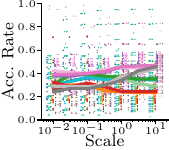}};
        \node[anchor=east] at (1,-2) {\includegraphics[height=70px]{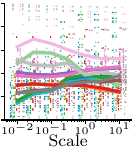}};
        \node[anchor=east] at (2,-2) {\includegraphics[height=70px]{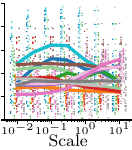}};
        \node[anchor=east] at (3,-2) {\includegraphics[height=70px]{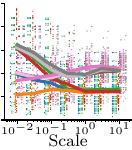}};
        \node[anchor=east] at (4,-2) {\includegraphics[height=70px]{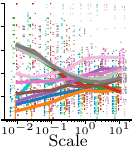}};
        \node[anchor=east] at (5,-2) {\includegraphics[height=70px]{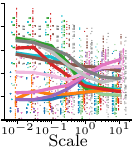}};
        \node[anchor=east] at (6,-2) {\includegraphics[height=70px]{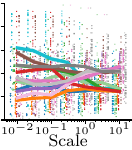}};
        \node[anchor=east] at (7,-2) {\includegraphics[height=70px]{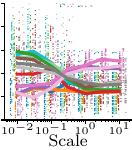}};
    \end{tikzpicture}
    }
    \caption{
        Detailed presentation of average acceptance rate
        of the tested samplers on our benchmark problems as described in \cref{sec:experiments},
        where optimal hyperparameters were chosen based on best achieved \minESS{}.
        Individual dots show relative performance of individual sampling runs,
        the solid lines show average relative performance across all problems.
        Algorithms shown are \protect\RWMHm{}, \protect\MALAm{}, \protect\smMALAem{}, \protect\smMALAdm{}, 
        \protect\Dikinm{}, \protect\MAPLAm{}, \protect\HRm{}, \protect\LHRm{}, 
        \protect\smHRem{}, \protect\smHRdm{},
        \protect\smLHRem{}, and \protect\smLHRdm{}.
    }
    \label{fig:acc2}
\end{figure}

\begin{figure}
    \centering
    \includegraphics[width=0.5\linewidth]{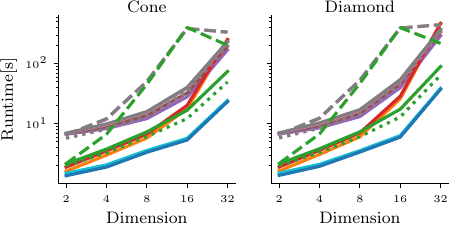}
    \caption{
        Average runtime of the tested samplers on the funnel and bowtie targets.
        Algorithms shown are \protect\RWMHm{}, \protect\MALAm{}, \protect\smMALAdm{}, 
        \protect\Dikinm{}, \protect\MAPLAm{}, \protect\HRm{}, \protect\LHRm{}, 
        \protect\smHRdm{},
        and \protect\smLHRdm{}.
        For the metric-based samplers, we show average runtimes using our metric from Eq.~(\ref{eq:metric}, solid),
        the SoftAbs metric (dashed) and the Monge metric (dotted).
    }
    \label{fig:runtimes}
\end{figure}

\begin{figure}
    \centering
    \includegraphics[width=0.5\linewidth]{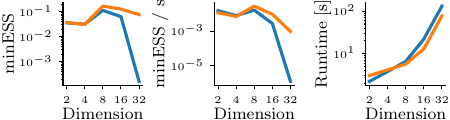}
    \caption{
        Comparison of using an exact\protect\cb{tabblu} vs. our moment-matched\protect\cb{taborg} $\chi_d$ distribution,
        evaluated using HR the sampler on the funnel distribution ($\sigma=0.1$) in a cone-shaped polytope ($\theta=19^{\circ}$).
        A fixed step size of $\varepsilon=0.1$ was used.
    }
    \label{fig:chi}
\end{figure}

\begin{figure}
    \centering
    \includegraphics[width=0.5\linewidth]{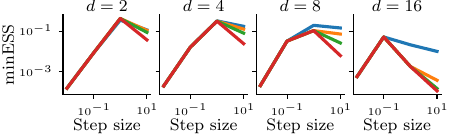}
    \caption{
        Sensitivity analysis of the gradient clipping factor from \cref{sec:lhr}.
        Tested fractions of $\kappa$ for clipping are 
        0.25\protect\cb{tabblu},
        0.5\protect\cb{taborg},
        0.75\protect\cb{tabgre}, and
        0.99\protect\cb{tabred}.
        Evaluated using the LHR sampler on the isotropic Gaussian with $\mu = 0$
        in a cone-shaped polytope ($\theta=90^{\circ}$) and across step sizes $\varepsilon=10^{-2}, 10^{-1}, 10^0, 10^1$.
    }
    \label{fig:clipping}
\end{figure}

\newpage

\begin{landscape}

\begin{table}[]
    \centering
    \caption{
        Average relative performance of the tested samplers across dimensions as presented as solid lines in \cref{fig:ess1}.
    }
    \resizebox{\linewidth}{!}{
    \begin{tikzpicture}
        \node at (0, 0) {%
            \begin{tabular}{r|
                S[table-format=1.4]S[table-format=1.4]S[table-format=1.4]S[table-format=1.4]S[table-format=1.4]|
                S[table-format=1.4]S[table-format=1.4]S[table-format=1.4]S[table-format=1.4]S[table-format=1.4]|
                S[table-format=1.4]S[table-format=1.4]S[table-format=1.4]S[table-format=1.4]S[table-format=1.4]|
                S[table-format=1.4]S[table-format=1.4]S[table-format=1.4]S[table-format=1.4]S[table-format=1.4]|
                S[table-format=1.4]S[table-format=1.4]S[table-format=1.4]S[table-format=1.4]S[table-format=1.4]|
                S[table-format=1.4]S[table-format=1.4]S[table-format=1.4]S[table-format=1.4]S[table-format=1.4]|
                S[table-format=1.4]S[table-format=1.4]S[table-format=1.4]S[table-format=1.4]S[table-format=1.4]|
                S[table-format=1.4]S[table-format=1.4]S[table-format=1.4]S[table-format=1.4]S[table-format=1.4]
            }
            \hline
            \multirow{2}{*}{Sampler ${\vcenter{\hbox{\begin{tikzpicture}[x=1mm,y=3.5mm]
                \draw[-] (1,-1) -- (-1,1);
            \end{tikzpicture}}}}$ \!\!\!\!\!\! $\begin{array}{c}\text{Target} \\d\end{array}$ \!\!\!\!\!\! }
                & \multicolumn{5}{c|}{Gauss $(\mu = 0)$}
                & \multicolumn{5}{c|}{Disc $(\mu = 0)$}
                & \multicolumn{5}{c|}{Cigar $(\mu = 0)$}
                & \multicolumn{5}{c|}{Gauss $(\mu = 0.5)$}
                & \multicolumn{5}{c|}{Disc $(\mu = 0.5)$}
                & \multicolumn{5}{c|}{Cigar $(\mu = 0.5)$}
                & \multicolumn{5}{c|}{Bowtie}
                & \multicolumn{5}{c}{Funnel}
            \\
            & \multicolumn{1}{c}{\footnotesize $2$} & \multicolumn{1}{c}{\footnotesize $4$} & \multicolumn{1}{c}{\footnotesize $8$} & \multicolumn{1}{c}{\footnotesize $16$} & \multicolumn{1}{c|}{\footnotesize $32$}
            & \multicolumn{1}{c}{\footnotesize $2$} & \multicolumn{1}{c}{\footnotesize $4$} & \multicolumn{1}{c}{\footnotesize $8$} & \multicolumn{1}{c}{\footnotesize $16$} & \multicolumn{1}{c|}{\footnotesize $32$}
            & \multicolumn{1}{c}{\footnotesize $2$} & \multicolumn{1}{c}{\footnotesize $4$} & \multicolumn{1}{c}{\footnotesize $8$} & \multicolumn{1}{c}{\footnotesize $16$} & \multicolumn{1}{c|}{\footnotesize $32$}
            & \multicolumn{1}{c}{\footnotesize $2$} & \multicolumn{1}{c}{\footnotesize $4$} & \multicolumn{1}{c}{\footnotesize $8$} & \multicolumn{1}{c}{\footnotesize $16$} & \multicolumn{1}{c|}{\footnotesize $32$}
            & \multicolumn{1}{c}{\footnotesize $2$} & \multicolumn{1}{c}{\footnotesize $4$} & \multicolumn{1}{c}{\footnotesize $8$} & \multicolumn{1}{c}{\footnotesize $16$} & \multicolumn{1}{c|}{\footnotesize $32$}
            & \multicolumn{1}{c}{\footnotesize $2$} & \multicolumn{1}{c}{\footnotesize $4$} & \multicolumn{1}{c}{\footnotesize $8$} & \multicolumn{1}{c}{\footnotesize $16$} & \multicolumn{1}{c|}{\footnotesize $32$}
            & \multicolumn{1}{c}{\footnotesize $2$} & \multicolumn{1}{c}{\footnotesize $4$} & \multicolumn{1}{c}{\footnotesize $8$} & \multicolumn{1}{c}{\footnotesize $16$} & \multicolumn{1}{c|}{\footnotesize $32$}
            & \multicolumn{1}{c}{\footnotesize $2$} & \multicolumn{1}{c}{\footnotesize $4$} & \multicolumn{1}{c}{\footnotesize $8$} & \multicolumn{1}{c}{\footnotesize $16$} & \multicolumn{1}{c }{\footnotesize $32$}
            \\
            \hline
 \RWMHm{}   & \cs 0.13           & \cs 0.067          & \cs 0.074          & \cs 0.09           & \cs 0.116          & \cs 0.414          & \cs 0.279          & \cs 0.272          & \cs 0.164          & \cs 0.226          & \cs 0.467         & \cs 0.219          & \cs 0.263          & \cs 0.247          & \cs 0.318          & \cs 0.118          & \cs 0.047          & \cs 0.077          & \cs 0.101        & \cs 0.159         & \cs 0.246          & \cs 0.2           & \cs 0.24           & \cs 0.29           & \cs 0.368          & \cs 0.297          & \cs 0.211          & \cs 0.247          & \cs 0.29           & \cs 0.369          & \cs 0.308          & \cs 0.374          & \cs 0.291          & \cs 0.253          & \cs 0.259          & \cs 0.471          & \cs 0.273          & \cs 0.334          & \cs 0.282          & \cs 0.333          \\
 \MALAm{}   & \cs 0.176          & \cs 0.089          & \cs 0.101          & \cs 0.096          & \cs 0.122          & \cs 0.582          & \cs 0.35           & \cs 0.332          & \cs 0.207          & \cs 0.261          & \cs 0.683         & \cs 0.3            & \cs 0.343          & \cs 0.254          & \cs 0.337          & \cs 0.137          & \cs 0.06           & \cs 0.1            & \cs 0.143        & \cs 0.212         & \cs 0.317          & \cs 0.29          & \cs 0.251          & \cs 0.307          & \cs 0.39           & \cs 0.646          & \cs 0.497          & \cs 0.495          & \cs 0.43           & \cs 0.51           & \cs 0.328          & \cs 0.341          & \cs 0.55           & \cs 0.573          & \cs 0.591          & \cs 0.392          & \cs 0.318          & \cs 0.437          & \cs 0.368          & \cs 0.49           \\
 \smMALAem{}& \cs 0.58           & \cs 0.593          & \cs 0.509          & \cs 0.432          & \cs 0.376          & \cs 0.094          & \cs 0.073          & \cs 0.053          & \cs 0.065          & \cs 0.084          & \cs 0.65          & \cs 0.295          & \cs 0.338          & \cs 0.269          & \cs 0.334          & \cs 0.68           & \cs 0.645          & \cs 0.619          & \cs 0.585        & \cs 0.628         & \cs 0.544          & \cs 0.41          & \cs 0.371          & \cs 0.315          & \cs 0.399          & \cs 0.625          & \cs 0.494          & \cs 0.471          & \cs 0.437          & \cs 0.515          & \cs 0.224          & \cs 0.152          & \cs 0.223          & \cs 0.215          & \cs 0.224          & \cs 0.398          & \cs 0.15           & \cs 0.26           & \cs 0.303          & \cs 0.328          \\
 \smMALAdm{}& \cs 0.58           & \cs 0.401          & \cs 0.43           & \cs 0.352          & \cs 0.183          & \cs 0.588          & \cs 0.346          & \cs 0.333          & \cs 0.22           & \cs 0.26           & \cs 0.653         & \cs 0.301          & \cs 0.358          & \cs 0.281          & \cs 0.348          & \cs 0.59           & \cs 0.495          & \cs 0.52           & \cs 0.575        & \cs 0.557         & \cs 0.645          & \cs 0.58          & \cs 0.455          & \cs 0.508          & \cs 0.55           & \cs 0.622          & \cs 0.53           & \cs 0.46           & \cs 0.477          & \cs 0.512          & \cs 0.495          & \cs 0.445          & \cs 0.463          & \cs 0.33           & \cs 0.366          & \cs 0.688          & \cs 0.38           & \cs 0.393          & \cs 0.483          & \cs 0.499          \\
 \Dikinm{}  & \cs 0.158          & \cs 0.148          & \cs 0.137          & \cs 0.179          & \cs 0.253          & \cs 0.619          & \cs \textbf{0.653} & \cs \textbf{0.719} & \cs \textbf{0.758} & \cs \textbf{0.854} & \cs 0.554         & \cs \textbf{0.581} & \cs \textbf{0.633} & \cs \textbf{0.645} & \cs \textbf{0.731} & \cs 0.102          & \cs 0.083          & \cs 0.095          & \cs 0.128        & \cs 0.211         & \cs 0.206          & \cs 0.253         & \cs 0.275          & \cs 0.347          & \cs 0.436          & \cs 0.312          & \cs 0.386          & \cs 0.419          & \cs 0.457          & \cs 0.527          & \cs 0.377          & \cs 0.628          & \cs 0.523          & \cs 0.38           & \cs 0.38           & \cs 0.337          & \cs 0.545          & \cs 0.512          & \cs 0.438          & \cs 0.517          \\
 \MAPLAm{}  & \cs 0.103          & \cs 0.13           & \cs 0.127          & \cs 0.142          & \cs 0.201          & \cs 0.38           & \cs 0.405          & \cs 0.471          & \cs 0.501          & \cs 0.624          & \cs 0.4           & \cs 0.413          & \cs 0.5            & \cs 0.472          & \cs 0.554          & \cs 0.163          & \cs 0.121          & \cs 0.124          & \cs 0.192        & \cs 0.235         & \cs 0.186          & \cs 0.214         & \cs 0.239          & \cs 0.282          & \cs 0.375          & \cs 0.332          & \cs 0.398          & \cs 0.429          & \cs 0.47           & \cs 0.52           & \cs 0.208          & \cs 0.544          & \cs 0.521          & \cs 0.407          & \cs 0.391          & \cs 0.207          & \cs 0.466          & \cs 0.503          & \cs 0.447          & \cs 0.534          \\
 \HRm{}     & \cs 0.207          & \cs 0.147          & \cs 0.137          & \cs 0.171          & \cs 0.216          & \cs 0.734          & \cs 0.513          & \cs 0.485          & \cs 0.466          & \cs 0.485          & \cs 0.804         & \cs 0.515          & \cs 0.499          & \cs 0.522          & \cs 0.549          & \cs 0.139          & \cs 0.114          & \cs 0.137          & \cs 0.186        & \cs 0.264         & \cs 0.329          & \cs 0.381         & \cs 0.418          & \cs 0.528          & \cs 0.619          & \cs 0.481          & \cs 0.438          & \cs 0.429          & \cs 0.496          & \cs 0.584          & \cs 0.367          & \cs 0.475          & \cs 0.374          & \cs 0.377          & \cs 0.416          & \cs 0.56           & \cs 0.446          & \cs 0.446          & \cs 0.443          & \cs 0.499          \\
 \LHRm{}    & \cs 0.261          & \cs 0.208          & \cs 0.226          & \cs 0.291          & \cs 0.313          & \cs \textbf{0.833} & \cs 0.522          & \cs 0.435          & \cs 0.47           & \cs 0.419          & \cs \textbf{0.85} & \cs 0.512          & \cs 0.539          & \cs 0.534          & \cs 0.491          & \cs 0.195          & \cs 0.175          & \cs 0.199          & \cs 0.258        & \cs 0.358         & \cs 0.38           & \cs 0.37          & \cs 0.366          & \cs 0.513          & \cs 0.539          & \cs \textbf{0.817} & \cs 0.575          & \cs \textbf{0.708} & \cs 0.674          & \cs 0.693          & \cs 0.404          & \cs 0.47           & \cs \textbf{0.713} & \cs \textbf{0.769} & \cs \textbf{0.813} & \cs 0.487          & \cs 0.462          & \cs 0.529          & \cs 0.514          & \cs 0.64           \\
 \smHRem{}  & \cs 0.557          & \cs 0.542          & \cs 0.496          & \cs 0.515          & \cs 0.53           & \cs 0.311          & \cs 0.145          & \cs 0.126          & \cs 0.14           & \cs 0.185          & \cs 0.77          & \cs 0.534          & \cs 0.498          & \cs 0.534          & \cs 0.545          & \cs 0.402          & \cs 0.375          & \cs 0.317          & \cs 0.379        & \cs 0.439         & \cs 0.361          & \cs 0.222         & \cs 0.253          & \cs 0.306          & \cs 0.38           & \cs 0.463          & \cs 0.443          & \cs 0.422          & \cs 0.506          & \cs 0.588          & \cs 0.449          & \cs 0.362          & \cs 0.286          & \cs 0.298          & \cs 0.267          & \cs 0.536          & \cs 0.468          & \cs 0.414          & \cs 0.439          & \cs 0.492          \\
 \smHRdm{}  & \cs 0.533          & \cs 0.578          & \cs 0.44           & \cs 0.472          & \cs 0.524          & \cs 0.688          & \cs 0.521          & \cs 0.486          & \cs 0.465          & \cs 0.469          & \cs 0.804         & \cs 0.518          & \cs 0.499          & \cs 0.526          & \cs 0.559          & \cs 0.305          & \cs 0.394          & \cs 0.402          & \cs 0.377        & \cs 0.4           & \cs 0.404          & \cs 0.432         & \cs 0.44           & \cs 0.509          & \cs 0.604          & \cs 0.458          & \cs 0.424          & \cs 0.47           & \cs 0.527          & \cs 0.586          & \cs 0.525          & \cs 0.499          & \cs 0.434          & \cs 0.404          & \cs 0.376          & \cs \textbf{0.844} & \cs \textbf{0.635} & \cs 0.528          & \cs 0.534          & \cs 0.553          \\
 \smLHRem{} & \cs 0.633          & \cs \textbf{0.659} & \cs \textbf{0.769} & \cs \textbf{0.772} & \cs \textbf{0.738} & \cs 0.121          & \cs 0.117          & \cs 0.104          & \cs 0.129          & \cs 0.165          & \cs 0.801         & \cs 0.513          & \cs 0.541          & \cs 0.512          & \cs 0.507          & \cs \textbf{0.855} & \cs 0.77           & \cs \textbf{0.801} & \cs \textbf{0.8} & \cs \textbf{0.81} & \cs 0.618          & \cs 0.48          & \cs 0.36           & \cs 0.49           & \cs 0.458          & \cs 0.808          & \cs 0.655          & \cs 0.704          & \cs 0.627          & \cs 0.72           & \cs 0.509          & \cs 0.346          & \cs 0.399          & \cs 0.417          & \cs 0.335          & \cs 0.549          & \cs 0.428          & \cs 0.458          & \cs 0.439          & \cs 0.449          \\
 \smLHRdm{} & \cs \textbf{0.699} & \cs 0.579          & \cs 0.586          & \cs 0.52           & \cs 0.475          & \cs 0.799          & \cs 0.511          & \cs 0.442          & \cs 0.468          & \cs 0.422          & \cs 0.835         & \cs 0.478          & \cs 0.545          & \cs 0.529          & \cs 0.515          & \cs 0.811          & \cs \textbf{0.832} & \cs 0.606          & \cs 0.686        & \cs 0.557         & \cs \textbf{0.673} & \cs \textbf{0.75} & \cs \textbf{0.697} & \cs \textbf{0.739} & \cs \textbf{0.729} & \cs 0.706          & \cs \textbf{0.683} & \cs 0.694          & \cs \textbf{0.714} & \cs \textbf{0.736} & \cs \textbf{0.729} & \cs \textbf{0.638} & \cs 0.666          & \cs 0.552          & \cs 0.521          & \cs 0.765          & \cs 0.497          & \cs \textbf{0.631} & \cs \textbf{0.657} & \cs \textbf{0.674} \\
\hline
\end{tabular}};
    \end{tikzpicture}
    }
\end{table}

\begin{table}[]
    \centering
    \caption{
        Average relative performance of the tested samplers across polytope angles as presented as solid lines in \cref{fig:ess1}.
    }
    \resizebox{\linewidth}{!}{
    \begin{tikzpicture}
        \node at (0, 0) {%
            \begin{tabular}{r|
                S[table-format=1.4]S[table-format=1.4]S[table-format=1.4]S[table-format=1.4]|
                S[table-format=1.4]S[table-format=1.4]S[table-format=1.4]S[table-format=1.4]|
                S[table-format=1.4]S[table-format=1.4]S[table-format=1.4]S[table-format=1.4]|
                S[table-format=1.4]S[table-format=1.4]S[table-format=1.4]S[table-format=1.4]|
                S[table-format=1.4]S[table-format=1.4]S[table-format=1.4]S[table-format=1.4]|
                S[table-format=1.4]S[table-format=1.4]S[table-format=1.4]S[table-format=1.4]|
                S[table-format=1.4]S[table-format=1.4]S[table-format=1.4]S[table-format=1.4]|
                S[table-format=1.4]S[table-format=1.4]S[table-format=1.4]S[table-format=1.4]
            }
            \hline
            \multirow{2}{*}{Sampler ${\vcenter{\hbox{\begin{tikzpicture}[x=1mm,y=3.5mm]
                \draw[-] (1,-1) -- (-1,1);
            \end{tikzpicture}}}}$ \!\!\!\!\!\! $\begin{array}{c}\text{Target} \\\theta\end{array}$ \!\!\!\!\!\! }
                & \multicolumn{4}{c|}{Gauss $(\mu = 0)$}
                & \multicolumn{4}{c|}{Disc $(\mu = 0)$}
                & \multicolumn{4}{c|}{Cigar $(\mu = 0)$}
                & \multicolumn{4}{c|}{Gauss $(\mu = 0.5)$}
                & \multicolumn{4}{c|}{Disc $(\mu = 0.5)$}
                & \multicolumn{4}{c|}{Cigar $(\mu = 0.5)$}
                & \multicolumn{4}{c|}{Bowtie}
                & \multicolumn{4}{c}{Funnel}
            \\
            & \multicolumn{1}{c}{\footnotesize $9^{\circ}$} & \multicolumn{1}{c}{\footnotesize $19^{\circ}$} & \multicolumn{1}{c}{\footnotesize $45^{\circ}$} & \multicolumn{1}{c|}{\footnotesize $90^{\circ}$}
            & \multicolumn{1}{c}{\footnotesize $9^{\circ}$} & \multicolumn{1}{c}{\footnotesize $19^{\circ}$} & \multicolumn{1}{c}{\footnotesize $45^{\circ}$} & \multicolumn{1}{c|}{\footnotesize $90^{\circ}$}
            & \multicolumn{1}{c}{\footnotesize $9^{\circ}$} & \multicolumn{1}{c}{\footnotesize $19^{\circ}$} & \multicolumn{1}{c}{\footnotesize $45^{\circ}$} & \multicolumn{1}{c|}{\footnotesize $90^{\circ}$}
            & \multicolumn{1}{c}{\footnotesize $9^{\circ}$} & \multicolumn{1}{c}{\footnotesize $19^{\circ}$} & \multicolumn{1}{c}{\footnotesize $45^{\circ}$} & \multicolumn{1}{c|}{\footnotesize $90^{\circ}$}
            & \multicolumn{1}{c}{\footnotesize $9^{\circ}$} & \multicolumn{1}{c}{\footnotesize $19^{\circ}$} & \multicolumn{1}{c}{\footnotesize $45^{\circ}$} & \multicolumn{1}{c|}{\footnotesize $90^{\circ}$}
            & \multicolumn{1}{c}{\footnotesize $9^{\circ}$} & \multicolumn{1}{c}{\footnotesize $19^{\circ}$} & \multicolumn{1}{c}{\footnotesize $45^{\circ}$} & \multicolumn{1}{c|}{\footnotesize $90^{\circ}$}
            & \multicolumn{1}{c}{\footnotesize $9^{\circ}$} & \multicolumn{1}{c}{\footnotesize $19^{\circ}$} & \multicolumn{1}{c}{\footnotesize $45^{\circ}$} & \multicolumn{1}{c|}{\footnotesize $90^{\circ}$}
            & \multicolumn{1}{c}{\footnotesize $9^{\circ}$} & \multicolumn{1}{c}{\footnotesize $19^{\circ}$} & \multicolumn{1}{c}{\footnotesize $45^{\circ}$} & \multicolumn{1}{c }{\footnotesize $90^{\circ}$}
            \\
            \hline
  \RWMHm{}   & \cs 0.01           & \cs 0.057          & \cs 0.102          & \cs 0.213          & \cs 0.031          & \cs 0.106          & \cs 0.413          & \cs 0.535          & \cs 0.033          & \cs 0.248          & \cs 0.435          & \cs 0.496         & \cs 0.013          & \cs 0.078          & \cs 0.119         & \cs 0.192          & \cs 0.198          & \cs 0.343          & \cs 0.287          & \cs 0.247          & \cs 0.124          & \cs 0.377          & \cs 0.312          & \cs 0.318          & \cs 0.185          & \cs 0.334          & \cs 0.272          & \cs 0.398          & \cs 0.149          & \cs 0.337          & \cs 0.366          & \cs 0.503          \\
  \MALAm{}   & \cs 0.012          & \cs 0.068          & \cs 0.12           & \cs 0.267          & \cs 0.035          & \cs 0.156          & \cs 0.547          & \cs 0.647          & \cs 0.036          & \cs 0.334          & \cs 0.496          & \cs 0.668         & \cs 0.017          & \cs 0.087          & \cs 0.152         & \cs 0.266          & \cs 0.247          & \cs 0.398          & \cs 0.335          & \cs 0.265          & \cs 0.268          & \cs 0.589          & \cs 0.587          & \cs 0.619          & \cs 0.29           & \cs 0.522          & \cs 0.452          & \cs 0.642          & \cs 0.198          & \cs 0.429          & \cs 0.461          & \cs 0.516          \\
  \smMALAem{}& \cs 0.41           & \cs 0.628          & \cs 0.482          & \cs 0.473          & \cs 0.002          & \cs 0.007          & \cs 0.024          & \cs 0.262          & \cs 0.043          & \cs 0.327          & \cs 0.495          & \cs 0.644         & \cs 0.596          & \cs 0.668          & \cs 0.664         & \cs 0.597          & \cs 0.22           & \cs 0.317          & \cs 0.462          & \cs 0.632          & \cs 0.266          & \cs 0.575          & \cs 0.58           & \cs 0.612          & \cs 0.124          & \cs 0.231          & \cs 0.207          & \cs 0.268          & \cs 0.147          & \cs 0.312          & \cs 0.351          & \cs 0.342          \\
  \smMALAdm{}& \cs 0.099          & \cs 0.317          & \cs 0.457          & \cs 0.685          & \cs 0.044          & \cs 0.166          & \cs 0.532          & \cs 0.657          & \cs 0.041          & \cs 0.328          & \cs 0.521          & \cs 0.663         & \cs 0.36           & \cs 0.479          & \cs 0.601         & \cs 0.75           & \cs 0.369          & \cs 0.59           & \cs 0.617          & \cs 0.612          & \cs 0.283          & \cs 0.578          & \cs 0.59           & \cs 0.631          & \cs 0.26           & \cs 0.437          & \cs 0.421          & \cs 0.561          & \cs 0.242          & \cs 0.52           & \cs 0.539          & \cs 0.654          \\
  \Dikinm{}  & \cs 0.309          & \cs 0.165          & \cs 0.103          & \cs 0.123          & \cs \textbf{0.978} & \cs \textbf{0.954} & \cs 0.573          & \cs 0.377          & \cs \textbf{0.946} & \cs \textbf{0.866} & \cs 0.446          & \cs 0.257         & \cs 0.168          & \cs 0.119          & \cs 0.089         & \cs 0.119          & \cs 0.496          & \cs 0.378          & \cs 0.169          & \cs 0.171          & \cs \textbf{0.715} & \cs 0.524          & \cs 0.213          & \cs 0.23           & \cs \textbf{0.652} & \cs 0.531          & \cs 0.336          & \cs 0.312          & \cs \textbf{0.697} & \cs 0.53           & \cs 0.336          & \cs 0.316          \\
  \MAPLAm{}  & \cs 0.232          & \cs 0.124          & \cs 0.081          & \cs 0.126          & \cs 0.647          & \cs 0.581          & \cs 0.337          & \cs 0.34           & \cs 0.728          & \cs 0.594          & \cs 0.298          & \cs 0.252         & \cs 0.287          & \cs 0.165          & \cs 0.091         & \cs 0.126          & \cs 0.431          & \cs 0.299          & \cs 0.129          & \cs 0.178          & \cs 0.631          & \cs 0.472          & \cs 0.243          & \cs 0.373          & \cs 0.477          & \cs 0.425          & \cs 0.381          & \cs 0.374          & \cs 0.633          & \cs 0.51           & \cs 0.309          & \cs 0.273          \\
  \HRm{}     & \cs 0.048          & \cs 0.113          & \cs 0.243          & \cs 0.299          & \cs 0.15           & \cs 0.34           & \cs \textbf{0.871} & \cs 0.786          & \cs 0.175          & \cs 0.495          & \cs 0.831          & \cs 0.809         & \cs 0.03           & \cs 0.116          & \cs 0.248         & \cs 0.279          & \cs 0.346          & \cs 0.542          & \cs 0.532          & \cs 0.401          & \cs 0.29           & \cs 0.61           & \cs 0.56           & \cs 0.483          & \cs 0.235          & \cs 0.432          & \cs 0.427          & \cs 0.513          & \cs 0.233          & \cs 0.458          & \cs 0.582          & \cs 0.641          \\
  \LHRm{}    & \cs 0.035          & \cs 0.161          & \cs 0.433          & \cs 0.41           & \cs 0.128          & \cs 0.363          & \cs 0.855          & \cs 0.798          & \cs 0.141          & \cs 0.503          & \cs 0.837          & \cs \textbf{0.86} & \cs 0.043          & \cs 0.171          & \cs 0.358         & \cs 0.377          & \cs 0.346          & \cs 0.525          & \cs 0.511          & \cs 0.353          & \cs 0.406          & \cs 0.753          & \cs \textbf{0.842} & \cs 0.773          & \cs 0.421          & \cs 0.654          & \cs \textbf{0.733} & \cs \textbf{0.728} & \cs 0.292          & \cs 0.493          & \cs 0.666          & \cs 0.655          \\
  \smHRem{}  & \cs 0.737          & \cs 0.625          & \cs 0.39           & \cs 0.359          & \cs 0.013          & \cs 0.029          & \cs 0.075          & \cs 0.608          & \cs 0.173          & \cs 0.502          & \cs 0.829          & \cs 0.8           & \cs 0.478          & \cs 0.446          & \cs 0.337         & \cs 0.268          & \cs 0.161          & \cs 0.23           & \cs 0.333          & \cs 0.494          & \cs 0.295          & \cs 0.583          & \cs 0.57           & \cs 0.49           & \cs 0.213          & \cs 0.362          & \cs 0.386          & \cs 0.367          & \cs 0.274          & \cs 0.498          & \cs 0.565          & \cs 0.541          \\
  \smHRdm{}  & \cs 0.804          & \cs 0.572          & \cs 0.311          & \cs 0.35           & \cs 0.129          & \cs 0.338          & \cs 0.758          & \cs \textbf{0.877} & \cs 0.192          & \cs 0.512          & \cs 0.858          & \cs 0.763         & \cs 0.459          & \cs 0.407          & \cs 0.321         & \cs 0.317          & \cs 0.346          & \cs 0.514          & \cs 0.528          & \cs 0.524          & \cs 0.297          & \cs 0.6            & \cs 0.566          & \cs 0.508          & \cs 0.333          & \cs 0.506          & \cs 0.435          & \cs 0.517          & \cs 0.432          & \cs \textbf{0.619} & \cs 0.73           & \cs 0.694          \\
  \smLHRem{} & \cs \textbf{0.872} & \cs \textbf{0.862} & \cs 0.548          & \cs 0.575          & \cs 0.005          & \cs 0.013          & \cs 0.05           & \cs 0.441          & \cs 0.138          & \cs 0.479          & \cs \textbf{0.864} & \cs 0.818         & \cs \textbf{0.933} & \cs \textbf{0.905} & \cs 0.732         & \cs 0.659          & \cs 0.283          & \cs 0.398          & \cs 0.531          & \cs 0.712          & \cs 0.386          & \cs \textbf{0.777} & \cs 0.838          & \cs \textbf{0.809} & \cs 0.307          & \cs 0.469          & \cs 0.446          & \cs 0.382          & \cs 0.299          & \cs 0.481          & \cs 0.579          & \cs 0.5            \\
  \smLHRdm{} & \cs 0.204          & \cs 0.493          & \cs \textbf{0.874} & \cs \textbf{0.717} & \cs 0.129          & \cs 0.342          & \cs 0.85           & \cs 0.792          & \cs 0.16           & \cs 0.498          & \cs 0.836          & \cs 0.827         & \cs 0.542          & \cs 0.695          & \cs \textbf{0.77} & \cs \textbf{0.786} & \cs \textbf{0.528} & \cs \textbf{0.748} & \cs \textbf{0.829} & \cs \textbf{0.766} & \cs 0.423          & \cs 0.756          & \cs 0.839          & \cs 0.808          & \cs 0.457          & \cs \textbf{0.671} & \cs 0.656          & \cs 0.699          & \cs 0.4            & \cs 0.606          & \cs \textbf{0.821} & \cs \textbf{0.752} \\
 \hline
\end{tabular}};
    \end{tikzpicture}
    }
\end{table}

\begin{table}[]
    \centering
    \caption{
        Average relative performance of the tested samplers across the densities' scale parameters
        as presented as solid lines in \cref{fig:ess1}.
    }
    \resizebox{.9\linewidth}{!}{
    \begin{tikzpicture}
        \node at (0, 0) {%
            \begin{tabular}{r|
                S[table-format=1.4]S[table-format=1.4]S[table-format=1.4]S[table-format=1.4]S[table-format=1.4]S[table-format=1.4]S[table-format=1.4]|
                S[table-format=1.4]S[table-format=1.4]S[table-format=1.4]S[table-format=1.4]S[table-format=1.4]S[table-format=1.4]S[table-format=1.4]|
                S[table-format=1.4]S[table-format=1.4]S[table-format=1.4]S[table-format=1.4]S[table-format=1.4]S[table-format=1.4]S[table-format=1.4]|
                S[table-format=1.4]S[table-format=1.4]S[table-format=1.4]S[table-format=1.4]S[table-format=1.4]S[table-format=1.4]S[table-format=1.4]
            }
            \hline
            \multirow{2}{*}{Sampler ${\vcenter{\hbox{\begin{tikzpicture}[x=1mm,y=3.5mm]
                \draw[-] (1,-1) -- (-1,1);
            \end{tikzpicture}}}}$ \!\!\!\!\!\! $\begin{array}{c}\text{Target} \\\sigma\end{array}$ \!\!\!\!\!\! }
                & \multicolumn{7}{c|}{Gauss $(\mu = 0)$}
                & \multicolumn{7}{c|}{Disc $(\mu = 0)$}
                & \multicolumn{7}{c|}{Cigar $(\mu = 0)$}
                & \multicolumn{7}{c }{Gauss $(\mu = 0.5)$}
            \\
            & \multicolumn{1}{c}{\footnotesize $10^{-2}$} & \multicolumn{1}{c}{\footnotesize $10^{-1.5}$} & \multicolumn{1}{c}{\footnotesize $10^{-1}$} & \multicolumn{1}{c}{\footnotesize $10^{-0.5}$} & \multicolumn{1}{c}{\footnotesize $10^{0}$} & \multicolumn{1}{c}{\footnotesize $10^{0.5}$} & \multicolumn{1}{c|}{\footnotesize $10^{1}$}
            & \multicolumn{1}{c}{\footnotesize $10^{-2}$} & \multicolumn{1}{c}{\footnotesize $10^{-1.5}$} & \multicolumn{1}{c}{\footnotesize $10^{-1}$} & \multicolumn{1}{c}{\footnotesize $10^{-0.5}$} & \multicolumn{1}{c}{\footnotesize $10^{0}$} & \multicolumn{1}{c}{\footnotesize $10^{0.5}$} & \multicolumn{1}{c|}{\footnotesize $10^{1}$}
            & \multicolumn{1}{c}{\footnotesize $10^{-2}$} & \multicolumn{1}{c}{\footnotesize $10^{-1.5}$} & \multicolumn{1}{c}{\footnotesize $10^{-1}$} & \multicolumn{1}{c}{\footnotesize $10^{-0.5}$} & \multicolumn{1}{c}{\footnotesize $10^{0}$} & \multicolumn{1}{c}{\footnotesize $10^{0.5}$} & \multicolumn{1}{c|}{\footnotesize $10^{1}$}
            & \multicolumn{1}{c}{\footnotesize $10^{-2}$} & \multicolumn{1}{c}{\footnotesize $10^{-1.5}$} & \multicolumn{1}{c}{\footnotesize $10^{-1}$} & \multicolumn{1}{c}{\footnotesize $10^{-0.5}$} & \multicolumn{1}{c}{\footnotesize $10^{0}$} & \multicolumn{1}{c}{\footnotesize $10^{0.5}$} & \multicolumn{1}{c }{\footnotesize $10^{1}$}
            \\
            \hline
  \RWMHm{}   & \cs 0.018          & \cs 0.024          & \cs 0.015          & \cs 0.044          & \cs 0.101          & \cs 0.226          & \cs 0.24           & \cs 0.305          & \cs 0.274          & \cs 0.272          & \cs 0.278         & \cs 0.255          & \cs 0.217          & \cs 0.296          & \cs 0.29           & \cs 0.257          & \cs 0.249          & \cs 0.278          & \cs 0.344          & \cs 0.366          & \cs 0.338          & \cs 0.008          & \cs 0.008         & \cs 0.017          & \cs 0.053        & \cs 0.132          & \cs 0.244          & \cs 0.239           \\
  \MALAm{}   & \cs 0.04           & \cs 0.038          & \cs 0.025          & \cs 0.074          & \cs 0.156          & \cs 0.237          & \cs 0.247          & \cs 0.386          & \cs 0.415          & \cs 0.33           & \cs 0.343         & \cs 0.329          & \cs 0.315          & \cs 0.306          & \cs 0.426          & \cs 0.415          & \cs 0.385          & \cs 0.386          & \cs 0.363          & \cs 0.366          & \cs 0.343          & \cs 0.022          & \cs 0.021         & \cs 0.043          & \cs 0.101        & \cs 0.192          & \cs 0.291          & \cs 0.242           \\
  \smMALAem{}& \cs 0.669          & \cs 0.686          & \cs 0.527          & \cs 0.601          & \cs 0.438          & \cs 0.296          & \cs 0.269          & \cs 0.12           & \cs 0.064          & \cs 0.059          & \cs 0.055         & \cs 0.055          & \cs 0.077          & \cs 0.088          & \cs 0.458          & \cs 0.38           & \cs 0.385          & \cs 0.36           & \cs 0.366          & \cs 0.347          & \cs 0.345          & \cs 0.942          & \cs 0.915         & \cs 0.865          & \cs 0.634        & \cs 0.483          & \cs 0.285          & \cs 0.295           \\
  \smMALAdm{}& \cs 0.431          & \cs 0.546          & \cs 0.44           & \cs 0.523          & \cs 0.291          & \cs 0.258          & \cs 0.236          & \cs 0.397          & \cs 0.366          & \cs 0.362          & \cs 0.369         & \cs 0.351          & \cs 0.312          & \cs 0.29           & \cs 0.436          & \cs 0.396          & \cs 0.402          & \cs 0.394          & \cs 0.378          & \cs 0.367          & \cs 0.346          & \cs 0.881          & \cs 0.82          & \cs 0.754          & \cs 0.536        & \cs 0.313          & \cs 0.277          & \cs 0.251           \\
  \Dikinm{}  & \cs 0.155          & \cs 0.171          & \cs 0.126          & \cs 0.176          & \cs 0.158          & \cs 0.216          & \cs 0.223          & \cs \textbf{0.758} & \cs \textbf{0.761} & \cs \textbf{0.751} & \cs \textbf{0.74} & \cs \textbf{0.749} & \cs \textbf{0.694} & \cs 0.59           & \cs \textbf{0.682} & \cs \textbf{0.668} & \cs \textbf{0.684} & \cs \textbf{0.648} & \cs 0.577          & \cs 0.575          & \cs 0.567          & \cs 0.075          & \cs 0.075         & \cs 0.074          & \cs 0.085        & \cs 0.138          & \cs 0.2            & \cs 0.22            \\
  \MAPLAm{}  & \cs 0.105          & \cs 0.108          & \cs 0.095          & \cs 0.131          & \cs 0.143          & \cs 0.196          & \cs 0.207          & \cs 0.459          & \cs 0.444          & \cs 0.444          & \cs 0.442         & \cs 0.482          & \cs 0.542          & \cs 0.522          & \cs 0.403          & \cs 0.411          & \cs 0.423          & \cs 0.5            & \cs 0.506          & \cs 0.531          & \cs 0.501          & \cs 0.207          & \cs 0.186         & \cs 0.174          & \cs 0.082        & \cs 0.131          & \cs 0.185          & \cs 0.203           \\
  \HRm{}     & \cs 0.033          & \cs 0.053          & \cs 0.038          & \cs 0.09           & \cs 0.185          & \cs 0.395          & \cs 0.435          & \cs 0.52           & \cs 0.508          & \cs 0.506          & \cs 0.535         & \cs 0.504          & \cs 0.566          & \cs \textbf{0.617} & \cs 0.498          & \cs 0.469          & \cs 0.501          & \cs 0.578          & \cs \textbf{0.678} & \cs \textbf{0.676} & \cs 0.644          & \cs 0.008          & \cs 0.009         & \cs 0.021          & \cs 0.097        & \cs 0.209          & \cs 0.383          & \cs 0.45            \\
  \LHRm{}    & \cs 0.101          & \cs 0.099          & \cs 0.085          & \cs 0.16           & \cs 0.322          & \cs 0.512          & \cs 0.542          & \cs 0.56           & \cs 0.556          & \cs 0.533          & \cs 0.551         & \cs 0.484          & \cs 0.556          & \cs 0.513          & \cs 0.543          & \cs 0.538          & \cs 0.572          & \cs 0.59           & \cs 0.575          & \cs 0.593          & \cs \textbf{0.685} & \cs 0.023          & \cs 0.022         & \cs 0.079          & \cs 0.175        & \cs 0.353          & \cs 0.526          & \cs 0.482           \\
  \smHRem{}  & \cs 0.425          & \cs 0.488          & \cs 0.419          & \cs 0.526          & \cs 0.594          & \cs 0.62           & \cs 0.624          & \cs 0.215          & \cs 0.17           & \cs 0.173          & \cs 0.153         & \cs 0.185          & \cs 0.175          & \cs 0.2            & \cs 0.489          & \cs 0.448          & \cs 0.513          & \cs 0.554          & \cs \textbf{0.678} & \cs 0.668          & \cs 0.683          & \cs 0.218          & \cs 0.213         & \cs 0.265          & \cs 0.289        & \cs 0.544          & \cs 0.542          & \cs 0.605           \\
  \smHRdm{}  & \cs 0.352          & \cs 0.363          & \cs 0.341          & \cs 0.455          & \cs 0.54           & \cs \textbf{0.751} & \cs \textbf{0.763} & \cs 0.505          & \cs 0.506          & \cs 0.498          & \cs 0.505         & \cs 0.554          & \cs 0.543          & \cs 0.57           & \cs 0.496          & \cs 0.46           & \cs 0.524          & \cs 0.569          & \cs 0.677          & \cs 0.673          & \cs 0.668          & \cs 0.142          & \cs 0.141         & \cs 0.205          & \cs 0.282        & \cs 0.463          & \cs 0.65           & \cs \textbf{0.748}  \\
  \smLHRem{} & \cs \textbf{0.807} & \cs \textbf{0.848} & \cs \textbf{0.739} & \cs \textbf{0.741} & \cs \textbf{0.685} & \cs 0.61           & \cs 0.569          & \cs 0.147          & \cs 0.114          & \cs 0.105          & \cs 0.103         & \cs 0.109          & \cs 0.145          & \cs 0.168          & \cs 0.569          & \cs 0.51           & \cs 0.594          & \cs 0.512          & \cs 0.575          & \cs 0.586          & \cs 0.676          & \cs \textbf{0.963} & \cs \textbf{0.94} & \cs \textbf{0.948} & \cs \textbf{0.8} & \cs \textbf{0.744} & \cs \textbf{0.659} & \cs 0.596           \\
  \smLHRdm{} & \cs 0.539          & \cs 0.617          & \cs 0.599          & \cs 0.653          & \cs 0.502          & \cs 0.55           & \cs 0.543          & \cs 0.542          & \cs 0.518          & \cs 0.54           & \cs 0.522         & \cs 0.505          & \cs 0.551          & \cs 0.521          & \cs 0.592          & \cs 0.494          & \cs 0.546          & \cs 0.582          & \cs 0.582          & \cs 0.611          & \cs 0.654          & \cs 0.795          & \cs 0.826         & \cs 0.795          & \cs 0.705        & \cs 0.592          & \cs 0.607          & \cs 0.568           \\
            \hline
\end{tabular}};
    \end{tikzpicture}
    }
\end{table}

\begin{table}[]
    \centering
    \caption{
        Average relative performance of the tested samplers across the densities' scale parameters
        as presented as solid lines in \cref{fig:ess1}.
    }
    \resizebox{.9\linewidth}{!}{
    \begin{tikzpicture}
        \node at (0, 0) {%
            \begin{tabular}{r|
                S[table-format=1.4]S[table-format=1.4]S[table-format=1.4]S[table-format=1.4]S[table-format=1.4]S[table-format=1.4]S[table-format=1.4]|
                S[table-format=1.4]S[table-format=1.4]S[table-format=1.4]S[table-format=1.4]S[table-format=1.4]S[table-format=1.4]S[table-format=1.4]|
                S[table-format=1.4]S[table-format=1.4]S[table-format=1.4]S[table-format=1.4]S[table-format=1.4]S[table-format=1.4]S[table-format=1.4]|
                S[table-format=1.4]S[table-format=1.4]S[table-format=1.4]S[table-format=1.4]S[table-format=1.4]S[table-format=1.4]S[table-format=1.4]
            }
            \hline
            \multirow{2}{*}{Sampler ${\vcenter{\hbox{\begin{tikzpicture}[x=1mm,y=3.5mm]
                \draw[-] (1,-1) -- (-1,1);
            \end{tikzpicture}}}}$ \!\!\!\!\!\! $\begin{array}{c}\text{Target} \\\sigma\end{array}$ \!\!\!\!\!\! }
                & \multicolumn{7}{c|}{Disc $(\mu = 0.5)$}
                & \multicolumn{7}{c|}{Cigar $(\mu = 0.5)$}
                & \multicolumn{7}{c|}{Bowtie}
                & \multicolumn{7}{c }{Funnel}
            \\
            & \multicolumn{1}{c}{\footnotesize $10^{-2}$} & \multicolumn{1}{c}{\footnotesize $10^{-1.5}$} & \multicolumn{1}{c}{\footnotesize $10^{-1}$} & \multicolumn{1}{c}{\footnotesize $10^{-0.5}$} & \multicolumn{1}{c}{\footnotesize $10^{0}$} & \multicolumn{1}{c}{\footnotesize $10^{0.5}$} & \multicolumn{1}{c|}{\footnotesize $10^{1}$}
            & \multicolumn{1}{c}{\footnotesize $10^{-2}$} & \multicolumn{1}{c}{\footnotesize $10^{-1.5}$} & \multicolumn{1}{c}{\footnotesize $10^{-1}$} & \multicolumn{1}{c}{\footnotesize $10^{-0.5}$} & \multicolumn{1}{c}{\footnotesize $10^{0}$} & \multicolumn{1}{c}{\footnotesize $10^{0.5}$} & \multicolumn{1}{c|}{\footnotesize $10^{1}$}
            & \multicolumn{1}{c}{\footnotesize $10^{-2}$} & \multicolumn{1}{c}{\footnotesize $10^{-1.5}$} & \multicolumn{1}{c}{\footnotesize $10^{-1}$} & \multicolumn{1}{c}{\footnotesize $10^{-0.5}$} & \multicolumn{1}{c}{\footnotesize $10^{0}$} & \multicolumn{1}{c}{\footnotesize $10^{0.5}$} & \multicolumn{1}{c|}{\footnotesize $10^{1}$}
            & \multicolumn{1}{c}{\footnotesize $10^{-2}$} & \multicolumn{1}{c}{\footnotesize $10^{-1.5}$} & \multicolumn{1}{c}{\footnotesize $10^{-1}$} & \multicolumn{1}{c}{\footnotesize $10^{-0.5}$} & \multicolumn{1}{c}{\footnotesize $10^{0}$} & \multicolumn{1}{c}{\footnotesize $10^{0.5}$} & \multicolumn{1}{c }{\footnotesize $10^{1}$}
            \\
            \hline
  \RWMHm{}   & \cs 0.103          & \cs 0.147         & \cs 0.221          & \cs 0.283          & \cs 0.366          & \cs 0.38           & \cs 0.382          & \cs 0.135          & \cs 0.162          & \cs 0.28           & \cs 0.323         & \cs 0.385          & \cs 0.352          & \cs 0.342          & \cs 0.317          & \cs 0.407          & \cs 0.332          & \cs 0.187          & \cs 0.248          & \cs 0.29           & \cs 0.302          & \cs 0.366          & \cs 0.326          & \cs 0.284          & \cs 0.304          & \cs 0.35           & \cs 0.377          & \cs 0.363          \\
  \MALAm{}   & \cs 0.104          & \cs 0.131         & \cs 0.273          & \cs 0.385          & \cs 0.464          & \cs 0.437          & \cs 0.383          & \cs 0.857          & \cs 0.671          & \cs 0.588          & \cs 0.415         & \cs 0.404          & \cs 0.352          & \cs 0.323          & \cs 0.62           & \cs 0.711          & \cs 0.587          & \cs 0.415          & \cs 0.373          & \cs 0.33           & \cs 0.299          & \cs 0.463          & \cs 0.374          & \cs 0.428          & \cs 0.404          & \cs 0.373          & \cs 0.386          & \cs 0.379          \\
  \smMALAem{}& \cs \textbf{0.866} & \cs 0.721         & \cs 0.51           & \cs 0.315          & \cs 0.195          & \cs 0.132          & \cs 0.114          & \cs 0.857          & \cs 0.647          & \cs 0.589          & \cs 0.395         & \cs 0.414          & \cs 0.33           & \cs 0.325          & \cs 0.072          & \cs 0.134          & \cs 0.15           & \cs 0.223          & \cs 0.313          & \cs 0.253          & \cs 0.307          & \cs 0.117          & \cs 0.176          & \cs 0.262          & \cs 0.305          & \cs 0.368          & \cs 0.423          & \cs 0.364          \\
  \smMALAdm{}& \cs 0.716          & \cs 0.683         & \cs 0.557          & \cs 0.517          & \cs 0.528          & \cs 0.448          & \cs 0.383          & \cs \textbf{0.859} & \cs 0.678          & \cs 0.618          & \cs 0.408         & \cs 0.394          & \cs 0.358          & \cs 0.327          & \cs 0.536          & \cs 0.575          & \cs 0.476          & \cs 0.402          & \cs 0.345          & \cs 0.306          & \cs 0.3            & \cs 0.703          & \cs 0.643          & \cs 0.512          & \cs 0.435          & \cs 0.364          & \cs 0.396          & \cs 0.368          \\
  \Dikinm{}  & \cs 0.036          & \cs 0.063         & \cs 0.157          & \cs 0.251          & \cs 0.511          & \cs 0.545          & \cs 0.562          & \cs 0.168          & \cs 0.232          & \cs 0.425          & \cs 0.507         & \cs 0.576          & \cs 0.521          & \cs 0.513          & \cs 0.536          & \cs 0.528          & \cs 0.492          & \cs 0.396          & \cs 0.399          & \cs 0.411          & \cs 0.443          & \cs 0.423          & \cs 0.476          & \cs 0.385          & \cs 0.487          & \cs 0.522          & \cs 0.494          & \cs 0.502          \\
  \MAPLAm{}  & \cs 0.027          & \cs 0.05          & \cs 0.118          & \cs 0.195          & \cs 0.421          & \cs 0.518          & \cs 0.485          & \cs 0.312          & \cs 0.342          & \cs 0.43           & \cs 0.449         & \cs 0.499          & \cs 0.488          & \cs 0.489          & \cs 0.4            & \cs 0.466          & \cs 0.479          & \cs 0.392          & \cs 0.362          & \cs 0.394          & \cs 0.407          & \cs 0.384          & \cs 0.441          & \cs 0.406          & \cs 0.466          & \cs 0.436          & \cs 0.442          & \cs 0.444          \\
  \HRm{}     & \cs 0.13           & \cs 0.213         & \cs 0.331          & \cs 0.467          & \cs 0.616          & \cs \textbf{0.702} & \cs \textbf{0.726} & \cs 0.174          & \cs 0.209          & \cs 0.416          & \cs 0.591         & \cs 0.692          & \cs 0.689          & \cs 0.628          & \cs 0.321          & \cs 0.415          & \cs 0.382          & \cs 0.264          & \cs 0.397          & \cs 0.489          & \cs 0.546          & \cs 0.304          & \cs 0.321          & \cs 0.353          & \cs 0.518          & \cs 0.64           & \cs 0.613          & \cs 0.602          \\
  \LHRm{}    & \cs 0.112          & \cs 0.203         & \cs 0.335          & \cs 0.519          & \cs 0.61           & \cs 0.662          & \cs 0.594          & \cs 0.787          & \cs 0.636          & \cs 0.73           & \cs \textbf{0.68} & \cs 0.682          & \cs 0.658          & \cs 0.683          & \cs \textbf{0.642} & \cs \textbf{0.727} & \cs \textbf{0.728} & \cs \textbf{0.569} & \cs 0.562          & \cs 0.614          & \cs 0.596          & \cs 0.46           & \cs 0.417          & \cs 0.522          & \cs \textbf{0.634} & \cs 0.53           & \cs 0.567          & \cs 0.555          \\
  \smHRem{}  & \cs 0.264          & \cs 0.286         & \cs 0.364          & \cs 0.387          & \cs 0.348          & \cs 0.25           & \cs 0.233          & \cs 0.168          & \cs 0.217          & \cs 0.416          & \cs 0.591         & \cs 0.692          & \cs 0.656          & \cs 0.651          & \cs 0.107          & \cs 0.204          & \cs 0.277          & \cs 0.352          & \cs 0.445          & \cs 0.423          & \cs 0.517          & \cs 0.22           & \cs 0.26           & \cs 0.416          & \cs 0.445          & \cs 0.613          & \cs 0.682          & \cs 0.655          \\
  \smHRdm{}  & \cs 0.283          & \cs 0.325         & \cs 0.394          & \cs 0.448          & \cs 0.628          & \cs 0.634          & \cs 0.635          & \cs 0.191          & \cs 0.222          & \cs 0.427          & \cs 0.576         & \cs \textbf{0.733} & \cs 0.656          & \cs 0.645          & \cs 0.392          & \cs 0.423          & \cs 0.374          & \cs 0.323          & \cs 0.426          & \cs 0.554          & \cs 0.641          & \cs 0.529          & \cs 0.447          & \cs 0.511          & \cs 0.545          & \cs \textbf{0.688} & \cs \textbf{0.805} & \cs \textbf{0.806} \\
  \smLHRem{} & \cs 0.837          & \cs 0.747         & \cs 0.691          & \cs 0.42           & \cs 0.27           & \cs 0.2            & \cs 0.203          & \cs 0.775          & \cs \textbf{0.706} & \cs 0.722          & \cs 0.66          & \cs 0.682          & \cs 0.655          & \cs 0.719          & \cs 0.086          & \cs 0.207          & \cs 0.445          & \cs 0.4            & \cs 0.569          & \cs 0.541          & \cs 0.561          & \cs 0.286          & \cs 0.39           & \cs 0.491          & \cs 0.414          & \cs 0.486          & \cs 0.612          & \cs 0.573          \\
  \smLHRdm{} & \cs 0.854          & \cs \textbf{0.84} & \cs \textbf{0.771} & \cs \textbf{0.703} & \cs \textbf{0.639} & \cs 0.598          & \cs 0.617          & \cs 0.8            & \cs 0.656          & \cs \textbf{0.739} & \cs 0.628         & \cs 0.68           & \cs \textbf{0.711} & \cs \textbf{0.731} & \cs 0.552          & \cs 0.614          & \cs 0.543          & \cs 0.56           & \cs \textbf{0.718} & \cs \textbf{0.686} & \cs \textbf{0.675} & \cs \textbf{0.828} & \cs \textbf{0.759} & \cs \textbf{0.701} & \cs 0.618          & \cs 0.531          & \cs 0.535          & \cs 0.541          \\
            \hline
\end{tabular}};
    \end{tikzpicture}
    }
\end{table}

\end{landscape}
\newpage
\begin{landscape}

\begin{table}[]
    \centering
    \caption{
        Average relative performance based on \minESS{}/s of the tested samplers across dimensions as presented as solid lines in \cref{fig:esst1}.
    }
    \resizebox{\linewidth}{!}{
    \begin{tikzpicture}
        \node at (0, 0) {%
            \begin{tabular}{r|
                S[table-format=1.4]S[table-format=1.4]S[table-format=1.4]S[table-format=1.4]S[table-format=1.4]|
                S[table-format=1.4]S[table-format=1.4]S[table-format=1.4]S[table-format=1.4]S[table-format=1.4]|
                S[table-format=1.4]S[table-format=1.4]S[table-format=1.4]S[table-format=1.4]S[table-format=1.4]|
                S[table-format=1.4]S[table-format=1.4]S[table-format=1.4]S[table-format=1.4]S[table-format=1.4]|
                S[table-format=1.4]S[table-format=1.4]S[table-format=1.4]S[table-format=1.4]S[table-format=1.4]|
                S[table-format=1.4]S[table-format=1.4]S[table-format=1.4]S[table-format=1.4]S[table-format=1.4]|
                S[table-format=1.4]S[table-format=1.4]S[table-format=1.4]S[table-format=1.4]S[table-format=1.4]|
                S[table-format=1.4]S[table-format=1.4]S[table-format=1.4]S[table-format=1.4]S[table-format=1.4]
            }
            \hline
            \multirow{2}{*}{Sampler ${\vcenter{\hbox{\begin{tikzpicture}[x=1mm,y=3.5mm]
                \draw[-] (1,-1) -- (-1,1);
            \end{tikzpicture}}}}$ \!\!\!\!\!\! $\begin{array}{c}\text{Target} \\d\end{array}$ \!\!\!\!\!\! }
                & \multicolumn{5}{c|}{Gauss $(\mu = 0)$}
                & \multicolumn{5}{c|}{Disc $(\mu = 0)$}
                & \multicolumn{5}{c|}{Cigar $(\mu = 0)$}
                & \multicolumn{5}{c|}{Gauss $(\mu = 0.5)$}
                & \multicolumn{5}{c|}{Disc $(\mu = 0.5)$}
                & \multicolumn{5}{c|}{Cigar $(\mu = 0.5)$}
                & \multicolumn{5}{c|}{Bowtie}
                & \multicolumn{5}{c}{Funnel}
            \\
            & \multicolumn{1}{c}{\footnotesize $2$} & \multicolumn{1}{c}{\footnotesize $4$} & \multicolumn{1}{c}{\footnotesize $8$} & \multicolumn{1}{c}{\footnotesize $16$} & \multicolumn{1}{c|}{\footnotesize $32$}
            & \multicolumn{1}{c}{\footnotesize $2$} & \multicolumn{1}{c}{\footnotesize $4$} & \multicolumn{1}{c}{\footnotesize $8$} & \multicolumn{1}{c}{\footnotesize $16$} & \multicolumn{1}{c|}{\footnotesize $32$}
            & \multicolumn{1}{c}{\footnotesize $2$} & \multicolumn{1}{c}{\footnotesize $4$} & \multicolumn{1}{c}{\footnotesize $8$} & \multicolumn{1}{c}{\footnotesize $16$} & \multicolumn{1}{c|}{\footnotesize $32$}
            & \multicolumn{1}{c}{\footnotesize $2$} & \multicolumn{1}{c}{\footnotesize $4$} & \multicolumn{1}{c}{\footnotesize $8$} & \multicolumn{1}{c}{\footnotesize $16$} & \multicolumn{1}{c|}{\footnotesize $32$}
            & \multicolumn{1}{c}{\footnotesize $2$} & \multicolumn{1}{c}{\footnotesize $4$} & \multicolumn{1}{c}{\footnotesize $8$} & \multicolumn{1}{c}{\footnotesize $16$} & \multicolumn{1}{c|}{\footnotesize $32$}
            & \multicolumn{1}{c}{\footnotesize $2$} & \multicolumn{1}{c}{\footnotesize $4$} & \multicolumn{1}{c}{\footnotesize $8$} & \multicolumn{1}{c}{\footnotesize $16$} & \multicolumn{1}{c|}{\footnotesize $32$}
            & \multicolumn{1}{c}{\footnotesize $2$} & \multicolumn{1}{c}{\footnotesize $4$} & \multicolumn{1}{c}{\footnotesize $8$} & \multicolumn{1}{c}{\footnotesize $16$} & \multicolumn{1}{c|}{\footnotesize $32$}
            & \multicolumn{1}{c}{\footnotesize $2$} & \multicolumn{1}{c}{\footnotesize $4$} & \multicolumn{1}{c}{\footnotesize $8$} & \multicolumn{1}{c}{\footnotesize $16$} & \multicolumn{1}{c }{\footnotesize $32$}
            \\
            \hline
 \RWMHm{}   &\cs 0.208         & \cs 0.162          & \cs 0.181          & \cs 0.226          & \cs 0.256          & \cs 0.485         & \cs 0.468          & \cs 0.478          & \cs 0.47           & \cs 0.477          & \cs 0.513          & \cs 0.464          & \cs 0.499          & \cs 0.616          & \cs 0.616         & \cs 0.183          & \cs 0.143          & \cs 0.198          & \cs 0.251          & \cs 0.31           & \cs 0.303          & \cs 0.325          & \cs 0.493          & \cs 0.649          & \cs 0.725          & \cs 0.389          & \cs 0.395         & \cs 0.512          & \cs 0.685          & \cs 0.76           & \cs 0.528          & \cs 0.613          & \cs 0.522          & \cs 0.566          & \cs 0.544          & \cs 0.626          & \cs 0.511         & \cs 0.633          & \cs 0.692          & \cs 0.654          \\
 \MALAm{}   &\cs 0.266         & \cs 0.191          & \cs 0.222          & \cs 0.227          & \cs 0.265          & \cs 0.642         & \cs 0.535          & \cs 0.55           & \cs 0.515          & \cs 0.521          & \cs \textbf{0.727} & \cs 0.546          & \cs 0.579          & \cs \textbf{0.634} & \cs \textbf{0.63} & \cs 0.192          & \cs 0.167          & \cs 0.239          & \cs 0.319          & \cs 0.379          & \cs 0.373          & \cs 0.417          & \cs 0.512          & \cs \textbf{0.665} & \cs \textbf{0.755} & \cs \textbf{0.749} & \cs \textbf{0.68} & \cs \textbf{0.784} & \cs \textbf{0.851} & \cs \textbf{0.921} & \cs 0.507          & \cs 0.565          & \cs \textbf{0.818} & \cs \textbf{0.898} & \cs \textbf{0.916} & \cs 0.452          & \cs 0.521         & \cs \textbf{0.724} & \cs \textbf{0.803} & \cs \textbf{0.881} \\
 \smMALAem{}&\cs \textbf{0.74} & \cs \textbf{0.786} & \cs \textbf{0.715} & \cs \textbf{0.594} & \cs \textbf{0.632} & \cs 0.087         & \cs 0.081          & \cs 0.062          & \cs 0.11           & \cs 0.092          & \cs 0.541          & \cs 0.389          & \cs 0.394          & \cs 0.384          & \cs 0.434         & \cs \textbf{0.793} & \cs \textbf{0.811} & \cs \textbf{0.807} & \cs \textbf{0.753} & \cs \textbf{0.792} & \cs 0.581          & \cs 0.503          & \cs 0.472          & \cs 0.398          & \cs 0.482          & \cs 0.565          & \cs 0.484         & \cs 0.523          & \cs 0.511          & \cs 0.662          & \cs 0.376          & \cs 0.19           & \cs 0.203          & \cs 0.203          & \cs 0.168          & \cs 0.375          & \cs 0.185         & \cs 0.243          & \cs 0.281          & \cs 0.254          \\
 \smMALAdm{}&\cs 0.652         & \cs 0.501          & \cs 0.516          & \cs 0.424          & \cs 0.239          & \cs 0.515         & \cs 0.379          & \cs 0.377          & \cs 0.334          & \cs 0.339          & \cs 0.545          & \cs 0.398          & \cs 0.416          & \cs 0.406          & \cs 0.457         & \cs 0.622          & \cs 0.599          & \cs 0.614          & \cs 0.628          & \cs 0.618          & \cs \textbf{0.677} & \cs \textbf{0.636} & \cs \textbf{0.583} & \cs 0.629          & \cs 0.678          & \cs 0.578          & \cs 0.534         & \cs 0.529          & \cs 0.559          & \cs 0.672          & \cs \textbf{0.592} & \cs 0.378          & \cs 0.364          & \cs 0.222          & \cs 0.262          & \cs \textbf{0.652} & \cs 0.363         & \cs 0.344          & \cs 0.426          & \cs 0.396          \\
 \Dikinm{}  &\cs 0.272         & \cs 0.21           & \cs 0.209          & \cs 0.184          & \cs 0.109          & \cs \textbf{0.65} & \cs \textbf{0.767} & \cs \textbf{0.715} & \cs \textbf{0.712} & \cs \textbf{0.565} & \cs 0.584          & \cs \textbf{0.717} & \cs \textbf{0.649} & \cs 0.567          & \cs 0.462         & \cs 0.164          & \cs 0.155          & \cs 0.152          & \cs 0.1            & \cs 0.044          & \cs 0.244          & \cs 0.264          & \cs 0.311          & \cs 0.259          & \cs 0.159          & \cs 0.37           & \cs 0.459         & \cs 0.449          & \cs 0.348          & \cs 0.194          & \cs 0.568          & \cs \textbf{0.682} & \cs 0.533          & \cs 0.276          & \cs 0.154          & \cs 0.47           & \cs \textbf{0.65} & \cs 0.542          & \cs 0.382          & \cs 0.17           \\
 \MAPLAm{}  &\cs 0.178         & \cs 0.174          & \cs 0.165          & \cs 0.126          & \cs 0.078          & \cs 0.351         & \cs 0.423          & \cs 0.422          & \cs 0.432          & \cs 0.375          & \cs 0.394          & \cs 0.476          & \cs 0.466          & \cs 0.366          & \cs 0.331         & \cs 0.211          & \cs 0.162          & \cs 0.142          & \cs 0.106          & \cs 0.042          & \cs 0.189          & \cs 0.197          & \cs 0.238          & \cs 0.198          & \cs 0.124          & \cs 0.358          & \cs 0.42          & \cs 0.385          & \cs 0.292          & \cs 0.159          & \cs 0.313          & \cs 0.517          & \cs 0.463          & \cs 0.241          & \cs 0.131          & \cs 0.287          & \cs 0.502         & \cs 0.47           & \cs 0.353          & \cs 0.159          \\
 \HRm{}     &\cs 0.084         & \cs 0.106          & \cs 0.104          & \cs 0.096          & \cs 0.074          & \cs 0.181         & \cs 0.266          & \cs 0.241          & \cs 0.259          & \cs 0.196          & \cs 0.201          & \cs 0.318          & \cs 0.262          & \cs 0.251          & \cs 0.165         & \cs 0.054          & \cs 0.097          & \cs 0.108          & \cs 0.098          & \cs 0.075          & \cs 0.092          & \cs 0.168          & \cs 0.25           & \cs 0.239          & \cs 0.179          & \cs 0.14           & \cs 0.244         & \cs 0.269          & \cs 0.233          & \cs 0.176          & \cs 0.146          & \cs 0.224          & \cs 0.205          & \cs 0.175          & \cs 0.132          & \cs 0.178          & \cs 0.258         & \cs 0.26           & \cs 0.222          & \cs 0.145          \\
 \LHRm{}    &\cs 0.11          & \cs 0.137          & \cs 0.153          & \cs 0.136          & \cs 0.096          & \cs 0.204         & \cs 0.249          & \cs 0.199          & \cs 0.228          & \cs 0.151          & \cs 0.198          & \cs 0.285          & \cs 0.271          & \cs 0.239          & \cs 0.131         & \cs 0.087          & \cs 0.145          & \cs 0.144          & \cs 0.116          & \cs 0.086          & \cs 0.102          & \cs 0.133          & \cs 0.208          & \cs 0.206          & \cs 0.137          & \cs 0.228          & \cs 0.256         & \cs 0.359          & \cs 0.266          & \cs 0.168          & \cs 0.184          & \cs 0.238          & \cs 0.339          & \cs 0.251          & \cs 0.194          & \cs 0.138          & \cs 0.241         & \cs 0.284          & \cs 0.221          & \cs 0.164          \\
 \smHRem{}  &\cs 0.23          & \cs 0.289          & \cs 0.346          & \cs 0.344          & \cs 0.369          & \cs 0.081         & \cs 0.057          & \cs 0.062          & \cs 0.07           & \cs 0.049          & \cs 0.189          & \cs 0.316          & \cs 0.265          & \cs 0.255          & \cs 0.163         & \cs 0.177          & \cs 0.249          & \cs 0.269          & \cs 0.202          & \cs 0.14           & \cs 0.104          & \cs 0.124          & \cs 0.162          & \cs 0.136          & \cs 0.103          & \cs 0.133          & \cs 0.241         & \cs 0.263          & \cs 0.24           & \cs 0.177          & \cs 0.223          & \cs 0.166          & \cs 0.134          & \cs 0.111          & \cs 0.059          & \cs 0.21           & \cs 0.269         & \cs 0.205          & \cs 0.169          & \cs 0.114          \\
 \smHRdm{}  &\cs 0.244         & \cs 0.32           & \cs 0.342          & \cs 0.331          & \cs 0.35           & \cs 0.166         & \cs 0.265          & \cs 0.239          & \cs 0.245          & \cs 0.195          & \cs 0.197          & \cs 0.306          & \cs 0.271          & \cs 0.254          & \cs 0.172         & \cs 0.149          & \cs 0.265          & \cs 0.327          & \cs 0.218          & \cs 0.138          & \cs 0.115          & \cs 0.189          & \cs 0.274          & \cs 0.227          & \cs 0.175          & \cs 0.13           & \cs 0.228         & \cs 0.291          & \cs 0.25           & \cs 0.178          & \cs 0.233          & \cs 0.214          & \cs 0.18           & \cs 0.13           & \cs 0.089          & \cs 0.403          & \cs 0.357         & \cs 0.252          & \cs 0.21           & \cs 0.129          \\
 \smLHRem{} &\cs 0.261         & \cs 0.286          & \cs 0.432          & \cs 0.391          & \cs 0.449          & \cs 0.029         & \cs 0.035          & \cs 0.045          & \cs 0.055          & \cs 0.035          & \cs 0.182          & \cs 0.29           & \cs 0.273          & \cs 0.224          & \cs 0.135         & \cs 0.352          & \cs 0.341          & \cs 0.442          & \cs 0.304          & \cs 0.197          & \cs 0.188          & \cs 0.199          & \cs 0.193          & \cs 0.174          & \cs 0.104          & \cs 0.234          & \cs 0.282         & \cs 0.353          & \cs 0.248          & \cs 0.18           & \cs 0.298          & \cs 0.18           & \cs 0.165          & \cs 0.143          & \cs 0.075          & \cs 0.204          & \cs 0.208         & \cs 0.216          & \cs 0.162          & \cs 0.1            \\
 \smLHRdm{} &\cs 0.305         & \cs 0.277          & \cs 0.3            & \cs 0.195          & \cs 0.155          & \cs 0.19          & \cs 0.232          & \cs 0.193          & \cs 0.221          & \cs 0.154          & \cs 0.192          & \cs 0.253          & \cs 0.269          & \cs 0.236          & \cs 0.138         & \cs 0.334          & \cs 0.423          & \cs 0.334          & \cs 0.224          & \cs 0.123          & \cs 0.214          & \cs 0.289          & \cs 0.366          & \cs 0.281          & \cs 0.183          & \cs 0.225          & \cs 0.301         & \cs 0.354          & \cs 0.262          & \cs 0.184          & \cs 0.374          & \cs 0.28           & \cs 0.246          & \cs 0.173          & \cs 0.118          & \cs 0.288          & \cs 0.232         & \cs 0.286          & \cs 0.247          & \cs 0.153          \\
\hline
\end{tabular}};
    \end{tikzpicture}
    }
\end{table}

\begin{table}[]
    \centering
    \caption{
        Average relative performance based on \minESS{}/s of the tested samplers across polytope angles as presented as solid lines in \cref{fig:esst1}.
    }
    \resizebox{\linewidth}{!}{
    \begin{tikzpicture}
        \node at (0, 0) {%
            \begin{tabular}{r|
                S[table-format=1.4]S[table-format=1.4]S[table-format=1.4]S[table-format=1.4]|
                S[table-format=1.4]S[table-format=1.4]S[table-format=1.4]S[table-format=1.4]|
                S[table-format=1.4]S[table-format=1.4]S[table-format=1.4]S[table-format=1.4]|
                S[table-format=1.4]S[table-format=1.4]S[table-format=1.4]S[table-format=1.4]|
                S[table-format=1.4]S[table-format=1.4]S[table-format=1.4]S[table-format=1.4]|
                S[table-format=1.4]S[table-format=1.4]S[table-format=1.4]S[table-format=1.4]|
                S[table-format=1.4]S[table-format=1.4]S[table-format=1.4]S[table-format=1.4]|
                S[table-format=1.4]S[table-format=1.4]S[table-format=1.4]S[table-format=1.4]
            }
            \hline
            \multirow{2}{*}{Sampler ${\vcenter{\hbox{\begin{tikzpicture}[x=1mm,y=3.5mm]
                \draw[-] (1,-1) -- (-1,1);
            \end{tikzpicture}}}}$ \!\!\!\!\!\! $\begin{array}{c}\text{Target} \\\theta\end{array}$ \!\!\!\!\!\! }
                & \multicolumn{4}{c|}{Gauss $(\mu = 0)$}
                & \multicolumn{4}{c|}{Disc $(\mu = 0)$}
                & \multicolumn{4}{c|}{Cigar $(\mu = 0)$}
                & \multicolumn{4}{c|}{Gauss $(\mu = 0.5)$}
                & \multicolumn{4}{c|}{Disc $(\mu = 0.5)$}
                & \multicolumn{4}{c|}{Cigar $(\mu = 0.5)$}
                & \multicolumn{4}{c|}{Bowtie}
                & \multicolumn{4}{c}{Funnel}
            \\
            & \multicolumn{1}{c}{\footnotesize $9^{\circ}$} & \multicolumn{1}{c}{\footnotesize $19^{\circ}$} & \multicolumn{1}{c}{\footnotesize $45^{\circ}$} & \multicolumn{1}{c|}{\footnotesize $90^{\circ}$}
            & \multicolumn{1}{c}{\footnotesize $9^{\circ}$} & \multicolumn{1}{c}{\footnotesize $19^{\circ}$} & \multicolumn{1}{c}{\footnotesize $45^{\circ}$} & \multicolumn{1}{c|}{\footnotesize $90^{\circ}$}
            & \multicolumn{1}{c}{\footnotesize $9^{\circ}$} & \multicolumn{1}{c}{\footnotesize $19^{\circ}$} & \multicolumn{1}{c}{\footnotesize $45^{\circ}$} & \multicolumn{1}{c|}{\footnotesize $90^{\circ}$}
            & \multicolumn{1}{c}{\footnotesize $9^{\circ}$} & \multicolumn{1}{c}{\footnotesize $19^{\circ}$} & \multicolumn{1}{c}{\footnotesize $45^{\circ}$} & \multicolumn{1}{c|}{\footnotesize $90^{\circ}$}
            & \multicolumn{1}{c}{\footnotesize $9^{\circ}$} & \multicolumn{1}{c}{\footnotesize $19^{\circ}$} & \multicolumn{1}{c}{\footnotesize $45^{\circ}$} & \multicolumn{1}{c|}{\footnotesize $90^{\circ}$}
            & \multicolumn{1}{c}{\footnotesize $9^{\circ}$} & \multicolumn{1}{c}{\footnotesize $19^{\circ}$} & \multicolumn{1}{c}{\footnotesize $45^{\circ}$} & \multicolumn{1}{c|}{\footnotesize $90^{\circ}$}
            & \multicolumn{1}{c}{\footnotesize $9^{\circ}$} & \multicolumn{1}{c}{\footnotesize $19^{\circ}$} & \multicolumn{1}{c}{\footnotesize $45^{\circ}$} & \multicolumn{1}{c|}{\footnotesize $90^{\circ}$}
            & \multicolumn{1}{c}{\footnotesize $9^{\circ}$} & \multicolumn{1}{c}{\footnotesize $19^{\circ}$} & \multicolumn{1}{c}{\footnotesize $45^{\circ}$} & \multicolumn{1}{c }{\footnotesize $90^{\circ}$}
            \\
            \hline
  \RWMHm{}   & \cs 0.01           & \cs 0.057          & \cs 0.102          & \cs 0.213          & \cs 0.031          & \cs 0.106          & \cs 0.413          & \cs 0.535          & \cs 0.033          & \cs 0.248          & \cs 0.435          & \cs 0.496         & \cs 0.013          & \cs 0.078          & \cs 0.119         & \cs 0.192          & \cs 0.198          & \cs 0.343          & \cs 0.287          & \cs 0.247          & \cs 0.124          & \cs 0.377          & \cs 0.312          & \cs 0.318          & \cs 0.185          & \cs 0.334          & \cs 0.272          & \cs 0.398          & \cs 0.149          & \cs 0.337          & \cs 0.366          & \cs 0.503          \\
  \MALAm{}   & \cs 0.012          & \cs 0.068          & \cs 0.12           & \cs 0.267          & \cs 0.035          & \cs 0.156          & \cs 0.547          & \cs 0.647          & \cs 0.036          & \cs 0.334          & \cs 0.496          & \cs 0.668         & \cs 0.017          & \cs 0.087          & \cs 0.152         & \cs 0.266          & \cs 0.247          & \cs 0.398          & \cs 0.335          & \cs 0.265          & \cs 0.268          & \cs 0.589          & \cs 0.587          & \cs 0.619          & \cs 0.29           & \cs 0.522          & \cs 0.452          & \cs 0.642          & \cs 0.198          & \cs 0.429          & \cs 0.461          & \cs 0.516          \\
  \smMALAem{}& \cs 0.41           & \cs 0.628          & \cs 0.482          & \cs 0.473          & \cs 0.002          & \cs 0.007          & \cs 0.024          & \cs 0.262          & \cs 0.043          & \cs 0.327          & \cs 0.495          & \cs 0.644         & \cs 0.596          & \cs 0.668          & \cs 0.664         & \cs 0.597          & \cs 0.22           & \cs 0.317          & \cs 0.462          & \cs 0.632          & \cs 0.266          & \cs 0.575          & \cs 0.58           & \cs 0.612          & \cs 0.124          & \cs 0.231          & \cs 0.207          & \cs 0.268          & \cs 0.147          & \cs 0.312          & \cs 0.351          & \cs 0.342          \\
  \smMALAdm{}& \cs 0.099          & \cs 0.317          & \cs 0.457          & \cs 0.685          & \cs 0.044          & \cs 0.166          & \cs 0.532          & \cs 0.657          & \cs 0.041          & \cs 0.328          & \cs 0.521          & \cs 0.663         & \cs 0.36           & \cs 0.479          & \cs 0.601         & \cs 0.75           & \cs 0.369          & \cs 0.59           & \cs 0.617          & \cs 0.612          & \cs 0.283          & \cs 0.578          & \cs 0.59           & \cs 0.631          & \cs 0.26           & \cs 0.437          & \cs 0.421          & \cs 0.561          & \cs 0.242          & \cs 0.52           & \cs 0.539          & \cs 0.654          \\
  \Dikinm{}  & \cs 0.309          & \cs 0.165          & \cs 0.103          & \cs 0.123          & \cs \textbf{0.978} & \cs \textbf{0.954} & \cs 0.573          & \cs 0.377          & \cs \textbf{0.946} & \cs \textbf{0.866} & \cs 0.446          & \cs 0.257         & \cs 0.168          & \cs 0.119          & \cs 0.089         & \cs 0.119          & \cs 0.496          & \cs 0.378          & \cs 0.169          & \cs 0.171          & \cs \textbf{0.715} & \cs 0.524          & \cs 0.213          & \cs 0.23           & \cs \textbf{0.652} & \cs 0.531          & \cs 0.336          & \cs 0.312          & \cs \textbf{0.697} & \cs 0.53           & \cs 0.336          & \cs 0.316          \\
  \MAPLAm{}  & \cs 0.232          & \cs 0.124          & \cs 0.081          & \cs 0.126          & \cs 0.647          & \cs 0.581          & \cs 0.337          & \cs 0.34           & \cs 0.728          & \cs 0.594          & \cs 0.298          & \cs 0.252         & \cs 0.287          & \cs 0.165          & \cs 0.091         & \cs 0.126          & \cs 0.431          & \cs 0.299          & \cs 0.129          & \cs 0.178          & \cs 0.631          & \cs 0.472          & \cs 0.243          & \cs 0.373          & \cs 0.477          & \cs 0.425          & \cs 0.381          & \cs 0.374          & \cs 0.633          & \cs 0.51           & \cs 0.309          & \cs 0.273          \\
  \HRm{}     & \cs 0.048          & \cs 0.113          & \cs 0.243          & \cs 0.299          & \cs 0.15           & \cs 0.34           & \cs \textbf{0.871} & \cs 0.786          & \cs 0.175          & \cs 0.495          & \cs 0.831          & \cs 0.809         & \cs 0.03           & \cs 0.116          & \cs 0.248         & \cs 0.279          & \cs 0.346          & \cs 0.542          & \cs 0.532          & \cs 0.401          & \cs 0.29           & \cs 0.61           & \cs 0.56           & \cs 0.483          & \cs 0.235          & \cs 0.432          & \cs 0.427          & \cs 0.513          & \cs 0.233          & \cs 0.458          & \cs 0.582          & \cs 0.641          \\
  \LHRm{}    & \cs 0.035          & \cs 0.161          & \cs 0.433          & \cs 0.41           & \cs 0.128          & \cs 0.363          & \cs 0.855          & \cs 0.798          & \cs 0.141          & \cs 0.503          & \cs 0.837          & \cs \textbf{0.86} & \cs 0.043          & \cs 0.171          & \cs 0.358         & \cs 0.377          & \cs 0.346          & \cs 0.525          & \cs 0.511          & \cs 0.353          & \cs 0.406          & \cs 0.753          & \cs \textbf{0.842} & \cs 0.773          & \cs 0.421          & \cs 0.654          & \cs \textbf{0.733} & \cs \textbf{0.728} & \cs 0.292          & \cs 0.493          & \cs 0.666          & \cs 0.655          \\
  \smHRem{}  & \cs 0.737          & \cs 0.625          & \cs 0.39           & \cs 0.359          & \cs 0.013          & \cs 0.029          & \cs 0.075          & \cs 0.608          & \cs 0.173          & \cs 0.502          & \cs 0.829          & \cs 0.8           & \cs 0.478          & \cs 0.446          & \cs 0.337         & \cs 0.268          & \cs 0.161          & \cs 0.23           & \cs 0.333          & \cs 0.494          & \cs 0.295          & \cs 0.583          & \cs 0.57           & \cs 0.49           & \cs 0.213          & \cs 0.362          & \cs 0.386          & \cs 0.367          & \cs 0.274          & \cs 0.498          & \cs 0.565          & \cs 0.541          \\
  \smHRdm{}  & \cs 0.804          & \cs 0.572          & \cs 0.311          & \cs 0.35           & \cs 0.129          & \cs 0.338          & \cs 0.758          & \cs \textbf{0.877} & \cs 0.192          & \cs 0.512          & \cs 0.858          & \cs 0.763         & \cs 0.459          & \cs 0.407          & \cs 0.321         & \cs 0.317          & \cs 0.346          & \cs 0.514          & \cs 0.528          & \cs 0.524          & \cs 0.297          & \cs 0.6            & \cs 0.566          & \cs 0.508          & \cs 0.333          & \cs 0.506          & \cs 0.435          & \cs 0.517          & \cs 0.432          & \cs \textbf{0.619} & \cs 0.73           & \cs 0.694          \\
  \smLHRem{} & \cs \textbf{0.872} & \cs \textbf{0.862} & \cs 0.548          & \cs 0.575          & \cs 0.005          & \cs 0.013          & \cs 0.05           & \cs 0.441          & \cs 0.138          & \cs 0.479          & \cs \textbf{0.864} & \cs 0.818         & \cs \textbf{0.933} & \cs \textbf{0.905} & \cs 0.732         & \cs 0.659          & \cs 0.283          & \cs 0.398          & \cs 0.531          & \cs 0.712          & \cs 0.386          & \cs \textbf{0.777} & \cs 0.838          & \cs \textbf{0.809} & \cs 0.307          & \cs 0.469          & \cs 0.446          & \cs 0.382          & \cs 0.299          & \cs 0.481          & \cs 0.579          & \cs 0.5            \\
  \smLHRdm{} & \cs 0.204          & \cs 0.493          & \cs \textbf{0.874} & \cs \textbf{0.717} & \cs 0.129          & \cs 0.342          & \cs 0.85           & \cs 0.792          & \cs 0.16           & \cs 0.498          & \cs 0.836          & \cs 0.827         & \cs 0.542          & \cs 0.695          & \cs \textbf{0.77} & \cs \textbf{0.786} & \cs \textbf{0.528} & \cs \textbf{0.748} & \cs \textbf{0.829} & \cs \textbf{0.766} & \cs 0.423          & \cs 0.756          & \cs 0.839          & \cs 0.808          & \cs 0.457          & \cs \textbf{0.671} & \cs 0.656          & \cs 0.699          & \cs 0.4            & \cs 0.606          & \cs \textbf{0.821} & \cs \textbf{0.752} \\
 \hline
\end{tabular}};
    \end{tikzpicture}
    }
\end{table}

\begin{table}[]
    \centering
    \caption{
        Average relative performance based on \minESS{}/s of the tested samplers across densities' scale parameter as presented as solid lines in \cref{fig:esst1}.
    }
    \resizebox{.9\linewidth}{!}{
    \begin{tikzpicture}
        \node at (0, 0) {%
            \begin{tabular}{r|
                S[table-format=1.4]S[table-format=1.4]S[table-format=1.4]S[table-format=1.4]S[table-format=1.4]S[table-format=1.4]S[table-format=1.4]|
                S[table-format=1.4]S[table-format=1.4]S[table-format=1.4]S[table-format=1.4]S[table-format=1.4]S[table-format=1.4]S[table-format=1.4]|
                S[table-format=1.4]S[table-format=1.4]S[table-format=1.4]S[table-format=1.4]S[table-format=1.4]S[table-format=1.4]S[table-format=1.4]|
                S[table-format=1.4]S[table-format=1.4]S[table-format=1.4]S[table-format=1.4]S[table-format=1.4]S[table-format=1.4]S[table-format=1.4]
            }
            \hline
            \multirow{2}{*}{Sampler ${\vcenter{\hbox{\begin{tikzpicture}[x=1mm,y=3.5mm]
                \draw[-] (1,-1) -- (-1,1);
            \end{tikzpicture}}}}$ \!\!\!\!\!\! $\begin{array}{c}\text{Target} \\\sigma\end{array}$ \!\!\!\!\!\! }
                & \multicolumn{7}{c|}{Gauss $(\mu = 0)$}
                & \multicolumn{7}{c|}{Disc $(\mu = 0)$}
                & \multicolumn{7}{c|}{Cigar $(\mu = 0)$}
                & \multicolumn{7}{c }{Gauss $(\mu = 0.5)$}
            \\
            & \multicolumn{1}{c}{\footnotesize $10^{-2}$} & \multicolumn{1}{c}{\footnotesize $10^{-1.5}$} & \multicolumn{1}{c}{\footnotesize $10^{-1}$} & \multicolumn{1}{c}{\footnotesize $10^{-0.5}$} & \multicolumn{1}{c}{\footnotesize $10^{0}$} & \multicolumn{1}{c}{\footnotesize $10^{0.5}$} & \multicolumn{1}{c|}{\footnotesize $10^{1}$}
            & \multicolumn{1}{c}{\footnotesize $10^{-2}$} & \multicolumn{1}{c}{\footnotesize $10^{-1.5}$} & \multicolumn{1}{c}{\footnotesize $10^{-1}$} & \multicolumn{1}{c}{\footnotesize $10^{-0.5}$} & \multicolumn{1}{c}{\footnotesize $10^{0}$} & \multicolumn{1}{c}{\footnotesize $10^{0.5}$} & \multicolumn{1}{c|}{\footnotesize $10^{1}$}
            & \multicolumn{1}{c}{\footnotesize $10^{-2}$} & \multicolumn{1}{c}{\footnotesize $10^{-1.5}$} & \multicolumn{1}{c}{\footnotesize $10^{-1}$} & \multicolumn{1}{c}{\footnotesize $10^{-0.5}$} & \multicolumn{1}{c}{\footnotesize $10^{0}$} & \multicolumn{1}{c}{\footnotesize $10^{0.5}$} & \multicolumn{1}{c|}{\footnotesize $10^{1}$}
            & \multicolumn{1}{c}{\footnotesize $10^{-2}$} & \multicolumn{1}{c}{\footnotesize $10^{-1.5}$} & \multicolumn{1}{c}{\footnotesize $10^{-1}$} & \multicolumn{1}{c}{\footnotesize $10^{-0.5}$} & \multicolumn{1}{c}{\footnotesize $10^{0}$} & \multicolumn{1}{c}{\footnotesize $10^{0.5}$} & \multicolumn{1}{c }{\footnotesize $10^{1}$}
            \\
            \hline
  \RWMHm{}   &  \cs 0.034          & \cs 0.035          & \cs 0.033          & \cs 0.097          & \cs 0.212          & \cs 0.49           & \cs \textbf{0.547} & \cs 0.452          & \cs 0.473          & \cs 0.49           & \cs 0.464         & \cs 0.471          & \cs 0.425          & \cs 0.555          & \cs 0.453          & \cs 0.437          & \cs 0.447          & \cs 0.486          & \cs 0.655          & \cs \textbf{0.667} & \cs \textbf{0.646} & \cs 0.012          & \cs 0.013          & \cs 0.032          & \cs 0.117          & \cs 0.291         & \cs 0.498         & \cs 0.556           \\
  \MALAm{}   &  \cs 0.071          & \cs 0.055          & \cs 0.05           & \cs 0.146          & \cs 0.286          & \cs 0.494          & \cs 0.538          & \cs 0.546          & \cs 0.625          & \cs 0.555          & \cs 0.532         & \cs 0.559          & \cs 0.52           & \cs 0.532          & \cs 0.588          & \cs 0.618          & \cs 0.589          & \cs \textbf{0.605} & \cs \textbf{0.672} & \cs 0.654          & \cs 0.636          & \cs 0.037          & \cs 0.032          & \cs 0.082          & \cs 0.199          & \cs 0.373         & \cs \textbf{0.55} & \cs 0.539           \\
  \smMALAem{}&  \cs \textbf{0.848} & \cs \textbf{0.817} & \cs \textbf{0.733} & \cs \textbf{0.701} & \cs \textbf{0.676} & \cs \textbf{0.555} & \cs 0.524          & \cs 0.126          & \cs 0.077          & \cs 0.057          & \cs 0.048         & \cs 0.06           & \cs 0.106          & \cs 0.129          & \cs 0.449          & \cs 0.399          & \cs 0.417          & \cs 0.393          & \cs 0.472          & \cs 0.43           & \cs 0.439          & \cs \textbf{0.969} & \cs \textbf{0.956} & \cs \textbf{0.941} & \cs \textbf{0.812} & \cs \textbf{0.73} & \cs 0.549         & \cs \textbf{0.581}  \\
  \smMALAdm{}&  \cs 0.497          & \cs 0.584          & \cs 0.55           & \cs 0.568          & \cs 0.354          & \cs 0.366          & \cs 0.345          & \cs 0.396          & \cs 0.383          & \cs 0.409          & \cs 0.391         & \cs 0.403          & \cs 0.374          & \cs 0.366          & \cs 0.431          & \cs 0.414          & \cs 0.443          & \cs 0.433          & \cs 0.485          & \cs 0.459          & \cs 0.445          & \cs 0.905          & \cs 0.837          & \cs 0.779          & \cs 0.648          & \cs 0.38          & \cs 0.379         & \cs 0.384           \\
  \Dikinm{}  &  \cs 0.141          & \cs 0.188          & \cs 0.167          & \cs 0.173          & \cs 0.201          & \cs 0.251          & \cs 0.259          & \cs \textbf{0.674} & \cs \textbf{0.678} & \cs \textbf{0.719} & \cs \textbf{0.71} & \cs \textbf{0.721} & \cs \textbf{0.667} & \cs \textbf{0.605} & \cs \textbf{0.627} & \cs \textbf{0.641} & \cs \textbf{0.624} & \cs 0.588          & \cs 0.569          & \cs 0.563          & \cs 0.557          & \cs 0.056          & \cs 0.058          & \cs 0.058          & \cs 0.073          & \cs 0.14          & \cs 0.227         & \cs 0.252           \\
  \MAPLAm{}  &  \cs 0.082          & \cs 0.104          & \cs 0.112          & \cs 0.114          & \cs 0.158          & \cs 0.214          & \cs 0.225          & \cs 0.355          & \cs 0.345          & \cs 0.369          & \cs 0.363         & \cs 0.402          & \cs 0.472          & \cs 0.498          & \cs 0.332          & \cs 0.353          & \cs 0.339          & \cs 0.403          & \cs 0.462          & \cs 0.491          & \cs 0.466          & \cs 0.134          & \cs 0.127          & \cs 0.111          & \cs 0.061          & \cs 0.108         & \cs 0.182         & \cs 0.206           \\
  \HRm{}     &  \cs 0.012          & \cs 0.017          & \cs 0.018          & \cs 0.036          & \cs 0.101          & \cs 0.218          & \cs 0.249          & \cs 0.183          & \cs 0.183          & \cs 0.216          & \cs 0.219         & \cs 0.229          & \cs 0.287          & \cs 0.283          & \cs 0.163          & \cs 0.176          & \cs 0.194          & \cs 0.223          & \cs 0.316          & \cs 0.313          & \cs 0.291          & \cs 0.002          & \cs 0.003          & \cs 0.007          & \cs 0.036          & \cs 0.102         & \cs 0.197         & \cs 0.259           \\
  \LHRm{}    &  \cs 0.03           & \cs 0.029          & \cs 0.034          & \cs 0.069          & \cs 0.157          & \cs 0.272          & \cs 0.295          & \cs 0.18           & \cs 0.19           & \cs 0.214          & \cs 0.208         & \cs 0.185          & \cs 0.246          & \cs 0.223          & \cs 0.173          & \cs 0.173          & \cs 0.216          & \cs 0.212          & \cs 0.255          & \cs 0.246          & \cs 0.301          & \cs 0.006          & \cs 0.006          & \cs 0.022          & \cs 0.067          & \cs 0.172         & \cs 0.273         & \cs 0.263           \\
  \smHRem{}  &  \cs 0.18           & \cs 0.238          & \cs 0.26           & \cs 0.26           & \cs 0.405          & \cs 0.426          & \cs 0.439          & \cs 0.066          & \cs 0.056          & \cs 0.053          & \cs 0.041         & \cs 0.058          & \cs 0.081          & \cs 0.089          & \cs 0.16           & \cs 0.166          & \cs 0.196          & \cs 0.215          & \cs 0.312          & \cs 0.301          & \cs 0.315          & \cs 0.065          & \cs 0.067          & \cs 0.097          & \cs 0.157          & \cs 0.308         & \cs 0.371         & \cs 0.388           \\
  \smHRdm{}  &  \cs 0.153          & \cs 0.183          & \cs 0.218          & \cs 0.233          & \cs 0.398          & \cs 0.511          & \cs 0.524          & \cs 0.175          & \cs 0.183          & \cs 0.207          & \cs 0.192         & \cs 0.245          & \cs 0.267          & \cs 0.284          & \cs 0.167          & \cs 0.158          & \cs 0.2            & \cs 0.222          & \cs 0.324          & \cs 0.306          & \cs 0.303          & \cs 0.045          & \cs 0.047          & \cs 0.082          & \cs 0.16           & \cs 0.294         & \cs 0.421         & \cs 0.485           \\
  \smLHRem{} &  \cs 0.299          & \cs 0.358          & \cs 0.404          & \cs 0.324          & \cs 0.391          & \cs 0.395          & \cs 0.376          & \cs 0.039          & \cs 0.036          & \cs 0.032          & \cs 0.026         & \cs 0.03           & \cs 0.055          & \cs 0.061          & \cs 0.177          & \cs 0.181          & \cs 0.22           & \cs 0.181          & \cs 0.251          & \cs 0.247          & \cs 0.29           & \cs 0.263          & \cs 0.263          & \cs 0.29           & \cs 0.315          & \cs 0.385         & \cs 0.414         & \cs 0.359           \\
  \smLHRdm{} &  \cs 0.205          & \cs 0.19           & \cs 0.283          & \cs 0.21           & \cs 0.244          & \cs 0.305          & \cs 0.287          & \cs 0.168          & \cs 0.172          & \cs 0.204          & \cs 0.182         & \cs 0.202          & \cs 0.236          & \cs 0.223          & \cs 0.189          & \cs 0.151          & \cs 0.204          & \cs 0.206          & \cs 0.243          & \cs 0.248          & \cs 0.28           & \cs 0.216          & \cs 0.237          & \cs 0.248          & \cs 0.291          & \cs 0.323         & \cs 0.363         & \cs 0.333           \\
            \hline
\end{tabular}};
    \end{tikzpicture}
    }
\end{table}

\begin{table}[]
    \centering
    \caption{
        Average relative performance based on \minESS{}/s of the tested samplers across densities' scale parameter as presented as solid lines in \cref{fig:esst1}.
    }
    \resizebox{.9\linewidth}{!}{
    \begin{tikzpicture}
        \node at (0, 0) {%
            \begin{tabular}{r|
                S[table-format=1.4]S[table-format=1.4]S[table-format=1.4]S[table-format=1.4]S[table-format=1.4]S[table-format=1.4]S[table-format=1.4]|
                S[table-format=1.4]S[table-format=1.4]S[table-format=1.4]S[table-format=1.4]S[table-format=1.4]S[table-format=1.4]S[table-format=1.4]|
                S[table-format=1.4]S[table-format=1.4]S[table-format=1.4]S[table-format=1.4]S[table-format=1.4]S[table-format=1.4]S[table-format=1.4]|
                S[table-format=1.4]S[table-format=1.4]S[table-format=1.4]S[table-format=1.4]S[table-format=1.4]S[table-format=1.4]S[table-format=1.4]
            }
            \hline
            \multirow{2}{*}{Sampler ${\vcenter{\hbox{\begin{tikzpicture}[x=1mm,y=3.5mm]
                \draw[-] (1,-1) -- (-1,1);
            \end{tikzpicture}}}}$ \!\!\!\!\!\! $\begin{array}{c}\text{Target} \\\sigma\end{array}$ \!\!\!\!\!\! }
                & \multicolumn{7}{c|}{Disc $(\mu = 0.5)$}
                & \multicolumn{7}{c|}{Cigar $(\mu = 0.5)$}
                & \multicolumn{7}{c|}{Bowtie}
                & \multicolumn{7}{c}{Funnel}
            \\
            & \multicolumn{1}{c}{\footnotesize $10^{-2}$} & \multicolumn{1}{c}{\footnotesize $10^{-1.5}$} & \multicolumn{1}{c}{\footnotesize $10^{-1}$} & \multicolumn{1}{c}{\footnotesize $10^{-0.5}$} & \multicolumn{1}{c}{\footnotesize $10^{0}$} & \multicolumn{1}{c}{\footnotesize $10^{0.5}$} & \multicolumn{1}{c|}{\footnotesize $10^{1}$}
            & \multicolumn{1}{c}{\footnotesize $10^{-2}$} & \multicolumn{1}{c}{\footnotesize $10^{-1.5}$} & \multicolumn{1}{c}{\footnotesize $10^{-1}$} & \multicolumn{1}{c}{\footnotesize $10^{-0.5}$} & \multicolumn{1}{c}{\footnotesize $10^{0}$} & \multicolumn{1}{c}{\footnotesize $10^{0.5}$} & \multicolumn{1}{c|}{\footnotesize $10^{1}$}
            & \multicolumn{1}{c}{\footnotesize $10^{-2}$} & \multicolumn{1}{c}{\footnotesize $10^{-1.5}$} & \multicolumn{1}{c}{\footnotesize $10^{-1}$} & \multicolumn{1}{c}{\footnotesize $10^{-0.5}$} & \multicolumn{1}{c}{\footnotesize $10^{0}$} & \multicolumn{1}{c}{\footnotesize $10^{0.5}$} & \multicolumn{1}{c }{\footnotesize $10^{1}$}
            \\
            \hline
  \RWMHm{}   & \cs 0.174          & \cs 0.268          & \cs 0.446          & \cs 0.541          & \cs 0.648         & \cs 0.685          & \cs \textbf{0.73} & \cs 0.223          & \cs 0.293          & \cs 0.466          & \cs 0.646          & \cs 0.724          & \cs \textbf{0.749} & \cs \textbf{0.738} & \cs 0.48           & \cs 0.524          & \cs 0.57           & \cs 0.404          & \cs 0.544          & \cs 0.654          & \cs \textbf{0.706} & \cs 0.609          & \cs 0.548         & \cs 0.496          & \cs 0.581          & \cs \textbf{0.682} & \cs \textbf{0.72} & \cs \textbf{0.726} \\
  \MALAm{}   & \cs 0.167          & \cs 0.251          & \cs 0.51           & \cs \textbf{0.668} & \cs \textbf{0.76} & \cs \textbf{0.748} & \cs 0.709         & \cs \textbf{0.982} & \cs \textbf{0.867} & \cs \textbf{0.792} & \cs \textbf{0.761} & \cs \textbf{0.733} & \cs 0.734          & \cs 0.71           & \cs \textbf{0.786} & \cs \textbf{0.824} & \cs \textbf{0.833} & \cs \textbf{0.695} & \cs \textbf{0.694} & \cs \textbf{0.693} & \cs 0.661          & \cs \textbf{0.668} & \cs \textbf{0.59} & \cs \textbf{0.687} & \cs \textbf{0.721} & \cs 0.68           & \cs 0.684         & \cs 0.704          \\
  \smMALAem{}& \cs \textbf{0.915} & \cs \textbf{0.824} & \cs 0.627          & \cs 0.414          & \cs 0.266         & \cs 0.194          & \cs 0.169         & \cs 0.703          & \cs 0.582          & \cs 0.551          & \cs 0.515          & \cs 0.513          & \cs 0.485          & \cs 0.494          & \cs 0.057          & \cs 0.08           & \cs 0.126          & \cs 0.299          & \cs 0.351          & \cs 0.288          & \cs 0.396          & \cs 0.098          & \cs 0.133         & \cs 0.274          & \cs 0.305          & \cs 0.383          & \cs 0.362         & \cs 0.32           \\
  \smMALAdm{}& \cs 0.743          & \cs 0.784          & \cs \textbf{0.686} & \cs 0.639          & \cs 0.59          & \cs 0.547          & \cs 0.495         & \cs 0.713          & \cs 0.638          & \cs 0.614          & \cs 0.536          & \cs 0.502          & \cs 0.518          & \cs 0.5            & \cs 0.408          & \cs 0.372          & \cs 0.385          & \cs 0.42           & \cs 0.332          & \cs 0.316          & \cs 0.313          & \cs 0.629          & \cs 0.566         & \cs 0.452          & \cs 0.407          & \cs 0.323          & \cs 0.352         & \cs 0.326          \\
  \Dikinm{}  & \cs 0.015          & \cs 0.021          & \cs 0.087          & \cs 0.2            & \cs 0.398         & \cs 0.468          & \cs 0.542         & \cs 0.093          & \cs 0.174          & \cs 0.329          & \cs 0.443          & \cs 0.513          & \cs 0.505          & \cs 0.491          & \cs 0.507          & \cs 0.436          & \cs 0.483          & \cs 0.423          & \cs 0.367          & \cs 0.428          & \cs 0.456          & \cs 0.361          & \cs 0.439         & \cs 0.343          & \cs 0.493          & \cs 0.496          & \cs 0.483         & \cs 0.484          \\
  \MAPLAm{}  & \cs 0.008          & \cs 0.013          & \cs 0.059          & \cs 0.128          & \cs 0.289         & \cs 0.412          & \cs 0.417         & \cs 0.148          & \cs 0.189          & \cs 0.274          & \cs 0.365          & \cs 0.417          & \cs 0.427          & \cs 0.439          & \cs 0.284          & \cs 0.298          & \cs 0.37           & \cs 0.344          & \cs 0.285          & \cs 0.364          & \cs 0.387          & \cs 0.255          & \cs 0.33          & \cs 0.311          & \cs 0.424          & \cs 0.392          & \cs 0.381         & \cs 0.388          \\
  \HRm{}     & \cs 0.037          & \cs 0.076          & \cs 0.139          & \cs 0.189          & \cs 0.245         & \cs 0.284          & \cs 0.329         & \cs 0.059          & \cs 0.081          & \cs 0.144          & \cs 0.251          & \cs 0.302          & \cs 0.342          & \cs 0.308          & \cs 0.111          & \cs 0.114          & \cs 0.144          & \cs 0.12           & \cs 0.185          & \cs 0.262          & \cs 0.298          & \cs 0.107          & \cs 0.116         & \cs 0.141          & \cs 0.241          & \cs 0.307          & \cs 0.285         & \cs 0.291          \\
  \LHRm{}    & \cs 0.025          & \cs 0.07           & \cs 0.118          & \cs 0.199          & \cs 0.201         & \cs 0.247          & \cs 0.24          & \cs 0.191          & \cs 0.184          & \cs 0.221          & \cs 0.287          & \cs 0.294          & \cs 0.297          & \cs 0.315          & \cs 0.151          & \cs 0.165          & \cs 0.233          & \cs 0.231          & \cs 0.253          & \cs 0.323          & \cs 0.333          & \cs 0.148          & \cs 0.133         & \cs 0.188          & \cs 0.294          & \cs 0.216          & \cs 0.24          & \cs 0.249          \\
  \smHRem{}  & \cs 0.086          & \cs 0.123          & \cs 0.149          & \cs 0.163          & \cs 0.156         & \cs 0.102          & \cs 0.102         & \cs 0.056          & \cs 0.083          & \cs 0.144          & \cs 0.254          & \cs 0.298          & \cs 0.322          & \cs 0.319          & \cs 0.032          & \cs 0.045          & \cs 0.096          & \cs 0.171          & \cs 0.196          & \cs 0.192          & \cs 0.239          & \cs 0.071          & \cs 0.09          & \cs 0.167          & \cs 0.193          & \cs 0.263          & \cs 0.282         & \cs 0.289          \\
  \smHRdm{}  & \cs 0.089          & \cs 0.118          & \cs 0.158          & \cs 0.186          & \cs 0.259         & \cs 0.272          & \cs 0.29          & \cs 0.064          & \cs 0.084          & \cs 0.152          & \cs 0.247          & \cs 0.331          & \cs 0.315          & \cs 0.317          & \cs 0.107          & \cs 0.091          & \cs 0.109          & \cs 0.11           & \cs 0.159          & \cs 0.277          & \cs 0.332          & \cs 0.19           & \cs 0.15          & \cs 0.188          & \cs 0.24           & \cs 0.322          & \cs 0.389         & \cs 0.415          \\
  \smLHRem{} & \cs 0.25           & \cs 0.245          & \cs 0.273          & \cs 0.171          & \cs 0.111         & \cs 0.07           & \cs 0.08          & \cs 0.187          & \cs 0.191          & \cs 0.219          & \cs 0.291          & \cs 0.291          & \cs 0.305          & \cs 0.33           & \cs 0.021          & \cs 0.044          & \cs 0.148          & \cs 0.188          & \cs 0.269          & \cs 0.268          & \cs 0.268          & \cs 0.087          & \cs 0.124         & \cs 0.181          & \cs 0.177          & \cs 0.195          & \cs 0.245         & \cs 0.237          \\
  \smLHRdm{} & \cs 0.26           & \cs 0.308          & \cs 0.301          & \cs 0.301          & \cs 0.222         & \cs 0.22           & \cs 0.255         & \cs 0.193          & \cs 0.202          & \cs 0.218          & \cs 0.26           & \cs 0.293          & \cs 0.342          & \cs 0.348          & \cs 0.151          & \cs 0.13           & \cs 0.151          & \cs 0.211          & \cs 0.322          & \cs 0.35           & \cs 0.352          & \cs 0.289          & \cs 0.247         & \cs 0.245          & \cs 0.259          & \cs 0.223          & \cs 0.202         & \cs 0.223          \\
            \hline
\end{tabular}};
    \end{tikzpicture}
    }
\end{table}

\end{landscape}
\newpage
\begin{landscape}

\begin{table}[]
    \centering
    \caption{
        Average L1 error of the tested samplers across dimensions as presented as solid lines in \cref{fig:meanerr1}.
    }
    \resizebox{\linewidth}{!}{
    \begin{tikzpicture}
        \node at (0, 0) {%
            \begin{tabular}{r|
                S[table-format=1.4]S[table-format=1.4]S[table-format=1.4]S[table-format=1.4]S[table-format=1.4]|
                S[table-format=1.4]S[table-format=1.4]S[table-format=1.4]S[table-format=1.4]S[table-format=1.4]|
                S[table-format=1.4]S[table-format=1.4]S[table-format=1.4]S[table-format=1.4]S[table-format=1.4]|
                S[table-format=1.4]S[table-format=1.4]S[table-format=1.4]S[table-format=1.4]S[table-format=1.4]|
                S[table-format=1.4]S[table-format=1.4]S[table-format=1.4]S[table-format=1.4]S[table-format=1.4]|
                S[table-format=1.4]S[table-format=1.4]S[table-format=1.4]S[table-format=1.4]S[table-format=1.4]|
                S[table-format=1.4]S[table-format=1.4]S[table-format=1.4]S[table-format=1.4]S[table-format=1.4]|
                S[table-format=1.4]S[table-format=1.4]S[table-format=1.4]S[table-format=1.4]S[table-format=1.4]
            }
            \hline
            \multirow{2}{*}{Sampler ${\vcenter{\hbox{\begin{tikzpicture}[x=1mm,y=3.5mm]
                \draw[-] (1,-1) -- (-1,1);
            \end{tikzpicture}}}}$ \!\!\!\!\!\! $\begin{array}{c}\text{Target} \\d\end{array}$ \!\!\!\!\!\! }
                & \multicolumn{5}{c|}{Gauss $(\mu = 0)$}
                & \multicolumn{5}{c|}{Disc $(\mu = 0)$}
                & \multicolumn{5}{c|}{Cigar $(\mu = 0)$}
                & \multicolumn{5}{c|}{Gauss $(\mu = 0.5)$}
                & \multicolumn{5}{c|}{Disc $(\mu = 0.5)$}
                & \multicolumn{5}{c|}{Cigar $(\mu = 0.5)$}
                & \multicolumn{5}{c|}{Bowtie}
                & \multicolumn{5}{c}{Funnel}
            \\
            & \multicolumn{1}{c}{\footnotesize $2$} & \multicolumn{1}{c}{\footnotesize $4$} & \multicolumn{1}{c}{\footnotesize $8$} & \multicolumn{1}{c}{\footnotesize $16$} & \multicolumn{1}{c|}{\footnotesize $32$}
            & \multicolumn{1}{c}{\footnotesize $2$} & \multicolumn{1}{c}{\footnotesize $4$} & \multicolumn{1}{c}{\footnotesize $8$} & \multicolumn{1}{c}{\footnotesize $16$} & \multicolumn{1}{c|}{\footnotesize $32$}
            & \multicolumn{1}{c}{\footnotesize $2$} & \multicolumn{1}{c}{\footnotesize $4$} & \multicolumn{1}{c}{\footnotesize $8$} & \multicolumn{1}{c}{\footnotesize $16$} & \multicolumn{1}{c|}{\footnotesize $32$}
            & \multicolumn{1}{c}{\footnotesize $2$} & \multicolumn{1}{c}{\footnotesize $4$} & \multicolumn{1}{c}{\footnotesize $8$} & \multicolumn{1}{c}{\footnotesize $16$} & \multicolumn{1}{c|}{\footnotesize $32$}
            & \multicolumn{1}{c}{\footnotesize $2$} & \multicolumn{1}{c}{\footnotesize $4$} & \multicolumn{1}{c}{\footnotesize $8$} & \multicolumn{1}{c}{\footnotesize $16$} & \multicolumn{1}{c|}{\footnotesize $32$}
            & \multicolumn{1}{c}{\footnotesize $2$} & \multicolumn{1}{c}{\footnotesize $4$} & \multicolumn{1}{c}{\footnotesize $8$} & \multicolumn{1}{c}{\footnotesize $16$} & \multicolumn{1}{c|}{\footnotesize $32$}
            & \multicolumn{1}{c}{\footnotesize $2$} & \multicolumn{1}{c}{\footnotesize $4$} & \multicolumn{1}{c}{\footnotesize $8$} & \multicolumn{1}{c}{\footnotesize $16$} & \multicolumn{1}{c|}{\footnotesize $32$}
            & \multicolumn{1}{c}{\footnotesize $2$} & \multicolumn{1}{c}{\footnotesize $4$} & \multicolumn{1}{c}{\footnotesize $8$} & \multicolumn{1}{c}{\footnotesize $16$} & \multicolumn{1}{c }{\footnotesize $32$}
            \\
            \hline
 \RWMHm{}   &  \cs 0.039          & \cs 0.046          & \cs 0.111          & \cs 0.221          & \cs 0.318          & \cs 0.036          & \cs 0.042          & \cs 0.123          & \cs 0.247          & \cs 0.335          & \cs 0.045          & \cs 0.048          & \cs 0.076          & \cs 0.162          & \cs 0.25           & \cs 0.04           & \cs 0.043         & \cs 0.041          & \cs 0.045          & \cs 0.049          & \cs 0.042          & \cs 0.05           & \cs 0.047          & \cs 0.048          & \cs 0.052          & \cs 0.05           & \cs 0.048          & \cs 0.048          & \cs 0.049          & \cs 0.054         & \cs 0.039          & \cs 0.039          & \cs 0.04           & \cs 0.041          & \cs 0.049          & \cs 0.047          & \cs 0.054          & \cs 0.056          & \cs 0.06           & \cs 0.069          \\
 \MALAm{}   &  \cs 0.037          & \cs 0.043          & \cs 0.108          & \cs 0.22           & \cs 0.317          & \cs 0.036          & \cs 0.041          & \cs 0.123          & \cs 0.25           & \cs 0.332          & \cs 0.045          & \cs 0.048          & \cs 0.076          & \cs 0.16           & \cs 0.251          & \cs 0.039          & \cs 0.038         & \cs 0.036          & \cs 0.038          & \cs 0.043          & \cs 0.041          & \cs 0.051          & \cs 0.047          & \cs 0.048          & \cs 0.052          & \cs 0.049          & \cs 0.048          & \cs 0.048          & \cs 0.048          & \cs 0.054         & \cs 0.038          & \cs 0.039          & \cs 0.037          & \cs 0.038          & \cs 0.044          & \cs 0.053          & \cs 0.057          & \cs 0.055          & \cs 0.052          & \cs 0.058          \\
 \smMALAem{}&  \cs 0.034          & \cs 0.034          & \cs 0.034          & \cs 0.042          & \cs 0.072          & \cs 0.069          & \cs 0.259          & \cs 0.427          & \cs 0.536          & \cs 0.535          & \cs 0.045          & \cs 0.049          & \cs 0.076          & \cs 0.169          & \cs 0.248          & \cs 0.033          & \cs 0.034         & \cs 0.033          & \cs 0.036          & \cs 0.039          & \cs 0.048          & \cs 0.066          & \cs 0.082          & \cs 0.114          & \cs 0.131          & \cs 0.049          & \cs 0.048          & \cs 0.048          & \cs 0.048          & \cs 0.054         & \cs 0.063          & \cs 0.051          & \cs 0.05           & \cs 0.05           & \cs 0.097          & \cs 0.061          & \cs 0.07           & \cs 0.064          & \cs 0.069          & \cs 0.091          \\
 \smMALAdm{}&  \cs 0.034          & \cs 0.037          & \cs 0.071          & \cs 0.208          & \cs 0.315          & \cs 0.036          & \cs 0.042          & \cs 0.125          & \cs 0.25           & \cs 0.333          & \cs 0.045          & \cs 0.048          & \cs 0.076          & \cs 0.163          & \cs 0.251          & \cs 0.032          & \cs 0.033         & \cs 0.034          & \cs 0.036          & \cs 0.041          & \cs \textbf{0.035} & \cs 0.045          & \cs 0.045          & \cs 0.047          & \cs 0.051          & \cs 0.049          & \cs 0.047          & \cs 0.048          & \cs 0.048          & \cs 0.054         & \cs 0.036          & \cs 0.037          & \cs 0.038          & \cs 0.04           & \cs 0.046          & \cs 0.043          & \cs 0.05           & \cs 0.05           & \cs 0.047          & \cs 0.052          \\
 \Dikinm{}  &  \cs 0.042          & \cs 0.039          & \cs 0.037          & \cs 0.036          & \cs 0.037          & \cs 0.039          & \cs \textbf{0.034} & \cs \textbf{0.034} & \cs \textbf{0.035} & \cs \textbf{0.036} & \cs 0.048          & \cs 0.044          & \cs \textbf{0.042} & \cs \textbf{0.041} & \cs \textbf{0.043} & \cs 0.039          & \cs 0.037         & \cs 0.037          & \cs 0.039          & \cs 0.04           & \cs 0.058          & \cs 0.063          & \cs 0.053          & \cs 0.05           & \cs 0.049          & \cs 0.054          & \cs 0.048          & \cs 0.047          & \cs 0.048          & \cs 0.051         & \cs 0.042          & \cs 0.039          & \cs 0.039          & \cs 0.041          & \cs 0.043          & \cs 0.05           & \cs 0.051          & \cs 0.052          & \cs 0.052          & \cs 0.054          \\
 \MAPLAm{}  &  \cs 0.042          & \cs 0.041          & \cs 0.038          & \cs 0.038          & \cs 0.039          & \cs 0.039          & \cs 0.039          & \cs 0.037          & \cs 0.037          & \cs 0.039          & \cs 0.048          & \cs 0.046          & \cs 0.044          & \cs 0.043          & \cs 0.045          & \cs 0.039          & \cs 0.037         & \cs 0.035          & \cs 0.037          & \cs 0.039          & \cs 0.077          & \cs 0.08           & \cs 0.065          & \cs 0.054          & \cs 0.051          & \cs 0.052          & \cs 0.049          & \cs 0.048          & \cs 0.048          & \cs 0.052         & \cs 0.043          & \cs 0.041          & \cs 0.039          & \cs 0.041          & \cs 0.043          & \cs 0.059          & \cs 0.056          & \cs 0.053          & \cs 0.05           & \cs 0.051          \\
 \HRm{}     &  \cs 0.038          & \cs 0.038          & \cs 0.045          & \cs 0.063          & \cs 0.114          & \cs \textbf{0.034} & \cs \textbf{0.034} & \cs 0.044          & \cs 0.062          & \cs 0.099          & \cs \textbf{0.043} & \cs \textbf{0.042} & \cs 0.045          & \cs 0.062          & \cs 0.113          & \cs 0.04           & \cs 0.04          & \cs 0.04           & \cs 0.042          & \cs 0.045          & \cs 0.042          & \cs 0.048          & \cs 0.044          & \cs \textbf{0.045} & \cs \textbf{0.047} & \cs 0.049          & \cs 0.046          & \cs \textbf{0.045} & \cs \textbf{0.046} & \cs \textbf{0.05} & \cs 0.038          & \cs 0.037          & \cs 0.037          & \cs 0.039          & \cs 0.045          & \cs 0.045          & \cs 0.051          & \cs 0.052          & \cs 0.055          & \cs 0.063          \\
 \LHRm{}    &  \cs 0.035          & \cs 0.036          & \cs 0.038          & \cs 0.054          & \cs 0.098          & \cs \textbf{0.034} & \cs 0.036          & \cs 0.044          & \cs 0.065          & \cs 0.113          & \cs \textbf{0.043} & \cs 0.043          & \cs 0.046          & \cs 0.063          & \cs 0.119          & \cs 0.04           & \cs 0.036         & \cs 0.033          & \cs 0.035          & \cs 0.038          & \cs 0.04           & \cs 0.049          & \cs 0.046          & \cs 0.046          & \cs 0.049          & \cs \textbf{0.047} & \cs \textbf{0.045} & \cs \textbf{0.045} & \cs \textbf{0.046} & \cs 0.051         & \cs 0.037          & \cs 0.037          & \cs 0.037          & \cs \textbf{0.037} & \cs \textbf{0.041} & \cs 0.051          & \cs 0.055          & \cs 0.053          & \cs 0.048          & \cs 0.053          \\
 \smHRem{}  &  \cs 0.035          & \cs 0.033          & \cs \textbf{0.032} & \cs \textbf{0.032} & \cs \textbf{0.033} & \cs 0.049          & \cs 0.214          & \cs 0.359          & \cs 0.475          & \cs 0.481          & \cs \textbf{0.043} & \cs \textbf{0.042} & \cs 0.045          & \cs 0.062          & \cs 0.112          & \cs 0.035          & \cs 0.033         & \cs 0.032          & \cs \textbf{0.034} & \cs \textbf{0.037} & \cs 0.042          & \cs 0.062          & \cs 0.072          & \cs 0.095          & \cs 0.119          & \cs 0.049          & \cs 0.046          & \cs \textbf{0.045} & \cs \textbf{0.046} & \cs \textbf{0.05} & \cs 0.047          & \cs 0.044          & \cs 0.044          & \cs 0.052          & \cs 0.08           & \cs 0.053          & \cs 0.057          & \cs 0.061          & \cs 0.066          & \cs 0.081          \\
 \smHRdm{}  &  \cs 0.035          & \cs 0.034          & \cs \textbf{0.032} & \cs 0.033          & \cs 0.035          & \cs 0.035          & \cs 0.035          & \cs 0.043          & \cs 0.064          & \cs 0.099          & \cs \textbf{0.043} & \cs \textbf{0.042} & \cs 0.045          & \cs 0.061          & \cs 0.112          & \cs 0.034          & \cs 0.033         & \cs 0.033          & \cs 0.037          & \cs 0.041          & \cs 0.038          & \cs 0.045          & \cs \textbf{0.043} & \cs \textbf{0.045} & \cs 0.048          & \cs 0.049          & \cs 0.046          & \cs \textbf{0.045} & \cs \textbf{0.046} & \cs \textbf{0.05} & \cs 0.036          & \cs \textbf{0.036} & \cs \textbf{0.036} & \cs 0.038          & \cs 0.043          & \cs 0.043          & \cs \textbf{0.046} & \cs 0.046          & \cs 0.047          & \cs 0.056          \\
 \smLHRem{} &  \cs 0.034          & \cs \textbf{0.032} & \cs \textbf{0.032} & \cs \textbf{0.032} & \cs \textbf{0.033} & \cs 0.057          & \cs 0.219          & \cs 0.362          & \cs 0.482          & \cs 0.504          & \cs \textbf{0.043} & \cs 0.044          & \cs 0.046          & \cs 0.064          & \cs 0.118          & \cs 0.032          & \cs 0.031         & \cs 0.032          & \cs \textbf{0.034} & \cs 0.038          & \cs 0.042          & \cs 0.062          & \cs 0.071          & \cs 0.1            & \cs 0.122          & \cs \textbf{0.047} & \cs 0.046          & \cs \textbf{0.045} & \cs \textbf{0.046} & \cs \textbf{0.05} & \cs 0.055          & \cs 0.046          & \cs 0.047          & \cs 0.048          & \cs 0.087          & \cs 0.055          & \cs 0.058          & \cs 0.059          & \cs 0.064          & \cs 0.084          \\
 \smLHRdm{} &  \cs \textbf{0.032} & \cs \textbf{0.032} & \cs 0.033          & \cs 0.044          & \cs 0.091          & \cs \textbf{0.034} & \cs 0.036          & \cs 0.045          & \cs 0.066          & \cs 0.111          & \cs \textbf{0.043} & \cs 0.044          & \cs 0.046          & \cs 0.065          & \cs 0.119          & \cs \textbf{0.031} & \cs \textbf{0.03} & \cs \textbf{0.031} & \cs \textbf{0.034} & \cs \textbf{0.037} & \cs \textbf{0.035} & \cs \textbf{0.044} & \cs \textbf{0.043} & \cs \textbf{0.045} & \cs 0.048          & \cs \textbf{0.047} & \cs \textbf{0.045} & \cs \textbf{0.045} & \cs \textbf{0.046} & \cs \textbf{0.05} & \cs \textbf{0.035} & \cs \textbf{0.036} & \cs 0.037          & \cs 0.039          & \cs 0.043          & \cs \textbf{0.042} & \cs \textbf{0.046} & \cs \textbf{0.044} & \cs \textbf{0.043} & \cs \textbf{0.048} \\
\hline
\end{tabular}};
    \end{tikzpicture}
    }
\end{table}

\begin{table}[]
    \centering
    \caption{
        Average L1 error of the tested samplers across polytope angles as presented as solid lines in \cref{fig:meanerr1}.
    }
    \resizebox{\linewidth}{!}{
    \begin{tikzpicture}
        \node at (0, 0) {%
            \begin{tabular}{r|
                S[table-format=1.4]S[table-format=1.4]S[table-format=1.4]S[table-format=1.4]|
                S[table-format=1.4]S[table-format=1.4]S[table-format=1.4]S[table-format=1.4]|
                S[table-format=1.4]S[table-format=1.4]S[table-format=1.4]S[table-format=1.4]|
                S[table-format=1.4]S[table-format=1.4]S[table-format=1.4]S[table-format=1.4]|
                S[table-format=1.4]S[table-format=1.4]S[table-format=1.4]S[table-format=1.4]|
                S[table-format=1.4]S[table-format=1.4]S[table-format=1.4]S[table-format=1.4]|
                S[table-format=1.4]S[table-format=1.4]S[table-format=1.4]S[table-format=1.4]|
                S[table-format=1.4]S[table-format=1.4]S[table-format=1.4]S[table-format=1.4]
            }
            \hline
            \multirow{2}{*}{Sampler ${\vcenter{\hbox{\begin{tikzpicture}[x=1mm,y=3.5mm]
                \draw[-] (1,-1) -- (-1,1);
            \end{tikzpicture}}}}$ \!\!\!\!\!\! $\begin{array}{c}\text{Target} \\\theta\end{array}$ \!\!\!\!\!\! }
                & \multicolumn{4}{c|}{Gauss $(\mu = 0)$}
                & \multicolumn{4}{c|}{Disc $(\mu = 0)$}
                & \multicolumn{4}{c|}{Cigar $(\mu = 0)$}
                & \multicolumn{4}{c|}{Gauss $(\mu = 0.5)$}
                & \multicolumn{4}{c|}{Disc $(\mu = 0.5)$}
                & \multicolumn{4}{c|}{Cigar $(\mu = 0.5)$}
                & \multicolumn{4}{c|}{Bowtie}
                & \multicolumn{4}{c}{Funnel}
            \\
            & \multicolumn{1}{c}{\footnotesize $9^{\circ}$} & \multicolumn{1}{c}{\footnotesize $19^{\circ}$} & \multicolumn{1}{c}{\footnotesize $45^{\circ}$} & \multicolumn{1}{c|}{\footnotesize $90^{\circ}$}
            & \multicolumn{1}{c}{\footnotesize $9^{\circ}$} & \multicolumn{1}{c}{\footnotesize $19^{\circ}$} & \multicolumn{1}{c}{\footnotesize $45^{\circ}$} & \multicolumn{1}{c|}{\footnotesize $90^{\circ}$}
            & \multicolumn{1}{c}{\footnotesize $9^{\circ}$} & \multicolumn{1}{c}{\footnotesize $19^{\circ}$} & \multicolumn{1}{c}{\footnotesize $45^{\circ}$} & \multicolumn{1}{c|}{\footnotesize $90^{\circ}$}
            & \multicolumn{1}{c}{\footnotesize $9^{\circ}$} & \multicolumn{1}{c}{\footnotesize $19^{\circ}$} & \multicolumn{1}{c}{\footnotesize $45^{\circ}$} & \multicolumn{1}{c|}{\footnotesize $90^{\circ}$}
            & \multicolumn{1}{c}{\footnotesize $9^{\circ}$} & \multicolumn{1}{c}{\footnotesize $19^{\circ}$} & \multicolumn{1}{c}{\footnotesize $45^{\circ}$} & \multicolumn{1}{c|}{\footnotesize $90^{\circ}$}
            & \multicolumn{1}{c}{\footnotesize $9^{\circ}$} & \multicolumn{1}{c}{\footnotesize $19^{\circ}$} & \multicolumn{1}{c}{\footnotesize $45^{\circ}$} & \multicolumn{1}{c|}{\footnotesize $90^{\circ}$}
            & \multicolumn{1}{c}{\footnotesize $9^{\circ}$} & \multicolumn{1}{c}{\footnotesize $19^{\circ}$} & \multicolumn{1}{c}{\footnotesize $45^{\circ}$} & \multicolumn{1}{c|}{\footnotesize $90^{\circ}$}
            & \multicolumn{1}{c}{\footnotesize $9^{\circ}$} & \multicolumn{1}{c}{\footnotesize $19^{\circ}$} & \multicolumn{1}{c}{\footnotesize $45^{\circ}$} & \multicolumn{1}{c }{\footnotesize $90^{\circ}$}
            \\
            \hline
  \RWMHm{}   &  \cs 0.401          & \cs 0.095          & \cs 0.047          & \cs 0.046          & \cs 0.404          & \cs 0.141          & \cs 0.042          & \cs 0.04           & \cs 0.314          & \cs 0.047          & \cs 0.05           & \cs 0.054          & \cs 0.044          & \cs 0.039         & \cs 0.044          & \cs 0.047          & \cs 0.047          & \cs 0.045          & \cs 0.05           & \cs 0.05           & \cs 0.045          & \cs 0.045          & \cs 0.053          & \cs 0.056          & \cs 0.042          & \cs 0.039          & \cs 0.042         & \cs 0.044         & \cs 0.058          & \cs 0.052          & \cs 0.057          & \cs 0.061         \\
  \MALAm{}   &  \cs 0.399          & \cs 0.094          & \cs 0.044          & \cs 0.043          & \cs 0.407          & \cs 0.14           & \cs 0.041          & \cs 0.039          & \cs 0.316          & \cs 0.045          & \cs 0.05           & \cs 0.053          & \cs 0.038          & \cs 0.035         & \cs 0.04           & \cs 0.042          & \cs 0.046          & \cs 0.045          & \cs 0.05           & \cs 0.05           & \cs 0.045          & \cs 0.045          & \cs 0.052          & \cs 0.055          & \cs 0.04           & \cs 0.037          & \cs \textbf{0.04} & \cs 0.041         & \cs 0.053          & \cs 0.05           & \cs 0.056          & \cs 0.061         \\
  \smMALAem{}&  \cs 0.059          & \cs 0.029          & \cs 0.038          & \cs 0.048          & \cs 0.696          & \cs 0.509          & \cs 0.2            & \cs 0.056          & \cs 0.321          & \cs 0.046          & \cs 0.049          & \cs 0.053          & \cs 0.026          & \cs 0.031         & \cs 0.038          & \cs 0.045          & \cs 0.169          & \cs 0.085          & \cs 0.053          & \cs 0.046          & \cs 0.045          & \cs 0.045          & \cs 0.052          & \cs 0.055          & \cs 0.069          & \cs 0.055          & \cs 0.06          & \cs 0.064         & \cs 0.067          & \cs 0.062          & \cs 0.075          & \cs 0.081         \\
  \smMALAdm{}&  \cs 0.371          & \cs 0.08           & \cs 0.04           & \cs \textbf{0.041} & \cs 0.409          & \cs 0.14           & \cs 0.041          & \cs 0.039          & \cs 0.319          & \cs 0.045          & \cs 0.049          & \cs 0.053          & \cs 0.035          & \cs 0.032         & \cs \textbf{0.036} & \cs 0.039          & \cs 0.044          & \cs 0.042          & \cs \textbf{0.046} & \cs 0.045          & \cs 0.045          & \cs 0.045          & \cs 0.052          & \cs 0.055          & \cs 0.04           & \cs 0.037          & \cs 0.041         & \cs 0.041         & \cs 0.047          & \cs 0.044          & \cs 0.05           & \cs 0.053         \\
  \Dikinm{}  &  \cs 0.027          & \cs 0.033          & \cs 0.043          & \cs 0.05           & \cs \textbf{0.027} & \cs \textbf{0.033} & \cs 0.041          & \cs 0.041          & \cs \textbf{0.029} & \cs \textbf{0.037} & \cs 0.052          & \cs 0.056          & \cs 0.029          & \cs 0.034         & \cs 0.041          & \cs 0.049          & \cs 0.056          & \cs 0.055          & \cs 0.055          & \cs 0.051          & \cs 0.04           & \cs 0.046          & \cs 0.055          & \cs 0.057          & \cs 0.036          & \cs 0.039          & \cs 0.043         & \cs 0.045         & \cs \textbf{0.039} & \cs 0.048          & \cs 0.057          & \cs 0.064         \\
  \MAPLAm{}  &  \cs 0.029          & \cs 0.035          & \cs 0.043          & \cs 0.051          & \cs 0.03           & \cs 0.036          & \cs 0.044          & \cs 0.044          & \cs 0.032          & \cs 0.038          & \cs 0.054          & \cs 0.057          & \cs 0.029          & \cs 0.034         & \cs 0.04           & \cs 0.046          & \cs 0.081          & \cs 0.069          & \cs 0.061          & \cs 0.052          & \cs 0.04           & \cs 0.047          & \cs 0.055          & \cs 0.057          & \cs 0.039          & \cs 0.04           & \cs 0.042         & \cs 0.044         & \cs 0.04           & \cs 0.049          & \cs 0.06           & \cs 0.066         \\
  \HRm{}     &  \cs 0.114          & \cs 0.04           & \cs 0.039          & \cs 0.044          & \cs 0.105          & \cs 0.04           & \cs \textbf{0.037} & \cs 0.037          & \cs 0.107          & \cs \textbf{0.037} & \cs \textbf{0.048} & \cs \textbf{0.051} & \cs 0.036          & \cs 0.038         & \cs 0.044          & \cs 0.047          & \cs \textbf{0.041} & \cs 0.043          & \cs 0.049          & \cs 0.048          & \cs \textbf{0.039} & \cs \textbf{0.043} & \cs 0.052          & \cs 0.054          & \cs 0.035          & \cs 0.037          & \cs 0.042         & \cs 0.043         & \cs 0.048          & \cs 0.049          & \cs 0.056          & \cs 0.059         \\
  \LHRm{}    &  \cs 0.097          & \cs 0.033          & \cs 0.037          & \cs 0.042          & \cs 0.116          & \cs 0.042          & \cs 0.038          & \cs 0.037          & \cs 0.113          & \cs 0.039          & \cs \textbf{0.048} & \cs 0.052          & \cs 0.031          & \cs 0.033         & \cs 0.039          & \cs 0.042          & \cs 0.042          & \cs 0.044          & \cs 0.049          & \cs 0.049          & \cs \textbf{0.039} & \cs \textbf{0.043} & \cs 0.052          & \cs \textbf{0.053} & \cs 0.035          & \cs \textbf{0.036} & \cs \textbf{0.04} & \cs \textbf{0.04} & \cs 0.045          & \cs 0.047          & \cs 0.056          & \cs 0.06          \\
  \smHRem{}  &  \cs \textbf{0.022} & \cs \textbf{0.027} & \cs 0.037          & \cs 0.045          & \cs 0.66           & \cs 0.441          & \cs 0.12           & \cs 0.042          & \cs 0.107          & \cs \textbf{0.037} & \cs \textbf{0.048} & \cs \textbf{0.051} & \cs 0.026          & \cs 0.031         & \cs 0.038          & \cs 0.043          & \cs 0.145          & \cs 0.073          & \cs 0.049          & \cs 0.045          & \cs \textbf{0.039} & \cs \textbf{0.043} & \cs 0.052          & \cs 0.054          & \cs 0.053          & \cs 0.048          & \cs 0.054         & \cs 0.059         & \cs 0.051          & \cs 0.057          & \cs 0.07           & \cs 0.076         \\
  \smHRdm{}  &  \cs 0.023          & \cs 0.031          & \cs 0.038          & \cs 0.043          & \cs 0.106          & \cs 0.04           & \cs 0.038          & \cs \textbf{0.036} & \cs 0.107          & \cs \textbf{0.037} & \cs \textbf{0.048} & \cs \textbf{0.051} & \cs 0.028          & \cs 0.033         & \cs 0.039          & \cs 0.042          & \cs \textbf{0.041} & \cs 0.042          & \cs 0.047          & \cs 0.045          & \cs \textbf{0.039} & \cs 0.044          & \cs 0.052          & \cs 0.054          & \cs \textbf{0.034} & \cs \textbf{0.036} & \cs \textbf{0.04} & \cs 0.041         & \cs 0.042          & \cs 0.044          & \cs 0.051          & \cs 0.053         \\
  \smLHRem{} &  \cs \textbf{0.022} & \cs \textbf{0.027} & \cs 0.037          & \cs 0.045          & \cs 0.665          & \cs 0.453          & \cs 0.136          & \cs 0.045          & \cs 0.113          & \cs 0.039          & \cs \textbf{0.048} & \cs 0.052          & \cs \textbf{0.025} & \cs \textbf{0.03} & \cs 0.037          & \cs 0.042          & \cs 0.149          & \cs 0.073          & \cs 0.05           & \cs 0.045          & \cs \textbf{0.039} & \cs \textbf{0.043} & \cs 0.052          & \cs \textbf{0.053} & \cs 0.058          & \cs 0.05           & \cs 0.057         & \cs 0.061         & \cs 0.053          & \cs 0.057          & \cs 0.069          & \cs 0.077         \\
  \smLHRdm{} &  \cs 0.081          & \cs 0.029          & \cs \textbf{0.036} & \cs \textbf{0.041} & \cs 0.115          & \cs 0.042          & \cs 0.038          & \cs 0.037          & \cs 0.114          & \cs 0.039          & \cs \textbf{0.048} & \cs 0.052          & \cs 0.027          & \cs \textbf{0.03} & \cs \textbf{0.036} & \cs \textbf{0.038} & \cs \textbf{0.041} & \cs \textbf{0.041} & \cs \textbf{0.046} & \cs \textbf{0.044} & \cs \textbf{0.039} & \cs \textbf{0.043} & \cs \textbf{0.051} & \cs \textbf{0.053} & \cs \textbf{0.034} & \cs \textbf{0.036} & \cs \textbf{0.04} & \cs 0.041         & \cs \textbf{0.039} & \cs \textbf{0.041} & \cs \textbf{0.048} & \cs \textbf{0.05} \\
 \hline
\end{tabular}};
    \end{tikzpicture}
    }
\end{table}

\begin{table}[]
    \centering
    \caption{
        Average L1 error of the tested samplers across the densities' scale parameters
        as presented as solid lines in \cref{fig:meanerr1}.
    }
    \resizebox{.9\linewidth}{!}{
    \begin{tikzpicture}
        \node at (0, 0) {%
            \begin{tabular}{r|
                S[table-format=1.4]S[table-format=1.4]S[table-format=1.4]S[table-format=1.4]S[table-format=1.4]S[table-format=1.4]S[table-format=1.4]|
                S[table-format=1.4]S[table-format=1.4]S[table-format=1.4]S[table-format=1.4]S[table-format=1.4]S[table-format=1.4]S[table-format=1.4]|
                S[table-format=1.4]S[table-format=1.4]S[table-format=1.4]S[table-format=1.4]S[table-format=1.4]S[table-format=1.4]S[table-format=1.4]|
                S[table-format=1.4]S[table-format=1.4]S[table-format=1.4]S[table-format=1.4]S[table-format=1.4]S[table-format=1.4]S[table-format=1.4]
            }
            \hline
            \multirow{2}{*}{Sampler ${\vcenter{\hbox{\begin{tikzpicture}[x=1mm,y=3.5mm]
                \draw[-] (1,-1) -- (-1,1);
            \end{tikzpicture}}}}$ \!\!\!\!\!\! $\begin{array}{c}\text{Target} \\\sigma\end{array}$ \!\!\!\!\!\! }
                & \multicolumn{7}{c|}{Gauss $(\mu = 0)$}
                & \multicolumn{7}{c|}{Disc $(\mu = 0)$}
                & \multicolumn{7}{c|}{Cigar $(\mu = 0)$}
                & \multicolumn{7}{c }{Gauss $(\mu = 0.5)$}
            \\
            & \multicolumn{1}{c}{\footnotesize $10^{-2}$} & \multicolumn{1}{c}{\footnotesize $10^{-1.5}$} & \multicolumn{1}{c}{\footnotesize $10^{-1}$} & \multicolumn{1}{c}{\footnotesize $10^{-0.5}$} & \multicolumn{1}{c}{\footnotesize $10^{0}$} & \multicolumn{1}{c}{\footnotesize $10^{0.5}$} & \multicolumn{1}{c|}{\footnotesize $10^{1}$}
            & \multicolumn{1}{c}{\footnotesize $10^{-2}$} & \multicolumn{1}{c}{\footnotesize $10^{-1.5}$} & \multicolumn{1}{c}{\footnotesize $10^{-1}$} & \multicolumn{1}{c}{\footnotesize $10^{-0.5}$} & \multicolumn{1}{c}{\footnotesize $10^{0}$} & \multicolumn{1}{c}{\footnotesize $10^{0.5}$} & \multicolumn{1}{c|}{\footnotesize $10^{1}$}
            & \multicolumn{1}{c}{\footnotesize $10^{-2}$} & \multicolumn{1}{c}{\footnotesize $10^{-1.5}$} & \multicolumn{1}{c}{\footnotesize $10^{-1}$} & \multicolumn{1}{c}{\footnotesize $10^{-0.5}$} & \multicolumn{1}{c}{\footnotesize $10^{0}$} & \multicolumn{1}{c}{\footnotesize $10^{0.5}$} & \multicolumn{1}{c|}{\footnotesize $10^{1}$}
            & \multicolumn{1}{c}{\footnotesize $10^{-2}$} & \multicolumn{1}{c}{\footnotesize $10^{-1.5}$} & \multicolumn{1}{c}{\footnotesize $10^{-1}$} & \multicolumn{1}{c}{\footnotesize $10^{-0.5}$} & \multicolumn{1}{c}{\footnotesize $10^{0}$} & \multicolumn{1}{c}{\footnotesize $10^{0.5}$} & \multicolumn{1}{c }{\footnotesize $10^{1}$}
            \\
            \hline
  \RWMHm{}   &  \cs 0.116          & \cs 0.144          & \cs 0.188          & \cs 0.171         & \cs 0.15           & \cs 0.135          & \cs 0.124          & \cs 0.136          & \cs 0.098          & \cs 0.092          & \cs 0.114          & \cs 0.207          & \cs 0.225          & \cs 0.225          & \cs 0.077          & \cs 0.084          & \cs 0.113          & \cs 0.131          & \cs 0.137          & \cs 0.136          & \cs 0.136          & \cs 0.038          & \cs 0.039          & \cs 0.041          & \cs 0.045          & \cs 0.046          & \cs 0.048          & \cs 0.048           \\
  \MALAm{}   &  \cs 0.111          & \cs 0.14           & \cs 0.187          & \cs 0.169         & \cs 0.148          & \cs 0.134          & \cs 0.125          & \cs 0.134          & \cs 0.099          & \cs 0.099          & \cs 0.112          & \cs 0.205          & \cs 0.224          & \cs 0.225          & \cs 0.078          & \cs 0.085          & \cs 0.113          & \cs 0.129          & \cs 0.136          & \cs 0.136          & \cs 0.136          & \cs 0.028          & \cs 0.03           & \cs 0.034          & \cs 0.038          & \cs 0.045          & \cs 0.047          & \cs 0.048           \\
  \smMALAem{}&  \cs \textbf{0.026} & \cs 0.028          & \cs 0.031          & \cs 0.041         & \cs 0.054          & \cs 0.061          & \cs 0.061          & \cs 0.16           & \cs 0.33           & \cs 0.413          & \cs 0.41           & \cs 0.437          & \cs 0.421          & \cs 0.385          & \cs 0.065          & \cs 0.082          & \cs 0.116          & \cs 0.131          & \cs 0.137          & \cs 0.141          & \cs 0.148          & \cs 0.023          & \cs \textbf{0.023} & \cs 0.027          & \cs 0.033          & \cs 0.042          & \cs 0.048          & \cs 0.05            \\
  \smMALAdm{}&  \cs 0.105          & \cs 0.118          & \cs 0.152          & \cs 0.15          & \cs 0.148          & \cs 0.135          & \cs 0.124          & \cs 0.132          & \cs 0.097          & \cs 0.1            & \cs 0.115          & \cs 0.208          & \cs 0.225          & \cs 0.225          & \cs 0.078          & \cs 0.086          & \cs 0.115          & \cs 0.129          & \cs 0.136          & \cs 0.136          & \cs 0.136          & \cs \textbf{0.022} & \cs \textbf{0.023} & \cs 0.029          & \cs 0.035          & \cs 0.044          & \cs 0.047          & \cs 0.048           \\
  \Dikinm{}  &  \cs 0.034          & \cs 0.033          & \cs 0.035          & \cs 0.037         & \cs 0.041          & \cs 0.044          & \cs 0.044          & \cs \textbf{0.033} & \cs \textbf{0.033} & \cs \textbf{0.032} & \cs \textbf{0.032} & \cs \textbf{0.034} & \cs \textbf{0.041} & \cs \textbf{0.044} & \cs \textbf{0.039} & \cs \textbf{0.039} & \cs \textbf{0.039} & \cs \textbf{0.044} & \cs \textbf{0.048} & \cs \textbf{0.048} & \cs \textbf{0.048} & \cs 0.032          & \cs 0.028          & \cs 0.032          & \cs 0.039          & \cs 0.044          & \cs 0.046          & \cs 0.047           \\
  \MAPLAm{}  &  \cs 0.035          & \cs 0.036          & \cs 0.037          & \cs 0.038         & \cs 0.042          & \cs 0.045          & \cs 0.045          & \cs 0.036          & \cs 0.036          & \cs 0.035          & \cs 0.035          & \cs 0.036          & \cs 0.043          & \cs 0.046          & \cs 0.041          & \cs 0.041          & \cs 0.041          & \cs 0.045          & \cs 0.049          & \cs 0.05           & \cs 0.049          & \cs 0.026          & \cs 0.027          & \cs 0.032          & \cs 0.038          & \cs 0.044          & \cs 0.047          & \cs 0.048           \\
  \HRm{}     &  \cs 0.047          & \cs 0.056          & \cs 0.059          & \cs 0.062         & \cs 0.064          & \cs 0.065          & \cs 0.063          & \cs 0.035          & \cs 0.038          & \cs 0.045          & \cs 0.051          & \cs 0.063          & \cs 0.073          & \cs 0.078          & \cs 0.048          & \cs 0.052          & \cs 0.057          & \cs 0.064          & \cs 0.069          & \cs 0.068          & \cs 0.068          & \cs 0.039          & \cs 0.039          & \cs 0.04           & \cs 0.041          & \cs 0.042          & \cs 0.043          & \cs 0.044           \\
  \LHRm{}    &  \cs 0.038          & \cs 0.044          & \cs 0.047          & \cs 0.052         & \cs 0.059          & \cs 0.062          & \cs 0.062          & \cs 0.045          & \cs 0.044          & \cs 0.049          & \cs 0.053          & \cs 0.065          & \cs 0.073          & \cs 0.078          & \cs 0.052          & \cs 0.055          & \cs 0.06           & \cs 0.066          & \cs 0.07           & \cs 0.069          & \cs 0.069          & \cs 0.029          & \cs 0.03           & \cs 0.033          & \cs 0.035          & \cs 0.04           & \cs 0.043          & \cs 0.043           \\
  \smHRem{}  &  \cs 0.027          & \cs \textbf{0.026} & \cs \textbf{0.027} & \cs \textbf{0.03} & \cs \textbf{0.037} & \cs 0.041          & \cs 0.041          & \cs 0.13           & \cs 0.285          & \cs 0.336          & \cs 0.351          & \cs 0.38           & \cs 0.377          & \cs 0.35           & \cs 0.047          & \cs 0.052          & \cs 0.057          & \cs 0.064          & \cs 0.069          & \cs 0.068          & \cs 0.068          & \cs 0.025          & \cs 0.025          & \cs 0.028          & \cs 0.033          & \cs 0.04           & \cs 0.044          & \cs 0.044           \\
  \smHRdm{}  &  \cs 0.03           & \cs 0.03           & \cs 0.03           & \cs 0.032         & \cs \textbf{0.037} & \cs \textbf{0.039} & \cs \textbf{0.039} & \cs 0.036          & \cs 0.04           & \cs 0.045          & \cs 0.052          & \cs 0.062          & \cs 0.074          & \cs 0.077          & \cs 0.048          & \cs 0.052          & \cs 0.057          & \cs 0.064          & \cs 0.068          & \cs 0.069          & \cs 0.068          & \cs 0.029          & \cs 0.029          & \cs 0.031          & \cs 0.036          & \cs 0.04           & \cs \textbf{0.042} & \cs \textbf{0.042}  \\
  \smLHRem{} &  \cs \textbf{0.026} & \cs \textbf{0.026} & \cs \textbf{0.027} & \cs \textbf{0.03} & \cs \textbf{0.037} & \cs 0.041          & \cs 0.042          & \cs 0.138          & \cs 0.296          & \cs 0.35           & \cs 0.36           & \cs 0.39           & \cs 0.385          & \cs 0.354          & \cs 0.052          & \cs 0.056          & \cs 0.06           & \cs 0.066          & \cs 0.07           & \cs 0.069          & \cs 0.069          & \cs 0.023          & \cs \textbf{0.023} & \cs \textbf{0.026} & \cs \textbf{0.032} & \cs 0.04           & \cs 0.044          & \cs 0.045           \\
  \smLHRdm{} &  \cs 0.032          & \cs 0.036          & \cs 0.039          & \cs 0.045         & \cs 0.054          & \cs 0.059          & \cs 0.06           & \cs 0.045          & \cs 0.043          & \cs 0.048          & \cs 0.054          & \cs 0.066          & \cs 0.074          & \cs 0.077          & \cs 0.051          & \cs 0.055          & \cs 0.062          & \cs 0.066          & \cs 0.07           & \cs 0.069          & \cs 0.069          & \cs \textbf{0.022} & \cs \textbf{0.023} & \cs 0.027          & \cs \textbf{0.032} & \cs \textbf{0.039} & \cs \textbf{0.042} & \cs 0.043           \\
            \hline
\end{tabular}};
    \end{tikzpicture}
    }
\end{table}

\begin{table}[]
    \centering
    \caption{
        Average L1 error of the tested samplers across the densities' scale parameters
        as presented as solid lines in \cref{fig:meanerr1}.
    }
    \resizebox{.9\linewidth}{!}{
    \begin{tikzpicture}
        \node at (0, 0) {%
            \begin{tabular}{r|
                S[table-format=1.4]S[table-format=1.4]S[table-format=1.4]S[table-format=1.4]S[table-format=1.4]S[table-format=1.4]S[table-format=1.4]|
                S[table-format=1.4]S[table-format=1.4]S[table-format=1.4]S[table-format=1.4]S[table-format=1.4]S[table-format=1.4]S[table-format=1.4]|
                S[table-format=1.4]S[table-format=1.4]S[table-format=1.4]S[table-format=1.4]S[table-format=1.4]S[table-format=1.4]S[table-format=1.4]|
                S[table-format=1.4]S[table-format=1.4]S[table-format=1.4]S[table-format=1.4]S[table-format=1.4]S[table-format=1.4]S[table-format=1.4]
            }
            \hline
            \multirow{2}{*}{Sampler ${\vcenter{\hbox{\begin{tikzpicture}[x=1mm,y=3.5mm]
                \draw[-] (1,-1) -- (-1,1);
            \end{tikzpicture}}}}$ \!\!\!\!\!\! $\begin{array}{c}\text{Target} \\\sigma\end{array}$ \!\!\!\!\!\! }
                & \multicolumn{7}{c|}{Disc $(\mu = 0.5)$}
                & \multicolumn{7}{c|}{Cigar $(\mu = 0.5)$}
                & \multicolumn{7}{c|}{Bowtie}
                & \multicolumn{7}{c}{Funnel}
            \\
            & \multicolumn{1}{c}{\footnotesize $10^{-2}$} & \multicolumn{1}{c}{\footnotesize $10^{-1.5}$} & \multicolumn{1}{c}{\footnotesize $10^{-1}$} & \multicolumn{1}{c}{\footnotesize $10^{-0.5}$} & \multicolumn{1}{c}{\footnotesize $10^{0}$} & \multicolumn{1}{c}{\footnotesize $10^{0.5}$} & \multicolumn{1}{c|}{\footnotesize $10^{1}$}
            & \multicolumn{1}{c}{\footnotesize $10^{-2}$} & \multicolumn{1}{c}{\footnotesize $10^{-1.5}$} & \multicolumn{1}{c}{\footnotesize $10^{-1}$} & \multicolumn{1}{c}{\footnotesize $10^{-0.5}$} & \multicolumn{1}{c}{\footnotesize $10^{0}$} & \multicolumn{1}{c}{\footnotesize $10^{0.5}$} & \multicolumn{1}{c|}{\footnotesize $10^{1}$}
            & \multicolumn{1}{c}{\footnotesize $10^{-2}$} & \multicolumn{1}{c}{\footnotesize $10^{-1.5}$} & \multicolumn{1}{c}{\footnotesize $10^{-1}$} & \multicolumn{1}{c}{\footnotesize $10^{-0.5}$} & \multicolumn{1}{c}{\footnotesize $10^{0}$} & \multicolumn{1}{c}{\footnotesize $10^{0.5}$} & \multicolumn{1}{c|}{\footnotesize $10^{1}$}
            & \multicolumn{1}{c}{\footnotesize $10^{-2}$} & \multicolumn{1}{c}{\footnotesize $10^{-1.5}$} & \multicolumn{1}{c}{\footnotesize $10^{-1}$} & \multicolumn{1}{c}{\footnotesize $10^{-0.5}$} & \multicolumn{1}{c}{\footnotesize $10^{0}$} & \multicolumn{1}{c}{\footnotesize $10^{0.5}$} & \multicolumn{1}{c }{\footnotesize $10^{1}$}
            \\
            \hline
  \RWMHm{}    & \cs 0.042          & \cs 0.047         & \cs 0.049          & \cs 0.051          & \cs 0.048          & \cs 0.048          & \cs 0.049          & \cs 0.045          & \cs 0.047          & \cs 0.048          & \cs 0.051          & \cs 0.052          & \cs 0.053          & \cs 0.052          & \cs 0.032          & \cs 0.034          & \cs 0.036          & \cs 0.044         & \cs 0.047          & \cs 0.049          & \cs 0.049          & \cs 0.07           & \cs 0.062          & \cs 0.06           & \cs 0.061          & \cs 0.052          & \cs 0.048          & \cs 0.048          \\
  \MALAm{}    & \cs 0.045          & \cs 0.047         & \cs 0.048          & \cs 0.05           & \cs 0.048          & \cs 0.048          & \cs 0.049          & \cs \textbf{0.044} & \cs 0.047          & \cs 0.047          & \cs 0.051          & \cs 0.052          & \cs 0.053          & \cs 0.052          & \cs \textbf{0.026} & \cs \textbf{0.031} & \cs \textbf{0.034} & \cs 0.042         & \cs 0.046          & \cs 0.048          & \cs 0.049          & \cs 0.069          & \cs 0.061          & \cs 0.055          & \cs 0.052          & \cs 0.052          & \cs 0.048          & \cs 0.048          \\
  \smMALAem{} & \cs 0.036          & \cs 0.041         & \cs 0.054          & \cs 0.091          & \cs 0.139          & \cs 0.136          & \cs 0.121          & \cs \textbf{0.044} & \cs 0.047          & \cs 0.047          & \cs 0.051          & \cs 0.052          & \cs 0.053          & \cs 0.053          & \cs 0.075          & \cs 0.075          & \cs 0.06           & \cs 0.059         & \cs 0.061          & \cs 0.057          & \cs 0.049          & \cs 0.123          & \cs 0.094          & \cs 0.069          & \cs 0.065          & \cs 0.051          & \cs 0.048          & \cs 0.048          \\
  \smMALAdm{} & \cs \textbf{0.035} & \cs \textbf{0.04} & \cs \textbf{0.044} & \cs 0.048          & \cs 0.048          & \cs 0.048          & \cs 0.049          & \cs \textbf{0.044} & \cs 0.046          & \cs 0.047          & \cs 0.051          & \cs 0.052          & \cs 0.052          & \cs 0.053          & \cs 0.028          & \cs \textbf{0.031} & \cs \textbf{0.034} & \cs 0.042         & \cs 0.046          & \cs 0.048          & \cs 0.049          & \cs 0.048          & \cs 0.047          & \cs 0.047          & \cs 0.05           & \cs 0.051          & \cs 0.048          & \cs 0.048          \\
  \Dikinm{}   & \cs 0.062          & \cs 0.064         & \cs 0.059          & \cs 0.055          & \cs 0.048          & \cs 0.047          & \cs 0.046          & \cs 0.047          & \cs 0.048          & \cs 0.047          & \cs 0.051          & \cs 0.051          & \cs 0.051          & \cs 0.051          & \cs 0.032          & \cs 0.036          & \cs 0.037          & \cs 0.042         & \cs 0.045          & \cs 0.047          & \cs 0.047          & \cs 0.068          & \cs 0.058          & \cs 0.051          & \cs 0.045          & \cs 0.047          & \cs 0.047          & \cs 0.046          \\
  \MAPLAm{}   & \cs 0.101          & \cs 0.091         & \cs 0.068          & \cs 0.055          & \cs 0.049          & \cs 0.048          & \cs 0.048          & \cs 0.046          & \cs 0.047          & \cs 0.047          & \cs 0.052          & \cs 0.052          & \cs 0.053          & \cs 0.053          & \cs 0.036          & \cs 0.036          & \cs 0.036          & \cs 0.041         & \cs 0.045          & \cs 0.048          & \cs 0.048          & \cs 0.075          & \cs 0.061          & \cs 0.051          & \cs 0.045          & \cs 0.048          & \cs 0.048          & \cs 0.048          \\
  \HRm{}      & \cs 0.042          & \cs 0.046         & \cs 0.048          & \cs 0.048          & \cs \textbf{0.045} & \cs \textbf{0.044} & \cs \textbf{0.044} & \cs 0.045          & \cs 0.046          & \cs \textbf{0.045} & \cs \textbf{0.048} & \cs \textbf{0.048} & \cs \textbf{0.048} & \cs \textbf{0.048} & \cs 0.031          & \cs 0.034          & \cs 0.035          & \cs 0.041         & \cs 0.043          & \cs \textbf{0.044} & \cs \textbf{0.044} & \cs 0.073          & \cs 0.06           & \cs 0.057          & \cs 0.05           & \cs \textbf{0.045} & \cs 0.044          & \cs 0.044          \\
  \LHRm{}     & \cs 0.044          & \cs 0.046         & \cs 0.047          & \cs 0.048          & \cs \textbf{0.045} & \cs 0.045          & \cs 0.045          & \cs \textbf{0.044} & \cs \textbf{0.045} & \cs \textbf{0.045} & \cs 0.049          & \cs 0.049          & \cs 0.049          & \cs \textbf{0.048} & \cs 0.028          & \cs \textbf{0.031} & \cs \textbf{0.034} & \cs \textbf{0.04} & \cs 0.043          & \cs \textbf{0.044} & \cs \textbf{0.044} & \cs 0.07           & \cs 0.059          & \cs 0.053          & \cs 0.045          & \cs 0.047          & \cs 0.045          & \cs 0.045          \\
  \smHRem{}   & \cs 0.037          & \cs 0.041         & \cs 0.052          & \cs 0.079          & \cs 0.12           & \cs 0.115          & \cs 0.102          & \cs 0.045          & \cs 0.046          & \cs \textbf{0.045} & \cs \textbf{0.048} & \cs \textbf{0.048} & \cs \textbf{0.048} & \cs \textbf{0.048} & \cs 0.062          & \cs 0.062          & \cs 0.051          & \cs 0.052         & \cs 0.052          & \cs 0.05           & \cs 0.045          & \cs 0.112          & \cs 0.082          & \cs 0.066          & \cs 0.052          & \cs \textbf{0.045} & \cs 0.044          & \cs 0.044          \\
  \smHRdm{}   & \cs 0.037          & \cs 0.041         & \cs 0.045          & \cs 0.048          & \cs \textbf{0.045} & \cs \textbf{0.044} & \cs \textbf{0.044} & \cs 0.045          & \cs 0.047          & \cs \textbf{0.045} & \cs \textbf{0.048} & \cs \textbf{0.048} & \cs \textbf{0.048} & \cs \textbf{0.048} & \cs 0.027          & \cs 0.032          & \cs 0.035          & \cs 0.041         & \cs 0.043          & \cs \textbf{0.044} & \cs \textbf{0.044} & \cs 0.053          & \cs 0.051          & \cs 0.05           & \cs 0.048          & \cs \textbf{0.045} & \cs \textbf{0.043} & \cs \textbf{0.043} \\
  \smLHRem{}  & \cs 0.036          & \cs \textbf{0.04} & \cs 0.051          & \cs 0.083          & \cs 0.124          & \cs 0.117          & \cs 0.103          & \cs \textbf{0.044} & \cs \textbf{0.045} & \cs \textbf{0.045} & \cs \textbf{0.048} & \cs 0.049          & \cs 0.049          & \cs \textbf{0.048} & \cs 0.072          & \cs 0.065          & \cs 0.054          & \cs 0.055         & \cs 0.054          & \cs 0.05           & \cs 0.046          & \cs 0.115          & \cs 0.082          & \cs 0.061          & \cs 0.055          & \cs 0.047          & \cs 0.044          & \cs 0.045          \\
  \smLHRdm{}  & \cs \textbf{0.035} & \cs \textbf{0.04} & \cs \textbf{0.044} & \cs \textbf{0.047} & \cs \textbf{0.045} & \cs 0.045          & \cs 0.045          & \cs \textbf{0.044} & \cs \textbf{0.045} & \cs \textbf{0.045} & \cs \textbf{0.048} & \cs 0.049          & \cs 0.049          & \cs 0.049          & \cs 0.029          & \cs 0.032          & \cs \textbf{0.034} & \cs \textbf{0.04} & \cs \textbf{0.042} & \cs \textbf{0.044} & \cs \textbf{0.044} & \cs \textbf{0.044} & \cs \textbf{0.043} & \cs \textbf{0.044} & \cs \textbf{0.044} & \cs 0.048          & \cs 0.045          & \cs 0.045          \\
            \hline
\end{tabular}};
    \end{tikzpicture}
    }
\end{table}

\end{landscape}

\end{document}